\pdfoutput=1
\documentclass[10pt,a4paper]{article}
\usepackage[T1]{fontenc}
\usepackage[utf8]{inputenc}
\usepackage{amsmath,amssymb,amsthm,mathtools,booktabs,enumitem}
\usepackage[margin=24mm]{geometry}
\usepackage[colorlinks=true,linkcolor=blue,citecolor=blue,urlcolor=blue]{hyperref}
\hypersetup{pdftitle={Simultaneous Reconstruction of Shuffled Histograms from a Retained Statistic: Exact leading coefficients for finite hypothesis families},pdfauthor={Alex Shvets},pdfsubject={Le Cam deficiency and reconstruction of shuffled histograms}}
\setlist{nosep}
\newtheorem{theorem}{Theorem}[section]
\newtheorem{lemma}[theorem]{Lemma}
\newtheorem{proposition}[theorem]{Proposition}
\newtheorem{corollary}[theorem]{Corollary}
\theoremstyle{remark}
\newtheorem{remark}[theorem]{Remark}
\newtheorem{example}[theorem]{Example}
\newtheorem{definition}[theorem]{Definition}
\newtheorem{conjecture}[theorem]{Conjecture}
\newcommand{\TV}{d_{\mathrm{TV}}}
\newcommand{\E}{\mathbb E}
\newcommand{\R}{\mathbb R}
\newcommand{\N}{\mathbb N}
\newcommand{\1}{\mathbf 1}
\newcommand{\Law}{\mathcal L}
\renewcommand{\Theta}{\varTheta}
\DeclareMathOperator{\Mult}{Mult}
\DeclareMathOperator{\Bin}{Bin}
\DeclareMathOperator{\ext}{ext}

\DeclareMathOperator{\Var}{Var}
\title{Simultaneous Reconstruction of Shuffled Histograms\\from a Retained Statistic:\\
\large Exact leading coefficients for finite hypothesis families}
\author{Alex Shvets}
\date{15 October 2026}
\begin{document}
\maketitle

\begin{abstract}
Le Cam deficiency is the largest increase in Bayes risk caused by retaining a statistic, over all priors and decision problems with losses in $[0,1]$. We compute its exact leading term for finite hypothesis families in growing shuffle models. For a bounded correlated message bundle in an iid categorical background with an observed finite tag, the coefficient is a convex program over exposed faces with Gaussian-polynomial scores; its order is set by the first differing conditional factorial-moment degree. Tag-preserving kernels can be suboptimal, and an exact example has family coefficient $4/3$ times the largest pairwise coefficient. For a critical rare-event channel along $m+1$ consecutive compositions, retaining a distinguished-symbol count has loss $n^{-1/2}\kappa\varphi(0)\mathcal V_m(\mu)+o(n^{-1/2})$, where $\mu=\theta\lambda$, $\kappa^2=\chi^2(q\|p)/(1-\theta)>0$, and $\varphi$ is the standard normal density. We evaluate $\mathcal V_2$ for all $\mu>0$. For every $m\ge2$, the full-chain value $e^{-\mu}(1-1/\sum_{k=0}^m\mu^k)$ is exact if and only if $\sum_{k=2}^{m+1}\mu^k\le2$; an explicit large-mean regime is attained by the endpoint pair. Exact-arithmetic computer-assisted certificates give the complete four-composition phase diagram. At $\mu=\sqrt2$, the ratio to the largest pairwise coefficient is strictly below two for every finite window and tends to two as $m\to\infty$. A universal upper bound of two and Hypothesis P, which asserts that admissible vertex paths attain every window coefficient, remain open outside the proved regimes. The comparison principle is classical; the contributions are the model-specific scores, exact coefficients, and certificates.
\end{abstract}

\section{Introduction}\label{sec:intro}
\paragraph{The question.}
A shuffle mechanism publishes the histogram of $n$ anonymous messages, and an analyst retains only a statistic. How much can this reduction increase the best achievable risk in a decision problem? For a finite hypothesis family, Le Cam's deficiency gives one answer simultaneously for every prior and every loss taking values in $[0,1]$: it equals the largest increase in Bayes risk when the statistic replaces the full histogram. Equivalently, it is the smallest worst-case total-variation error of a single reconstruction kernel,
\[
 \Delta=\inf_K\max_{\vartheta\in\Theta}\TV\bigl(KQ_\vartheta,P_\vartheta\bigr),
\]
where $P_\vartheta$ is the histogram law, $Q_\vartheta$ is its image under the retained statistic, and the same Markov kernel $K$ is used under every hypothesis \cite{LeCam1964,Torgersen1970,Torgersen1991}. Thus an expansion $\Delta_n=Cn^{-1/2}+o(n^{-1/2})$ quantifies approximate sufficiency with an exact first-order coefficient: the Bayes-risk loss of every such decision problem is at most this amount, and the supremum of these losses has this expansion. In the critical rare-event model studied here the retained statistic is, equivalently, the count $N(c)$ of the distinguished symbol $c$; the window of compositions is fixed as $n$ grows. Throughout, reconstruction means simulation of the histogram law under a specified hypothesis family. Skovgaard \cite[\S4, (4.5)]{Skovgaard1981} used the same quantity to measure information lost by reducing an iid sample to its maximum likelihood estimator, with upper bounds $O(n^{-1/2})$ and, after adjoining the observed information, $O(n^{-1})$ (the latter proved more generally by Michel \cite{Michel1978}). Here we consider finite-alphabet shuffle models with correlated message bundles or a critical rare-event channel and compute the exact leading coefficient for a finite fixed hypothesis family. The author's related papers concern full-support Gaussian asymptotics \cite{Shvets2602} and globally private count estimation \cite{Shvets2603count}; Section~\ref{sec:lit} specifies their relation to the present reconstruction problem.

\paragraph{Why more than two hypotheses.}
For a pair of hypotheses the deficiency is classical: by Torgersen's criterion \cite[Thms.\ 10 and 15]{Torgersen1970}, see also
\cite[Thm.\ 5, \S4.1]{Jencova2012}, it equals $\sup_{t\ge0}[D_t(P_1,P_0)-D_t(Q_1,Q_0)]/(1+t)$ with
$D_t(P,Q)=\sum_x(P(x)-tQ(x))_+$, so the pair problem is a one-dimensional optimisation over likelihood thresholds. For
three or more hypotheses the optimal kernels of the pairs need not be compatible with one kernel. The simplest
illustration is elementary (an uninformative statistic and three perfectly distinguishable targets give pairwise errors
$1/2$ but a common error $2/3$) and the general distinction between pairwise and family comparison is classical
\cite{Torgersen1970}. This paper computes the exact asymptotic size of the incompatibility in a structured growing
experiment: the leading coefficient of the common-kernel loss for a finite family, and its comparison with the pairwise
coefficients.

\paragraph{Results.}
Theorem~\ref{thm:tangent} gives the leading coefficient for a finite family in the bundle model of Section~\ref{sec:bundle}
as the value of the program \eqref{eq:tangent}. Section~\ref{sec:consequences} derives the two-hypothesis formula from it,
gives the exact finite-$N$ value of the best kernel that never moves mass between tag values
(Proposition~\ref{prop:fp}), proves the exact $4/3$ example (Proposition~\ref{prop:fourthirds}) and its finite-sample identity
(Corollary~\ref{cor:finiteidentity}), and records what the
program says about cross-tag transport. Proposition~\ref{prop:abstract} restates the limit identification in an abstract
form that separates the classical ingredients (the support-functional criterion, the exposed-face derivative, compactness of
convex bodies) from the model-specific input (a reference density tending to one and a score vector converging weakly with
uniformly integrable norm). Section~\ref{sec:rare} treats the critical rare-event channel along a window of
consecutive compositions (Theorem~\ref{thm:rare}) and evaluates the three-composition coefficient for every $\mu>0$
(Theorem~\ref{thm:three}: the simplex value $\sigma(\mu)$ for $0<\mu\le1$ by two dual certificates, the outer-pair value
$p_2(\mu)$ for $\mu\ge1$ by a single chain certificate), together with the exact comparison with the pairwise coefficients
and the structure of the optimal priors (Corollary~\ref{cor:priors}). Section~\ref{sec:windows-general} introduces the path envelope and Hypothesis~P (Conjecture~\ref{conj:paths}).
It proves the exact endpoint of the full-chain phase for every window (Theorem~\ref{thm:small-general}), with the family--pair
ratio in Corollary~\ref{cor:generalgapphase}; an explicit endpoint-pair regime (Theorem~\ref{thm:outer-general}); and the
complete four-composition phase diagram by exact computer-assisted certificates (Theorem~\ref{thm:four-compositions}).
Proposition~\ref{prop:unconditional-comparison-bounds} gives unconditional comparison bounds, including
$\mathcal V_m\le2P_m$ on the unit mean interval and at sufficiently large means. Together with
Proposition~\ref{prop:gap-two-lower}, it shows that the ratio at $\mu=\sqrt2$ is strictly below two for every finite window and
tends to two as $m\to\infty$. Proposition~\ref{prop:gap-unit} proves a strict family gain for $m\ge3$ throughout the unit mean
interval, Proposition~\ref{prop:path-envelope} gives a sharp factor-two comparison for the path envelope, and
Proposition~\ref{prop:path-dual-knots} identifies necessary knot constraints. Hypothesis~P remains open in the intermediate
regimes for $m\ge4$; the universal finite-window upper bound of two also remains open outside the proved comparison regimes.
Section~\ref{sec:num} lists the numerical checks with their tolerances, and Section~\ref{sec:lit} the closest prior results and what we regard as new.

\paragraph{Conventions.}
$\TV(P,Q)=\frac12\sum_x|P(x)-Q(x)|$. A kernel is a row-stochastic matrix; it is the same under every hypothesis and may send
a tag value anywhere. $\varphi$ and $\Phi$ are the standard normal density and distribution function,
$\varphi(0)=(2\pi)^{-1/2}=\E Z_+$ for $Z\sim N(0,1)$. Falling factorials are $(x)_j=x(x-1)\cdots(x-j+1)$, $\N$ includes $0$,
and for $v\in\R^\Theta$ and $u\in\R^\Theta$ we write $u\circ v$ for the coordinatewise product.

\section{Deficiency of a statistic for a finite family}\label{sec:finite}
Let $\Theta$ be a finite set, $\mathcal F=(Q_\vartheta)_{\vartheta\in\Theta}$ probability vectors on a finite set $Y$ and
$\mathcal E=(P_\vartheta)_{\vartheta\in\Theta}$ probability vectors on a finite set $X$. Define
\begin{equation}\label{eq:def}
 \delta(\mathcal F,\mathcal E)=\min_K\max_{\vartheta\in\Theta}\TV(KQ_\vartheta,P_\vartheta),
\end{equation}
the minimum over Markov kernels $K$ from $Y$ to $X$ (it is attained by compactness). When $Y$ is the image of $X$ under a
map $T$ and $Q_\vartheta=P_\vartheta\circ T^{-1}$, then $\delta(\mathcal E,\mathcal F)=0$ (push forward by $T$), and
$\delta(\mathcal F,\mathcal E)$ is both the Le Cam distance of the two experiments and the \emph{reconstruction loss} of the
statistic $T$. For $w$ in the simplex $\Delta(\Theta)$ put $B_w=\prod_\vartheta[0,w_\vartheta]\subset\R^\Theta$.

\begin{lemma}[Finite-family identity]\label{lem:identity}
\begin{equation}\label{eq:identity}
 \delta(\mathcal F,\mathcal E)
 =\max_{w\in\Delta(\Theta)}\ \max_{\alpha:X\to B_w}
 \Bigl\{\sum_{x\in X}\sum_{\vartheta}\alpha_\vartheta(x)P_\vartheta(x)
 -\sum_{y\in Y}\max_{x\in X}\sum_\vartheta\alpha_\vartheta(x)Q_\vartheta(y)\Bigr\}.
\end{equation}
\end{lemma}
\begin{proof}
By $\TV(KQ_\vartheta,P_\vartheta)=\max_{f:X\to[0,1]}\{P_\vartheta f-(KQ_\vartheta)f\}$ and
$\max_\vartheta c_\vartheta=\max_{w\in\Delta(\Theta)}\sum_\vartheta w_\vartheta c_\vartheta$,
\[
 \max_\vartheta\TV(KQ_\vartheta,P_\vartheta)
 =\max_{w,\alpha}\Bigl\{\sum_{x,\vartheta}\alpha_\vartheta(x)P_\vartheta(x)
   -\sum_{y}\sum_xK(y,x)\sum_\vartheta\alpha_\vartheta(x)Q_\vartheta(y)\Bigr\},
\]
where $\alpha_\vartheta(x)=w_\vartheta f_\vartheta(x)$ ranges over all maps $X\to B_w$. The set of pairs $(w,\alpha)$ and the set
of kernels are compact and convex, and the expression in braces is bilinear, so the finite-dimensional minimax theorem
interchanges $\min_K$ and $\max_{w,\alpha}$. For fixed $(w,\alpha)$ the minimum over $K$ is attained by sending each $y$
to a maximiser of $x\mapsto\sum_\vartheta\alpha_\vartheta(x)Q_\vartheta(y)$.
\end{proof}
Lemma~\ref{lem:identity} is the finite-dimensional form of Torgersen's Bayes-risk criterion
\cite[Thm.\ 1, p.\ 222]{Torgersen1970} with action space $X$ and losses $1-f_\vartheta\in[0,1]$; it is recorded only to fix
the normalisation. Two elementary facts about total variation are used later.

\begin{lemma}\label{lem:tvfacts}
For probability vectors $P,Q$ on a finite set and $c,c'\ge0$,
$\min_R\{c\,\TV(R,P)+c'\TV(R,Q)\}=\min(c,c')\,\TV(P,Q)$, the minimum over probability vectors $R$. If $P$ and $Q$ have
densities $f,g$ with respect to a common reference $M$, then $\TV(P,Q)=\E_M(g-f)_+=\frac12\E_M|g-f|$.
\end{lemma}
\begin{proof}
The triangle inequality gives $c\TV(R,P)+c'\TV(R,Q)\ge\min(c,c')\TV(P,Q)$, and $R\in\{P,Q\}$ attains it. The second
statement is the definition together with $\E_M(g-f)=0$.
\end{proof}

\section{The bundle model and its scores}\label{sec:bundle}
This section collects, with proofs, the facts about one correlated message bundle in an iid categorical background that
the main theorem uses. The projection algebra is Hoeffding's \cite{Hoeffding1948}; in the multinomial setting the
orthogonal polynomials are the multivariate Krawtchouk polynomials \cite{DG14}, and exact expansions of dependent
Bernoulli sums go back to Takeuchi and Takemura \cite{TT87}. The $L^2$ coupling of Lemma~\ref{lem:coupling} is the form
in which these classical expansions enter the family argument.

\subsection{Model}
Fix an alphabet $D=\{1,\dots,d\}$, a probability vector $p$ with $p_i>0$ for every $i$, and an integer $L\ge1$. Put
$\mathcal B_L=\{b\in\N^d:|b|=L\}$ and $M_{L,p}(b)=\frac{L!}{\prod_ib_i!}\prod_ip_i^{b_i}$. For a probability law $\nu$ on
$\mathcal B_L$ and $N\ge L$ define
\begin{equation}\label{eq:bundlemodel}
 P_{\nu,N}=\Law(S_{N-L}+B),\qquad S_{N-L}\sim\Mult(N-L,p),\quad B\sim\nu\ \text{independent},
\end{equation}
a law on $\mathcal B_N$, and the reference law $M_N=\Mult(N,p)$. Write $R_{\nu,N}=dP_{\nu,N}/dM_N$ (both laws are
supported in $\mathcal B_N$ and $M_N$ charges every point of it). For a string $z\in D^L$ with histogram $b(z)$ put
\begin{equation}\label{eq:g}
 g_\nu(z)=\frac{\nu(b(z))}{M_{L,p}(b(z))},
\end{equation}
a bounded symmetric function with $\E g_\nu(Y_1,\dots,Y_L)=1$ for iid $Y_i\sim p$.

\begin{lemma}[Exact $U$-statistic identity]\label{lem:ustat}
Let $Y_1,\dots,Y_N$ be iid with law $p$ and $S_N$ their histogram. Then, almost surely,
\begin{equation}\label{eq:ustat}
 R_{\nu,N}(S_N)=\E\bigl[g_\nu(Y_1,\dots,Y_L)\,\big|\,S_N\bigr]
 =\binom NL^{-1}\sum_{I\subseteq\{1,\dots,N\},|I|=L}g_\nu(Y_I).
\end{equation}
In particular $0\le R_{\nu,N}\le\max_zg_\nu(z)$, and for $x\in\mathcal B_N$
\begin{equation}\label{eq:falling}
 R_{\nu,N}(x)=\frac{\E_\nu\prod_i(x_i)_{B_i}\,p_i^{-B_i}}{(N)_L}.
\end{equation}
\end{lemma}
\begin{proof}
Let $\tilde P$ be the law on $D^N$ with $\tilde P(y)=g_\nu(y_1,\dots,y_L)\prod_{i\le N}p_{y_i}$. By \eqref{eq:g},
$\tilde P(y_1,\dots,y_L)=\nu(b)/\#\{z:b(z)=b\}$ with $b=b(y_1,\dots,y_L)$, so under $\tilde P$ the first $L$ labels have a
histogram with law $\nu$, and the remaining $N-L$ labels are iid $p$ and independent of them. Hence the histogram of all
$N$ labels has law $P_{\nu,N}$ under $\tilde P$, and for every function $\phi$ on $\mathcal B_N$,
$\E_{P_{\nu,N}}\phi(S)=\E[g_\nu(Y_1,\dots,Y_L)\phi(S_N)]=\E[\E[g_\nu\mid S_N]\phi(S_N)]$; this identifies the density. Given
$S_N=s$, the vector $(Y_1,\dots,Y_N)$ is uniformly distributed over the arrangements with histogram $s$, so
$\E[g_\nu(Y_I)\mid S_N]$ does not depend on the $L$-subset $I$; averaging over $I$ gives the $U$-statistic, which is a
symmetric function of the labels and hence a function of $S_N$. The bound follows. For \eqref{eq:falling}, condition on
$S_N=x$: the first $L$ labels are drawn without replacement from the histogram $x$, so the probability that they form a
string with histogram $b$ is $\frac{L!}{\prod b_i!}\prod_i(x_i)_{b_i}/(N)_L$; multiply by $g_\nu$ and sum.
\end{proof}

\subsection{Projections and factorial moments}
Let $g:D^L\to\R$ be symmetric with $\E g(Y_1,\dots,Y_L)=0$ for iid $Y_i\sim p$. For $0\le j\le L$ let
\begin{equation}\label{eq:psidef}
 \psi_j(z_1,\dots,z_j)=\E\,g(z_1,\dots,z_j,Y_{j+1},\dots,Y_L)\qquad(\text{so }\psi_0=0,\ \psi_L=g),
\end{equation}
\begin{equation}\label{eq:canonical}
 h_j(z_1,\dots,z_j)=\sum_{A\subseteq\{1,\dots,j\}}(-1)^{j-|A|}\psi_{|A|}(z_A),\qquad 1\le j\le L .
\end{equation}
Put $h_0=\psi_0=0$. Let $\phi_1,\dots,\phi_{d-1}$ be an orthonormal basis of the mean-zero subspace of $L^2(p)$ and $\Phi=(\phi_a)_a$.

\begin{lemma}[Hoeffding decomposition]\label{lem:hoeffding}
\begin{enumerate}[label=\textup{(\alph*)}]
\item $g(z)=\sum_{j=1}^L\sum_{|A|=j}h_j(z_A)$ for all $z\in D^L$.
\item Each $h_j$ is symmetric and canonical: $\E h_j(z_1,\dots,z_{j-1},Y)=0$ for all $z$.
\item $h_j(z_1,\dots,z_j)=\sum_{a_1,\dots,a_j=1}^{d-1}\theta^{(j)}_{a_1\cdots a_j}\prod_{l=1}^j\phi_{a_l}(z_l)$ with a symmetric
 tensor $\theta^{(j)}$, and $\E h_j(Y_1,\dots,Y_j)^2=\sum_a(\theta^{(j)}_a)^2$.
\item With $U_{N,j}(h_j)=\binom Nj^{-1}\sum_{|I|=j}h_j(Y_I)$, for $j\ne k$ the variables $U_{N,j}(h_j)$ and $U_{N,k}(h_k)$
 are orthogonal, and $\E U_{N,j}(h_j)^2=\binom Nj^{-1}\E h_j^2$.
\end{enumerate}
\end{lemma}
\begin{proof}
(a) By M\"obius inversion, \eqref{eq:canonical} is equivalent to
\[
 \psi_{|B|}(z_B)=\sum_{A\subseteq B}h_{|A|}(z_A)\qquad\text{for every }B\subseteq\{1,\dots,L\};
\]
take $B=\{1,\dots,L\}$. (b) Symmetry is inherited from $\psi$. Integrate the last variable in \eqref{eq:canonical} and
use $\E\,\psi_k(z_1,\dots,z_{k-1},Y)=\psi_{k-1}(z_1,\dots,z_{k-1})$: the terms with $j\in A$ and those with $j\notin A$ cancel
in pairs $A\leftrightarrow A\setminus\{j\}$. (c) Expand $h_j$ in the orthonormal basis
$\{1,\phi_1,\dots,\phi_{d-1}\}^{\otimes j}$ of $L^2(p^{\otimes j})$; canonicity kills every term containing the constant in
some coordinate; symmetry of $h_j$ gives symmetry of the tensor; Parseval gives the norm. (d) For $I\ne I'$ choose an
index in one set and not the other and take the conditional expectation given all other labels; by (b) it vanishes.
The same works for different orders. For $I=I'$ the term is $\E h_j^2$, and there are $\binom Nj$ subsets.
\end{proof}

\begin{lemma}[Factorial-moment criterion]\label{lem:moments}
Let $\nu_0\ne\nu_1$ be laws on $\mathcal B_L$ and $g=g_{\nu_1}-g_{\nu_0}$. For a string $z\in D^s$ with histogram $c$,
\begin{equation}\label{eq:psi}
 \psi_s(z)=\frac{\E_{\nu_1}\prod_i(B_i)_{c_i}-\E_{\nu_0}\prod_i(B_i)_{c_i}}{(L)_s\prod_ip_i^{c_i}} .
\end{equation}
Consequently $h_1=\dots=h_s=0$ if and only if all factorial moments $\E\prod_i(B_i)_{c_i}$ of total degree $|c|\le s$
agree, and the integer
\begin{equation}\label{eq:r}
 r=\min\{j\ge1:\E h_j(Y_1,\dots,Y_j)^2>0\}
 =\min\{|c|\ge1:\E_{\nu_1}\textstyle\prod_i(B_i)_{c_i}\ne\E_{\nu_0}\prod_i(B_i)_{c_i}\}
\end{equation}
exists and satisfies $r\le L$.
\end{lemma}
\begin{proof}
Let $\tilde P_\nu$ be the law of a uniformly ordered bundle with histogram law $\nu$. Then
\[
 \psi_s^{(\nu)}(z)=\sum_{y\in D^{L-s}}g_\nu(z,y)\prod_{i}p_{y_i}
 =\frac{\tilde P_\nu(Y_1=z_1,\dots,Y_s=z_s)}{\prod_ip_{z_i}},\quad
 \tilde P_\nu(Y_1=z_1,\dots,Y_s=z_s)=\frac{\E_\nu\prod_i(B_i)_{c_i}}{(L)_s},
\]
the second identity by drawing without replacement from the histogram; subtract the two laws. The
equivalence follows from the triangular relation between $(\psi_j)$ and $(h_j)$. At $s=L$ and $c=b\in\mathcal B_L$,
$\prod_i(B_i)_{b_i}$ vanishes unless $B=b$, where it equals $\prod_ib_i!>0$; so distinct laws differ at some degree at most
$L$, and $r$ exists.
\end{proof}

\subsection{Wick polynomials and the \texorpdfstring{$L^2$}{L2} coupling}
For $Z\sim N(0,I_{d-1})$ define the Wick products by
\begin{equation}\label{eq:wick}
 {:}Z_{a_1}\cdots Z_{a_j}{:}=\partial_{a_1}\cdots\partial_{a_j}\exp\bigl(t\cdot Z-\tfrac12|t|^2\bigr)\Big|_{t=0};
\end{equation}
if the index $a$ occurs $\alpha_a$ times this is $\prod_a\mathrm{He}_{\alpha_a}(Z_a)$ with the probabilists' Hermite
polynomials. From $\E[\exp(t\cdot Z-\frac12|t|^2)\exp(s\cdot Z-\frac12|s|^2)]=\exp(t\cdot s)$ one reads off, comparing
coefficients, that Wick products of different degrees are orthogonal and that for a symmetric tensor $\theta$ of order $j$,
$\E\bigl(\sum_a\theta_a{:}Z_{a_1}\cdots Z_{a_j}{:}\bigr)^2=j!\sum_a\theta_a^2$. For the difference kernel
$g=g_{\nu_1}-g_{\nu_0}$ with first non-vanishing order $r$ and tensor $\theta^{(r)}$ from Lemma~\ref{lem:hoeffding}(c) set
\begin{equation}\label{eq:H}
 H(z)=\binom Lr\sum_{a_1,\dots,a_r}\theta^{(r)}_{a_1\cdots a_r}{:}z_{a_1}\cdots z_{a_r}{:},\qquad z\in\R^{d-1},
\end{equation}
a polynomial with $\E H(Z)=0$ and $\Var H(Z)=\binom Lr^2r!\,\E h_r^2>0$. Finally let
$Z_N=N^{-1/2}\sum_{i\le N}\Phi(Y_i)$, which is a function of the histogram $S_N$: $Z_N(x)=N^{-1/2}\sum_{i\in D}x_i\Phi(i)$.

\begin{lemma}[$L^2$ coupling]\label{lem:coupling}
With $r$ as in \eqref{eq:r} and $H$ as in \eqref{eq:H},
\[
 \Bigl\|N^{r/2}\bigl(R_{\nu_1,N}-R_{\nu_0,N}\bigr)(S_N)-H(Z_N)\Bigr\|_{L^2(M_N)}\le\frac C{\sqrt N},
\]
with $C$ depending only on $p,L,\nu_0,\nu_1$. In particular $N^{r/2}(R_{\nu_1,N}-R_{\nu_0,N})\Rightarrow H(Z)$ under $M_N$,
the family $\{N^{r/2}(R_{\nu_1,N}-R_{\nu_0,N})\}_N$ is bounded in $L^2(M_N)$, and
$N^{r/2}\TV(P_{\nu_0,N},P_{\nu_1,N})\to\frac12\E|H(Z)|$.
\end{lemma}
\begin{proof}
By Lemmas~\ref{lem:ustat} and \ref{lem:hoeffding}(a), with $g=g_{\nu_1}-g_{\nu_0}$,
\begin{equation}\label{eq:hoeffexp}
 R_{\nu_1,N}-R_{\nu_0,N}=\sum_{j=r}^L\binom LjU_{N,j}(h_j),
\end{equation}
and by Lemma~\ref{lem:hoeffding}(d) the terms with $j>r$ have $L^2$ norm $O(N^{-(r+1)/2})$. It remains to compare
$N^{r/2}\binom LrU_{N,r}(h_r)$ with $H(Z_N)$. Consider the polynomial in $t\in\R^{d-1}$
\[
 F_N(t)=\prod_{i=1}^N\bigl(1+N^{-1/2}t\cdot\Phi(Y_i)\bigr)=\sum_{k=0}^NN^{-k/2}\sum_{|I|=k}\prod_{i\in I}t\cdot\Phi(Y_i).
\]
In the degree-$r$ part, $\partial_{a_1}\cdots\partial_{a_r}\prod_{i\in I}(t\cdot\Phi(Y_i))\big|_{t=0}$ is the sum over
bijections $\sigma:\{1,\dots,r\}\to I$ of $\prod_l\phi_{a_l}(Y_{\sigma(l)})$; by symmetry of $h_r$,
\begin{equation}\label{eq:contract}
 \sum_a\theta^{(r)}_a\,\partial_{a_1}\cdots\partial_{a_r}F_N(t)\Big|_{t=0}
 =N^{-r/2}\sum_{|I|=r}r!\,h_r(Y_I)=N^{-r/2}(N)_rU_{N,r}(h_r).
\end{equation}
As formal power series in $t$,
\begin{align*}
 \log F_N(t)&=t\cdot Z_N-\tfrac12t^{\mathsf T}\Sigma_Nt+\sum_{j\ge3}\frac{(-1)^{j+1}}{j}N^{1-j/2}c_{N,j}(t),\\
 \Sigma_N&=\frac1N\sum_i\Phi(Y_i)\Phi(Y_i)^{\mathsf T},\qquad c_{N,j}(t)=\frac1N\sum_i(t\cdot\Phi(Y_i))^j,
\end{align*}
and $F_N=\exp(\log F_N)$ as formal series, so the Taylor coefficients of $F_N$ through degree $r$ are polynomials in the
entries of $Z_N$, of $\Sigma_N-I$, and of the coefficients of $N^{1-j/2}c_{N,j}$ for $3\le j\le r$. Writing
$F_N=\exp(t\cdot Z_N-\frac12|t|^2)\cdot\exp(\Delta_N(t))$ with
$\Delta_N(t)=-\frac12t^{\mathsf T}(\Sigma_N-I)t+\sum_{j\ge3}(\ldots)$, the degree-$r$ coefficients of $F_N$ are those of
$\exp(t\cdot Z_N-\frac12|t|^2)$, which produce exactly $\sum_a\theta_a{:}Z_{N,a_1}\cdots Z_{N,a_r}{:}=H(Z_N)/\binom Lr$ after
contraction, plus a finite sum of monomials each containing at least one entry of $\Sigma_N-I$ or one factor
$N^{1-j/2}$ with $j\ge3$, multiplied by a polynomial in $Z_N$ and in the (bounded) coefficients of $c_{N,j}$. The
summands $\Phi(Y_i)$ are bounded and centred, so by the Marcinkiewicz--Zygmund inequality every moment of $Z_N$ is bounded
uniformly in $N$ and $\|(\Sigma_N-I)_{ab}\|_{L^q}\le C_qN^{-1/2}$ for every $q<\infty$; H\"older's inequality bounds the
$L^2$ norm of each extra monomial by $CN^{-1/2}$. Hence
$\|N^{-r/2}(N)_rU_{N,r}(h_r)-H(Z_N)/\binom Lr\|_{L^2}\le CN^{-1/2}$, and since $N^r/(N)_r=1+O(N^{-1})$ and $H(Z_N)$ is
bounded in $L^2$, also $\|N^{r/2}\binom LrU_{N,r}(h_r)-H(Z_N)\|_{L^2}\le CN^{-1/2}$. Weak convergence follows from the
central limit theorem $Z_N\Rightarrow Z$ and the continuity of $H$; uniform integrability from the $L^2$ bound; and the
total-variation limit from Lemma~\ref{lem:tvfacts}.
\end{proof}

\begin{lemma}[Central region]\label{lem:central}
Let $p_{\min}=\min_ip_i$, $\eta=p_{\min}/(4L)$ and $C_N=\{x\in\mathcal B_N:|x_i/N-p_i|\le\eta\ \forall i\}$. Then
$M_N(C_N^c)\le2d\,e^{-2N\eta^2}$, and for $N\ge4L^2/p_{\min}$ and every law $\nu$ on $\mathcal B_L$,
$R_{\nu,N}(x)\ge\frac12$ for all $x\in C_N$.
\end{lemma}
\begin{proof}
The first bound is Hoeffding's inequality applied to each coordinate of the histogram. On $C_N$, \eqref{eq:falling}
gives $R_{\nu,N}(x)\ge\E_\nu\prod_i\bigl((x_i-L)/(Np_i)\bigr)^{B_i}\ge\bigl(1-(\eta+L/N)/p_{\min}\bigr)^L\ge(1-\tfrac1{2L})^L\ge\tfrac12$.
\end{proof}

\begin{corollary}[First order]\label{cor:firstorder}
Let $\Delta\mu=\E_{\nu_1}B-\E_{\nu_0}B$. If $\Delta\mu\ne0$ then $r=1$, $H(Z)$ is centred Gaussian with variance
$J_1=\sum_i(\Delta\mu_i)^2/p_i$, and $\sqrt N\TV(P_{\nu_0,N},P_{\nu_1,N})\to\sqrt{J_1}\,\varphi(0)$. For $d=2$, $L=1$,
$p=(p_1,1-p_1)$, bundle laws putting mass $q^{(h)}$ on the first symbol, and the basis function
$\phi_1=\bigl(\sqrt{(1-p_1)/p_1},\,-\sqrt{p_1/(1-p_1)}\bigr)$, one has $H(Z)=\frac{q^{(1)}-q^{(0)}}{\sqrt{p_1(1-p_1)}}Z$ (with the
opposite orientation of $\phi_1$ the sign of $H$ flips; its law does not depend on the orientation).
\end{corollary}
\begin{proof}
By \eqref{eq:psi}, $h_1(i)=\psi_1(i)=\Delta\mu_i/(Lp_i)$, and $H=L\sum_a\langle h_1,\phi_a\rangle Z_a$ has variance
$L^2\|h_1\|^2_{L^2(p)}=J_1$ because $\sum_i\Delta\mu_i=0$ makes $h_1$ mean-zero. Then $\E|H|/2=\sqrt{J_1}\,\varphi(0)$.
For $d=2$, $L=1$: $\Delta\mu=(\delta,-\delta)$ with $\delta=q^{(1)}-q^{(0)}$, $h_1=(\delta/p_1,-\delta/(1-p_1))$, and
$\langle h_1,\phi_1\rangle_{L^2(p)}=\delta\bigl(\sqrt{(1-p_1)/p_1}+\sqrt{p_1/(1-p_1)}\bigr)=\delta/\sqrt{p_1(1-p_1)}$.
\end{proof}

\section{The tangent program for a finite family}\label{sec:tangent}
\subsection{Model and scores}
Fix a finite tag set $\mathcal Y$, an alphabet $D$ with a strictly positive law $p$, a bundle size $L$, and a finite
hypothesis set $\Theta$. For $\vartheta\in\Theta$ and $y\in\mathcal Y$ fix coarse weights $a_{\vartheta,y}\ge0$ with
$\sum_ya_{\vartheta,y}=1$ and bundle laws $\nu_{\vartheta,y}$ on $\mathcal B_L$. Under $\vartheta$ one observes $(Y,X)$ with
\begin{equation}\label{eq:model}
 \Pr_\vartheta(Y=y)=a_{\vartheta,y},\qquad X\mid\{Y=y\}\sim P_{\vartheta,y,N}:=P_{\nu_{\vartheta,y},N}.
\end{equation}
The coarse weights do not depend on $N$; they may vanish, and we put $q_y=(a_{\vartheta,y})_\vartheta\in\R^\Theta$ and
$\Theta_y=\{\vartheta:a_{\vartheta,y}>0\}$, assumed non-empty for every $y$. The full experiment is
$\mathcal E_N=(P_{\vartheta,N})_\vartheta$ with $P_{\vartheta,N}(y,x)=a_{\vartheta,y}P_{\vartheta,y,N}(x)$ on
$\mathcal Y\times\mathcal B_N$, the reduced experiment is $\mathcal F=(q_y)_{y}$ (the law of $Y$), and the loss is
$\Delta_N=\delta(\mathcal F,\mathcal E_N)$. This is an anonymous-message model: one changed user sends a tag $y$ on a
separate alphabet together with a correlated bundle of $L$ untagged messages, the other messages are iid $p$, and the analyst
sees the tag and the histogram but not the bundle membership.

In each fiber $y$ fix a reference $\vartheta_y\in\Theta_y$ and set, for $x\in\mathcal B_N$,
\begin{equation}\label{eq:Vdef}
 V_{\vartheta,y,N}(x)=R_{\vartheta,y,N}(x)-R_{\vartheta_y,y,N}(x)\ \ (\vartheta\in\Theta_y),\qquad V_{\vartheta,y,N}=0\ \ (\vartheta\notin\Theta_y),
\end{equation}
where $R_{\vartheta,y,N}=dP_{\vartheta,y,N}/dM_N$. Let $r\ge1$ be the smallest total degree at which the factorial moments of
$\nu_{\vartheta,y}$ and $\nu_{\vartheta_y,y}$ differ for some $y$ and some $\vartheta\in\Theta_y$ (Lemma~\ref{lem:moments}). If
no such degree exists, all conditional laws agree within every fiber, drawing $X$ from the common conditional law given $Y$
is an exact common kernel, and $\Delta_N=0$; we exclude this case. For $\vartheta\in\Theta_y$ let $H_{\vartheta,y}$ be the
degree-$r$ polynomial \eqref{eq:H} of the pair $(\nu_{\vartheta_y,y},\nu_{\vartheta,y})$, with $H_{\vartheta,y}=0$ if that pair
first differs at a degree above $r$ or if $\vartheta\notin\Theta_y$ (including $\vartheta=\vartheta_y$). All polynomials of
one fiber are evaluated at the same Gaussian vector $Z_y\sim N(0,I_{d-1})$; by Lemma~\ref{lem:coupling},
\begin{equation}\label{eq:jointcoupling}
 \bigl\|N^{r/2}V_{\vartheta,y,N}-H_{\vartheta,y}(Z_N)\bigr\|_{L^2(M_N)}\le CN^{-1/2}\qquad\text{for all }\vartheta,y,
\end{equation}
with the same $Z_N=Z_N(x)$ for all $\vartheta$ in a fiber. Put $h_N=N^{-r/2}$ and $q_y\circ H_y(z)=(a_{\vartheta,y}H_{\vartheta,y}(z))_\vartheta$.

\subsection{The program}
For $w\in\Delta(\Theta)$ and $m\in\R^{\mathcal Y}$ define the polytopes
\begin{equation}\label{eq:polytope}
 \Pi_y(w,m)=\bigl\{\alpha\in B_w:\ \langle\alpha,q_y\rangle=m_y,\ \langle\alpha,q_{y'}\rangle\le m_{y'}\ \ (y'\ne y)\bigr\},
\end{equation}
let $h_\Pi(v)=\max_{\alpha\in\Pi}\langle\alpha,v\rangle$ ($h_\emptyset\equiv-\infty$), and define the \emph{tangent value}
\begin{equation}\label{eq:tangent}
 \mathcal V=\sup_{w\in\Delta(\Theta)}\ \sup_{m\in\R^{\mathcal Y}}\
 \sum_{y\in\mathcal Y}\E\,h_{\Pi_y(w,m)}\bigl(q_y\circ H_y(Z_y)\bigr).
\end{equation}
Geometrically, when non-empty, $\Pi_y(w,m)$ is the face exposed by the direction $q_y$ of the body
$A(w,m)=\{\alpha\in B_w:\langle\alpha,q_{y'}\rangle\le m_{y'}\ \forall y'\}$, and every compact convex $A\subseteq B_w$ has
$\E h_{F_y(A)}(q_y\circ H_y)\le\E h_{\Pi_y(w,h_A(q))}(q_y\circ H_y)$ for its exposed faces $F_y(A)$, because the intersection of
supporting half-spaces contains $A$ and has larger faces. That the directional derivative of a support function
$h_A(q+\epsilon v)$ at $\epsilon=0$ is the support function of the exposed face $F_A(q)$ is Danskin's theorem
\cite{Danskin1967,Bertsekas}; the program \eqref{eq:tangent} is an optimisation over such derivatives. The passage from the
scores of Section~\ref{sec:bundle} to this optimised program is a general first-order comparison statement, recorded in
abstract form in Proposition~\ref{prop:abstract} below; the model-specific content of Theorem~\ref{thm:tangent} is the
identification of the scores $H_{\vartheta,y}$ and of the rate $N^{-r/2}$.
Since every $H_{\vartheta,y}$ is centred, $\mathcal V\ge0$, and a fiber whose polytope is a single point contributes $0$.

\begin{theorem}[Tangent program]\label{thm:tangent}
In the model \eqref{eq:model},
\[
 \lim_{N\to\infty}N^{r/2}\Delta_N=\mathcal V .
\]
The function $(w,m)\mapsto\sum_y\E h_{\Pi_y(w,m)}(q_y\circ H_y(Z_y))$ is concave on the set where all polytopes are non-empty,
the supremum is attained, and $\mathcal V<\infty$.
\end{theorem}

\begin{proof}[Proof of the lower bound]
Fix $w$, $m$ with all $\Pi_y(w,m)$ non-empty and measurable $\alpha_y:\R^{d-1}\to\Pi_y(w,m)$. Mollifying $\alpha_y$ keeps its
values in the convex set $\Pi_y(w,m)$ and changes $\E\langle\alpha_y(Z),q_y\circ H_y(Z)\rangle$ by as little as we like,
since $H_y$ is square integrable; so assume $\alpha_y$ continuous. Define finite-$N$ tests by $\alpha(y,x)=\alpha_y(Z_N(x))$
and insert them into \eqref{eq:identity}. For $(y,x)$ in fiber $y$, $\langle\alpha(y,x),q_y\rangle=m_y$; for $(y',x)$ in
another fiber, $\langle\alpha(y',x),q_y\rangle\le m_y$; hence
$\max_{(y',x)}\sum_\vartheta\alpha_\vartheta(y',x)a_{\vartheta,y}=m_y$ exactly. Using the weighted identity
$a_{\vartheta,y}R_{\vartheta,y,N}=a_{\vartheta,y}\bigl(R_{\vartheta_y,y,N}+V_{\vartheta,y,N}\bigr)$ for every $\vartheta$ (both sides
vanish for $\vartheta\notin\Theta_y$) and $\E_{M_N}R_{\vartheta_y,y,N}=1$,
\[
 \sum_{(y,x)}\sum_\vartheta\alpha_\vartheta(y,x)P_{\vartheta,N}(y,x)
 =\sum_ym_y+\sum_y\E_{M_N}\bigl\langle\alpha_y(Z_N(X)),\,q_y\circ V_{\cdot,y,N}(X)\bigr\rangle ,
\]
so the value of these tests in \eqref{eq:identity} is exactly
$\sum_y\E_{M_N}\langle\alpha_y(Z_N(X)),q_y\circ V_{\cdot,y,N}(X)\rangle$. Multiply by $N^{r/2}$: by
\eqref{eq:jointcoupling}, boundedness of $\alpha_y$, the central limit theorem $Z_N\Rightarrow Z$ and uniform integrability of
$H_y(Z_N)$, it converges to $\sum_y\E\langle\alpha_y(Z),q_y\circ H_y(Z)\rangle$. Choosing $\alpha_y(z)$ as a measurable
selection of $\arg\max_{\alpha\in\Pi_y(w,m)}\langle\alpha,q_y\circ H_y(z)\rangle$ gives
$\liminf_NN^{r/2}\Delta_N\ge\sum_y\E h_{\Pi_y(w,m)}(q_y\circ H_y(Z))$ for every admissible $(w,m)$.
\end{proof}

\begin{proof}[Proof of the upper bound]
For every $N$ let $(w^N,\alpha^N)$ be a maximiser in \eqref{eq:identity} and put
\[
 M^N_y=\max_{(y',x)}\sum_\vartheta\alpha^N_\vartheta(y',x)\,a_{\vartheta,y},\qquad
 \varepsilon^N(y,x)=M^N_y-\langle\alpha^N(y,x),q_y\rangle\in[0,1]\quad\text{for $(y,x)$ in fiber $y$}.
\]
The computation of the lower bound (with the same weighted identity) gives the exact decomposition
\begin{equation}\label{eq:split}
 \Delta_N=-\sum_y\E_{P_{\vartheta_y,y,N}}\varepsilon^N(y,X)
 +\sum_y\E_{M_N}\bigl\langle\alpha^N(y,X),q_y\circ V_{\cdot,y,N}(X)\bigr\rangle .
\end{equation}
The second sum is at most $\sum_y\sum_\vartheta a_{\vartheta,y}\E_{M_N}|V_{\vartheta,y,N}|\le Ch_N$ by
\eqref{eq:jointcoupling}, hence $\sum_y\E_{P_{\vartheta_y,y,N}}\varepsilon^N\le Ch_N$. By Lemma~\ref{lem:central},
$R_{\vartheta_y,y,N}\ge\frac12$ on $C_N$ for large $N$ and $M_N(C_N^c)\to0$; since $0\le\varepsilon^N\le1$,
$\E_{M_N}\varepsilon^N(y,X)\le2\E_{P_{\vartheta_y,y,N}}\varepsilon^N+M_N(C_N^c)\to0$ for every $y$.

First choose a subsequence realising $\limsup_N N^{r/2}\Delta_N$, and then a further subsequence along which
$w^N\to w$, $M^N_y\to M_y$ for all $y$, and for each $y$ the law of
$(Z_N(X),\alpha^N(y,X))$ under $M_N$ converges weakly to the law of a pair $(Z,A_y)$; this is possible because
$Z_N\Rightarrow Z$ and $\alpha^N$ takes values in $[0,1]^\Theta$. Along the subsequence, by \eqref{eq:split},
\eqref{eq:jointcoupling}, uniform integrability of $H_y(Z_N)$ and boundedness of $\alpha^N$,
\[
 \limsup_NN^{r/2}\Delta_N\le\lim_N\sum_yN^{r/2}\E_{M_N}\langle\alpha^N(y,X),q_y\circ V_{\cdot,y,N}(X)\rangle
 =\sum_y\E\langle A_y,q_y\circ H_y(Z)\rangle .
\]
The closed constraints $\langle\alpha^N(y,x),q_{y'}\rangle\le M^N_{y'}$ and $\alpha^N\in B_{w^N}$ pass to the limit:
$\langle A_y,q_{y'}\rangle\le M_{y'}$ and $A_y\in B_w$ almost surely; and $\E_{M_N}\varepsilon^N\to0$ gives
$\langle A_y,q_y\rangle=M_y$ almost surely. Replacing $A_y$ by $\bar\alpha_y(Z)=\E[A_y\mid Z]$ changes nothing, by
linearity, and $\bar\alpha_y(Z)\in\Pi_y(w,M)$ almost surely because the polytope is closed and convex. Therefore
$\limsup_NN^{r/2}\Delta_N\le\sum_y\E h_{\Pi_y(w,M)}(q_y\circ H_y(Z))\le\mathcal V$.
\end{proof}

\begin{proof}[Proof of the remaining assertions]
For fixed $v$, $h_{\Pi_y(w,m)}(v)$ is the optimal value of a linear program with a fixed constraint matrix and a right-hand side
linear in $(w,m)$, hence concave in $(w,m)$ on the set where the polytope is non-empty, and so is its expectation. On that
set $m_y\in[0,1]$, and the objective is bounded above by
$\sum_y\E\sum_\vartheta|a_{\vartheta,y}H_{\vartheta,y}(Z_y)|<\infty$, since every $\alpha\in B_w$ has coordinates in $[0,1]$.
Attainment follows directly from the upper-bound argument: its limit point $(w,M)$ has all polytopes non-empty and satisfies
\[
 \mathcal V\le\liminf_NN^{r/2}\Delta_N\le\limsup_NN^{r/2}\Delta_N
 \le\sum_y\E h_{\Pi_y(w,M)}(q_y\circ H_y(Z_y))\le\mathcal V.
\]
Thus this point attains the supremum. Proposition~\ref{prop:abstract} gives attainment in the equivalent convex-body form.
\end{proof}

\subsection{An abstract form of the limit}\label{sec:abstract}
The proof of Theorem~\ref{thm:tangent} uses the bundle model only through two facts: the reference density in each fiber
tends to one, and the scaled score vector converges weakly with uniformly integrable norm. The following proposition
records the limit identification under exactly these hypotheses. It is a first-order local comparison at a fixed
\emph{informative} coarse experiment (the vectors $q_y$ differ at order one), and its proof consists of the classical
ingredients only: the support-functional form of the finite-family identity, the exposed-face derivative of a support
function, and compactness of convex bodies. The hypotheses require a finite tag set and coarse vectors that do not depend
on $N$; the rare-event channel of Section~\ref{sec:rare}, with countably many fibers, is handled there by a separate
truncation argument.

Let $\Theta$ and $\mathcal Y$ be finite, $q_y=(a_{\vartheta,y})_\vartheta\in[0,1]^\Theta$ with $\sum_ya_{\vartheta,y}=1$ for each
$\vartheta$. For each $N$ and $y$ let $X_{y,N}$ be a finite set, $M_{y,N}$ a probability vector charging every point of it, and
$f_{\vartheta,y,N}\ge0$ functions on $X_{y,N}$ with $\sum_xf_{\vartheta,y,N}(x)M_{y,N}(x)=a_{\vartheta,y}$; the full experiment is
$P_{\vartheta,N}(y,x)=f_{\vartheta,y,N}(x)M_{y,N}(x)$ on $\bigsqcup_y\{y\}\times X_{y,N}$ and the reduced experiment is
$\mathcal F=(q_y)_y$. Write $f_{y,N}=(f_{\vartheta,y,N})_\vartheta$. Assume that for some $\epsilon_N\downarrow0$
\begin{equation}\label{eq:abstractmodel}
 f_{y,N}=\rho_{y,N}\,q_y+\epsilon_N\,W_{y,N},
\end{equation}
where (A1) $\rho_{y,N}\ge0$, $\E_{M_{y,N}}\rho_{y,N}=1$ and $\rho_{y,N}\to1$ in $M_{y,N}$-probability, and (A2) under
$M_{y,N}$ the $\R^\Theta$-valued vector $W_{y,N}$ converges weakly to a limit $W_y$ and $(|W_{y,N}|)_N$ is uniformly
integrable. (Then $\E W_{y,N}=0$.)  For $w\in\Delta(\Theta)$ let $\mathcal K_w$ be the set of non-empty compact convex
$A\subseteq B_w$, $h_A(v)=\max_{\alpha\in A}\langle\alpha,v\rangle$ and $F_A(q)=\{\alpha\in A:\langle\alpha,q\rangle=h_A(q)\}$.

\begin{proposition}[Abstract limit]\label{prop:abstract}
Under \textup{(A1)}--\textup{(A2)},
\[
 \lim_{N\to\infty}\frac{\delta(\mathcal F,\mathcal E_N)}{\epsilon_N}
 =\max_{w\in\Delta(\Theta)}\max_{A\in\mathcal K_w}\sum_y\E\,h_{F_A(q_y)}(W_y)
 =\sup_{w\in\Delta(\Theta)}\sup_{m\in\R^{\mathcal Y}}\sum_y\E\,h_{\Pi_y(w,m)}(W_y),
\]
and the first maximum is attained. In the model \eqref{eq:model} the hypotheses hold with $\epsilon_N=N^{-r/2}$,
$\rho_{y,N}=R_{\vartheta_y,y,N}$, $W_{y,N}=N^{r/2}\,q_y\circ V_{\cdot,y,N}$ and $W_y=q_y\circ H_y(Z_y)$, and the proposition
gives the limit in Theorem~\ref{thm:tangent} again.
\end{proposition}
\begin{proof}
\emph{Body form of Lemma~\ref{lem:identity}.}  For a test $\alpha:\bigsqcup_y\{y\}\times X_{y,N}\to B_w$ let $A$ be the
convex hull of its (finite) range, an element of $\mathcal K_w$. Then
\[
 \sum_{(y,x)}\langle\alpha(y,x),P_{\cdot,N}(y,x)\rangle=\sum_y\E_{M_{y,N}}\langle\alpha(y,\cdot),f_{y,N}\rangle\le\sum_y\E_{M_{y,N}}h_A(f_{y,N}),
\]
and $\max_{(y',x)}\langle\alpha(y',x),q_y\rangle=h_A(q_y)$ because a linear functional has the same maximum over a finite set
and over its convex hull. Conversely, for $A\in\mathcal K_w$ choose $\alpha(y,x)\in\arg\max_{\alpha\in A}\langle\alpha,f_{y,N}(x)\rangle$;
then the first sum equals $\sum_y\E h_A(f_{y,N})$ and $\max_{(y',x)}\langle\alpha(y',x),q_y\rangle\le h_A(q_y)$. Hence
\begin{equation}\label{eq:bodyform}
 \delta(\mathcal F,\mathcal E_N)=\max_{w\in\Delta(\Theta)}\ \sup_{A\in\mathcal K_w}\ \sum_y\Bigl[\E_{M_{y,N}}h_A(f_{y,N})-h_A(q_y)\Bigr].
\end{equation}

\emph{Normalisation.}  For $\alpha\in B_w$, $|\alpha|\le\sum_\vartheta\alpha_\vartheta\le1$, so $h_A$ is $1$-Lipschitz and
positively homogeneous. Since $\E\rho_{y,N}=1$, the summand of \eqref{eq:bodyform} equals
$\E[h_A(\rho q_y+\epsilon_NW)-\rho h_A(q_y)]$ with $\rho=\rho_{y,N}$, $W=W_{y,N}$, and its absolute value is at most
$\epsilon_N\E|W_{y,N}|$. For $t>0$ put $g_A(t,v)=[h_A(q_y+tv)-h_A(q_y)]/t$: by convexity of $t\mapsto h_A(q_y+tv)$ this is
nondecreasing in $t$, it decreases to $h_{F_A(q_y)}(v)$ as $t\downarrow0$ (the exposed-face derivative), and
$|g_A(t,v)|\le|v|$, $|h_{F_A(q_y)}(v)|\le|v|$. On $\{\rho>0\}$ the summand divided by $\epsilon_N$ equals
$g_A(\epsilon_N/\rho,W)$, by homogeneity; on $\{\rho=0\}$ it equals $h_A(W)\ge h_{F_A(q_y)}(W)$.

\emph{Lower bound.}  Fix $(w,A)$. By monotonicity in $t$, the summand divided by $\epsilon_N$ is at least
$h_{F_A(q_y)}(W_{y,N})$ on $\{\rho>0\}$ and on $\{\rho=0\}$ alike, so
$\delta(\mathcal F,\mathcal E_N)/\epsilon_N\ge\sum_y\E h_{F_A(q_y)}(W_{y,N})\to\sum_y\E h_{F_A(q_y)}(W_y)$ by weak convergence,
continuity of $h_{F_A(q_y)}$ and uniform integrability.

\emph{Upper bound.}  Along a subsequence realising the $\limsup$ of $\delta(\mathcal F,\mathcal E_N)/\epsilon_N$ choose
$(w_N,A_N)$ within $\epsilon_N^2$ of the maximum in \eqref{eq:bodyform}. The set $\{(w,A):w\in\Delta(\Theta),A\in\mathcal K_w\}$
is compact for the Euclidean metric on $w$ and the Hausdorff metric on $A$ (Blaschke selection theorem; the constraint
$A\subseteq B_w$ is closed), so along a further subsequence $w_N\to w$ and $A_N\to A\in\mathcal K_w$. Fix $\eta>0$. On the
event $\{\rho_{y,N}\ge\epsilon_N/\eta\}$ we have $\epsilon_N/\rho_{y,N}\le\eta$ and hence
$g_{A_N}(\epsilon_N/\rho_{y,N},W_{y,N})\le g_{A_N}(\eta,W_{y,N})$; the complementary event has probability tending to zero by
(A1), and the summand divided by $\epsilon_N$ is bounded by $|W_{y,N}|$ there, so its contribution tends to zero by uniform
integrability. Hausdorff continuity of support functions, $|h_{A_N}(v)-h_A(v)|\le d_H(A_N,A)|v|$, gives
$|g_{A_N}(\eta,v)-g_A(\eta,v)|\le d_H(A_N,A)(2|q_y|/\eta+|v|)$, whose expectation at $v=W_{y,N}$ tends to zero; and
$v\mapsto g_A(\eta,v)$ is continuous with $|g_A(\eta,v)|\le|v|$, so $\E g_A(\eta,W_{y,N})\to\E g_A(\eta,W_y)$. Therefore
$\limsup_N\delta(\mathcal F,\mathcal E_N)/\epsilon_N\le\sum_y\E g_A(\eta,W_y)$ for every $\eta>0$, and letting $\eta\downarrow0$
(dominated convergence with the majorant $|W_y|$, integrable by uniform integrability) gives
$\limsup_N\delta(\mathcal F,\mathcal E_N)/\epsilon_N\le\sum_y\E h_{F_A(q_y)}(W_y)$. Together with the lower bound, the limit
exists, equals the first maximum, and the maximum is attained at this $(w,A)$.

\emph{Polytope form.}  Given $(w,A)$ put $m_y=h_A(q_y)$ and $P=A(w,m)$, the body of Section~\ref{sec:tangent}; then
$A\subseteq P$, $h_P(q_y)=m_y$ since $A$ supplies a maximiser, so $F_A(q_y)\subseteq F_P(q_y)=\Pi_y(w,m)$ and
$h_{F_A(q_y)}\le h_{\Pi_y(w,m)}$. Conversely, if every $\Pi_y(w,m)$ is non-empty then $A=A(w,m)\in\mathcal K_w$ has
$h_A(q_y)=m_y$ and $F_A(q_y)=\Pi_y(w,m)$. Hence the two optimised values coincide (with $h_\emptyset=-\infty$ as in
\eqref{eq:tangent}).

\emph{The bundle model.}  Take $X_{y,N}=\mathcal B_N$, $M_{y,N}=M_N$, $f_{\vartheta,y,N}=a_{\vartheta,y}R_{\vartheta,y,N}$. The
weighted identity $a_{\vartheta,y}R_{\vartheta,y,N}=a_{\vartheta,y}(R_{\vartheta_y,y,N}+V_{\vartheta,y,N})$ is
\eqref{eq:abstractmodel} with $\rho_{y,N}=R_{\vartheta_y,y,N}$ and $W_{y,N}=N^{r/2}q_y\circ V_{\cdot,y,N}$. (A1): $\E_{M_N}\rho=1$,
and \eqref{eq:hoeffexp} applied to $g=g_{\nu}-1$ together with Lemma~\ref{lem:hoeffding}(d) gives
$\|R_{\nu,N}-1\|_{L^2(M_N)}=O(N^{-1/2})$ for every bundle law $\nu$, so $\rho_{y,N}\to1$ in probability. (A2):
\eqref{eq:jointcoupling} and $Z_N\Rightarrow Z$ give $W_{y,N}\Rightarrow q_y\circ H_y(Z)$, and the $L^2$ bound gives uniform
integrability.
\end{proof}

The proposition isolates what is model-specific in Theorem~\ref{thm:tangent}: the identification of the score vector
$q_y\circ H_y(Z_y)$ and of the rate $N^{-r/2}$ (Section~\ref{sec:bundle}), that is, the inputs (A1)--(A2). The passage from
these inputs to the optimised face program is a general first-order comparison statement at an informative coarse
experiment, in the spirit of Torgersen's local comparison theory (Section~\ref{sec:lit}).

\section{Consequences of the program}\label{sec:consequences}
\subsection{Kernels confined to their own fiber}
\begin{proposition}[Exact value of fiber-preserving reconstruction, two hypotheses]\label{prop:fp}
Let $\Theta=\{0,1\}$, $q_y=(a_y,b_y)$, and let $\Delta^{\mathrm{fp}}_N$ be the minimum of $\max_h\TV(KQ_h,P_h)$ over kernels
with $K(y,\cdot)$ supported in fiber $y$ for every $y$. Then, for every $N$,
\begin{equation}\label{eq:fp}
 \Delta^{\mathrm{fp}}_N=\max_{u\in[0,1]}\sum_y\min\{u\,a_y,(1-u)\,b_y\}\,\TV(P_{0,y,N},P_{1,y,N}),
\end{equation}
and $N^{r/2}\Delta^{\mathrm{fp}}_N\to\max_u\sum_y\min\{ua_y,(1-u)b_y\}\,\E|H_y|/2=:\mathcal V^{\mathrm{fp}}$, where $H_y$ is the
polynomial of the pair $(\nu_{0,y},\nu_{1,y})$, with $H_y:=0$ if the pair agrees, if it differs only above degree $r$, or if
$a_yb_y=0$ (such fibers are reconstructed exactly from the hypothesis that charges them, or carry no mass). Let
$\tau_y=b_y/a_y$ for the fibers with $a_yb_y>0$ and $G_\tau=\sum_{y:\tau_y=\tau}b_y\E|H_y|/2$. Then
\begin{equation}\label{eq:fpcompare}
 \mathcal V^{\mathrm{fp}}\ \ge\ \mathcal V=\max_\tau\frac{G_\tau}{1+\tau},
\end{equation}
with strict inequality as soon as two distinct ratios $\tau\ne\tau'$ in $(0,\infty)$ have $G_\tau>0$ and $G_{\tau'}>0$.
\end{proposition}
\begin{proof}
A fiber-preserving kernel is a family of laws $R_y$ on the fibers, with error $\sum_ya_{h,y}\TV(R_y,P_{h,y,N})$ under $h$,
where $a_{0,y}=a_y$, $a_{1,y}=b_y$. The minimax over $(R_y)$ equals, by the finite-dimensional minimax theorem with a prior
$(u,1-u)$ on $\{0,1\}$, $\max_u\sum_y\min_{R_y}\{ua_y\TV(R_y,P_{0,y,N})+(1-u)b_y\TV(R_y,P_{1,y,N})\}$, and
Lemma~\ref{lem:tvfacts} evaluates the inner minima. The limit follows from Lemma~\ref{lem:coupling}. The value of
$\mathcal V$ for two hypotheses is computed in Corollary~\ref{cor:binary} below. Since $\min\{ua_y,(1-u)b_y\}=b_y\min\{u/\tau_y,1-u\}$
for $a_yb_y>0$, and fibers with $a_yb_y=0$ contribute nothing to either side,
$\mathcal V^{\mathrm{fp}}=\max_u\sum_\tau\min\{u/\tau,1-u\}\,G_\tau$, while
$\max_\tau G_\tau/(1+\tau)=\max_\tau\max_u\min\{u/\tau,1-u\}\,G_\tau$ (the inner maximum is at $u=\tau/(1+\tau)$). So
$\mathcal V^{\mathrm{fp}}\ge\mathcal V$, and if the maximising $\tau$ is attained at $u^*=\tau/(1+\tau)\in(0,1)$ while a second
ratio $\tau'$ has $G_{\tau'}>0$, the term $\min\{u^*/\tau',1-u^*\}G_{\tau'}$ is strictly positive.
\end{proof}
The program \eqref{eq:tangent} with the cross constraints $\langle\alpha,q_{y'}\rangle\le m_{y'}$, $y'\ne y$, deleted is the
tangent value of fiber-preserving kernels for any finite $\Theta$: the proof of Theorem~\ref{thm:tangent} goes through with
the maximum in \eqref{eq:identity} taken over the own fiber only; we omit the repeated verification. Thus cross-fiber transport improves the leading
coefficient exactly when the optimal value with the cross constraints is strictly smaller than without them; for two
hypotheses \eqref{eq:fpcompare} is the explicit criterion. (Activity of a cross constraint at an optimum is not the right
notion: a binding constraint can be redundant.)

\begin{example}\label{ex:fpgap}
$q_1=(\frac12,\frac34)$, $q_2=(\frac12,\frac14)$, $d=2$, $L=1$, $p=\frac12$, and in both fibers $y=1,2$ let
$\nu_{0,y}=\mathrm{Ber}(\frac12)$ and $\nu_{1,y}=\mathrm{Ber}(\frac45)$.
Thus $\E|H_y|/2=0.6\,\varphi(0)$ for both fibers (Corollary~\ref{cor:firstorder}). Then
$\mathcal V=\frac34\cdot0.6\,\varphi(0)/\frac52=\frac9{50}\varphi(0)$ and
$\mathcal V^{\mathrm{fp}}=0.6\,\varphi(0)\max_u[\min\{\frac u2,\frac{3(1-u)}4\}+\min\{\frac u2,\frac{1-u}4\}]=0.6\cdot0.4\,\varphi(0)=\frac6{25}\varphi(0)$,
a ratio of $4/3$. The finite-experiment programs of Section~\ref{sec:num} confirm both limits.
\end{example}

\subsection{Collinear scores}
Suppose that in every fiber $H_{\vartheta,y}=s_{\vartheta,y}G_y$ with a centred square-integrable $G_y$ and constants
$s_{\vartheta,y}$ ($s_{\vartheta_y,y}=0$). This holds whenever $d=2$, for every $r$, because the degree-$r$ chaos of one
Gaussian variable is spanned by $\mathrm{He}_r$; and it holds in the rare channel of Section~\ref{sec:rare}. Put
$c_y=\E(G_y)_+=\frac12\E|G_y|$ and $\ell_y(\alpha)=\sum_\vartheta\alpha_\vartheta a_{\vartheta,y}s_{\vartheta,y}$. For a polytope
$\Pi$, $\E h_\Pi(q_y\circ H_y)=\E[(G_y)_+]\max_\Pi\ell_y-\E[(G_y)_-]\min_\Pi\ell_y=c_y\ext_\Pi(\ell_y)$ with
$\ext_\Pi(\ell)=\max_\Pi\ell-\min_\Pi\ell$. Hence
\begin{equation}\label{eq:collinearLP}
 \mathcal V=\max\Bigl\{\sum_yc_y\bigl(\ell_y(\alpha^+_y)-\ell_y(\alpha^-_y)\bigr):\ w\in\Delta(\Theta),\ m\in\R^{\mathcal Y},\
 \alpha^\pm_y\in\Pi_y(w,m)\Bigr\},
\end{equation}
a linear program with $|\Theta|+|\mathcal Y|+2|\Theta||\mathcal Y|$ variables. Adding a multiple of
$\langle\cdot,q_y\rangle$ to $\ell_y$ does not change $\ext_{\Pi_y}(\ell_y)$, so the choice of reference $\vartheta_y$ is
immaterial.

\begin{corollary}[Two hypotheses]\label{cor:binary}
Let $\Theta=\{0,1\}$, $q_y=(a_y,b_y)$, $\tau_y=b_y/a_y$ for $a_yb_y>0$, and $G_\tau=\sum_{y:\tau_y=\tau}b_y\E|H_y|/2$ as above.
Then $\mathcal V=\max_\tau G_\tau/(1+\tau)$, the maximum over the finitely many distinct ratios. Hence
$N^{r/2}\Delta_N\to\max_\tau G_\tau/(1+\tau)$.
\end{corollary}
\begin{proof}
Take $\vartheta_y=0$ when $a_y>0$. A fiber with $b_y=0$ has $q_y\circ H_y=(0,0)$; a fiber with $a_y=0$ has reference $1$ and
again $q_y\circ H_y=0$; both contribute nothing. For a fiber with $a_yb_y>0$, $q_y\circ H_y=(0,b_yH_y)$ and the contribution is
$b_y\cdot\frac12\E|H_y|\cdot\mathrm{width}_1(F_y)$, where $\mathrm{width}_1(F)$ is the extent of the coordinate $\alpha_1$ on
the face $F_y$ of the body $A\subseteq[0,w_0]\times[0,w_1]$ exposed by the direction $(1,\tau_y)$. Fibers with the same
ratio expose the same face. Distinct ratios $\tau$ expose faces of $A$ (vertices or edges) whose relative interiors are
disjoint, an edge with outer normal $(1,\tau)$ has slope $-1/\tau$, so its $\alpha_0$-extent is $\tau$ times its
$\alpha_1$-extent, and the extents of distinct edges are disjoint subintervals of $[0,w_1]$ and $[0,w_0]$. Writing
$\omega_\tau$ for the $\alpha_1$-extent of the face exposed by $(1,\tau)$,
\[
 \sum_\tau\omega_\tau\le w_1,\qquad\sum_\tau\tau\,\omega_\tau\le w_0,\qquad
 \text{value}=\sum_\tau\omega_\tau G_\tau\le\max_\tau\frac{G_\tau}{1+\tau}\sum_\tau(1+\tau)\omega_\tau\le\max_\tau\frac{G_\tau}{1+\tau}.
\]
Equality is attained by $w=(\tau^*,1)/(1+\tau^*)$ and the body $A=\{\alpha\in B_w:\alpha_0+\tau^*\alpha_1\le\tau^*w_1\}$, whose
face exposed by $(1,\tau^*)$ is the whole diagonal edge, of $\alpha_1$-extent $w_1$; all other directions expose a vertex.
\end{proof}
This is the hockey-stick knot formula of the binary theory (the maximum over the knots $t=\tau$ of
$[D_t(P_1,P_0)-D_t(Q_1,Q_0)]/(1+t)$, \cite{Torgersen1970,Jencova2012}, with the gain of each knot evaluated by
Lemma~\ref{lem:coupling}); the proof above derives it from the program.

\begin{corollary}[A single fiber]\label{cor:single}
If $|\mathcal Y|=1$ and the scores are collinear, $H_\vartheta=s_\vartheta G$, then
$\mathcal V=\frac14(\max_\vartheta s_\vartheta-\min_\vartheta s_\vartheta)\E|G|$. In particular a one-parameter local family
without coarse information has no incompatibility: its value is the largest pairwise value.
\end{corollary}
\begin{proof}
Here $q=(1,\dots,1)$ and $\Pi(w,m)=\{\alpha\in B_w:\sum_\vartheta\alpha_\vartheta=m\}$. For $\alpha,\alpha'\in\Pi$ the vector
$\delta=\alpha-\alpha'$ has $\sum_\vartheta\delta_\vartheta=0$ and $|\delta_\vartheta|\le w_\vartheta$, so with
$\bar s=(\max s+\min s)/2$, $\ell(\alpha)-\ell(\alpha')=\sum_\vartheta(s_\vartheta-\bar s)\delta_\vartheta\le\frac12(\max s-\min s)\sum_\vartheta w_\vartheta$.
Thus $\ext_\Pi(\ell)\le\frac12(\max s-\min s)$, with equality for $w$ equal to $\frac12$ on a maximiser and a minimiser of $s$
and $m=\frac12$. Multiply by $c=\frac12\E|G|$.
\end{proof}

\subsection{An exact three-hypothesis example}
\begin{proposition}[Ratio $4/3$]\label{prop:fourthirds}
Let $d=2$, $L=1$, $p=(\frac12,\frac12)$, $\mathcal Y=\{1,2\}$, $\Theta=\{0,1,2\}$, coarse vectors
$q_1=(\frac12,\frac34,\frac25)$, $q_2=(\frac12,\frac14,\frac35)$, and bundle laws $\nu_{\vartheta,y}=\mathrm{Ber}(\frac12)$ except
$\nu_{1,1}=\mathrm{Ber}(0.8)$ and $\nu_{2,2}=\mathrm{Ber}(0.2)$ (mass on the first symbol). Then
\[
 \lim_{N\to\infty}\sqrt N\,\Delta_N=\frac6{25}\,\varphi(0),
\]
while the three pairs $\{0,1\},\{0,2\},\{1,2\}$ have limits $\frac9{50}\varphi(0)$, $\frac9{55}\varphi(0)$, $\frac{18}{115}\varphi(0)$.
The family coefficient exceeds the largest pairwise coefficient by the factor $4/3$ exactly.
\end{proposition}
\begin{proof}
By Corollary~\ref{cor:firstorder}, $r=1$, $H_{1,1}=\frac35Z_1$, $H_{2,2}=-\frac35Z_2$ and all other $H_{\vartheta,y}=0$; the scores
are collinear with $G_y=Z_y$, $c_y=\varphi(0)$, $s_1=(0,\frac35,0)$, $s_2=(0,0,-\frac35)$. Write
$A=\alpha^+_1$, $B=\alpha^-_1$, $C=\alpha^+_2$, $D=\alpha^-_2$ (coordinates indexed by $0,1,2$); the objective of
\eqref{eq:collinearLP} divided by $\varphi(0)$ is $J=\frac9{20}(A_1-B_1)+\frac9{25}(D_2-C_2)$. The point
$w=(\frac25,\frac4{15},\frac13)$, $m_1=m_2=\frac15$, $A=(0,\frac4{15},0)$, $B=(\frac25,0,0)$, $C=(\frac25,0,0)$,
$D=(0,0,\frac13)$ is feasible (each test lies in the box, $\langle A,q_1\rangle=\langle B,q_1\rangle=\frac15=m_1$,
$\langle C,q_2\rangle=\langle D,q_2\rangle=\frac15=m_2$, $\langle A,q_2\rangle=\frac1{15}$, $\langle B,q_2\rangle=\frac15$,
$\langle C,q_1\rangle=\frac15$, $\langle D,q_1\rangle=\frac2{15}$) and has $J=\frac9{20}\cdot\frac4{15}+\frac9{25}\cdot\frac13=\frac6{25}$.
Conversely, for every feasible point the following identity holds modulo the equalities $\sum_jw_j=1$,
$\langle A,q_1\rangle=\langle B,q_1\rangle=m_1$, $\langle C,q_2\rangle=\langle D,q_2\rangle=m_2$:
\begin{align*}
\frac6{25}-J={}&\frac6{25}(w_1-A_1)+\frac{13}{175}(w_0-B_0)+\frac{52}{175}\bigl(m_2-\langle B,q_2\rangle\bigr)\\
&+\frac{29}{175}(w_0-C_0)+\frac{29}{175}\bigl(m_1-\langle C,q_1\rangle\bigr)+\frac6{25}(w_2-D_2)\\
&+\frac7{50}A_0+\frac{14}{125}A_2+\frac{16}{125}C_2+\frac{19}{100}B_1+\frac1{10}D_0+\frac1{20}D_1 ;
\end{align*}
precisely, $\frac6{25}-J$ minus the right-hand side equals
$-\frac6{25}(\sum_jw_j-1)-\frac7{25}(\langle A,q_1\rangle-m_1)+\frac{78}{175}(\langle B,q_1\rangle-m_1)
+\frac{87}{175}(\langle C,q_2\rangle-m_2)-\frac15(\langle D,q_2\rangle-m_2)$ identically in all variables, a coefficientwise
rational identity. Every term on the
right-hand side is a box slack, a cross-constraint slack or a non-negative variable, so $J\le\frac6{25}$ at every feasible
point. Hence $\mathcal V=\frac6{25}\varphi(0)$ and Theorem~\ref{thm:tangent} gives the limit. The pairwise values follow
from Corollary~\ref{cor:binary}: for $\{0,1\}$ only fiber $1$ carries a score, $\tau_1=\frac32$, value
$\frac34\cdot\frac35\varphi(0)/\frac52=\frac9{50}\varphi(0)$; for $\{0,2\}$ only fiber $2$, $\tau_2=\frac65$, value
$\frac35\cdot\frac35\varphi(0)/\frac{11}5=\frac9{55}\varphi(0)$; for $\{1,2\}$ fiber $1$ has ratio $\frac{2/5}{3/4}=\frac8{15}$ and score
difference $\frac35$, fiber $2$ has ratio $\frac{12}5$ and score difference $\frac35$, giving
$\max\{\frac25\cdot\frac35/\frac{23}{15},\ \frac35\cdot\frac35/\frac{17}5\}\varphi(0)=\frac{18}{115}\varphi(0)$.
\end{proof}
The example is a variant of Example~\ref{ex:fpgap}, with a modified coarse row and the two native changes, one in each
direction from the background law, assigned to different hypotheses. The equality of its family coefficient with the
fiber-preserving coefficient of the binary example holds already at every finite sample size.

\begin{corollary}[A finite-sample identity]\label{cor:finiteidentity}
Let $\Delta_N^{(3)}$ be the loss in Proposition~\ref{prop:fourthirds}, and let $\Delta_N^{\mathrm{fp},(2)}$ be the
fiber-preserving loss in Example~\ref{ex:fpgap}. For every integer $N\ge1$,
\begin{equation}\label{eq:finiteidentity}
 \Delta_N^{(3)}=\Delta_N^{\mathrm{fp},(2)}
 =\frac25\TV(P_N,Q_N)
 =\frac3{25}\,2^{-(N-1)}\binom{N-1}{\lfloor(N-1)/2\rfloor},
\end{equation}
where $P_N=\Bin(N,\frac12)$ and $Q_N=\Law(S+U)$ for independent $S\sim\Bin(N-1,\frac12)$ and
$U\sim\mathrm{Ber}(\frac45)$.
\end{corollary}
\begin{proof}
Write $P=P_N$, $D=Q_N-P$, $f=D_+$, $g=D_-$ and $t=\sum_xf(x)=\sum_xg(x)=\TV(P_N,Q_N)$. The law with native parameter
$\frac15$ is $P-D$, so the three full laws, written as pairs of subprobability laws on the two output fibers, are
\[
 P_0=(\tfrac12P,\tfrac12P),\qquad
 P_1=(\tfrac34(P+D),\tfrac14P),\qquad
 P_2=(\tfrac25P,\tfrac35(P-D)).
\]
For the lower bound use $w=(\frac25,\frac4{15},\frac13)$. In fiber $1$ take the test vector
$(0,\frac4{15},0)$ when $D(x)\ge0$ and $(\frac25,0,0)$ otherwise; in fiber $2$ take $(\frac25,0,0)$ when $D(x)\ge0$ and
$(0,0,\frac13)$ otherwise. Both coarse maxima in \eqref{eq:identity} are $\frac15$, and the baseline
contribution is $\frac25$. The remaining contribution is
$\frac34\cdot\frac4{15}t+\frac35\cdot\frac13t=\frac25t$, giving $\Delta_N^{(3)}\ge\frac25t$.

For the upper bound the following kernel sends an input tag (row) to a law on the two output fibers (columns):
\[
 K=\begin{pmatrix}
 P+\frac7{15}f-\frac45g & \frac13f\\[2pt]
 \frac23g & P-f+\frac13g
 \end{pmatrix}.
\]
Since $P+D$ and $P-D$ are probability laws, $0\le f,g\le P$; hence all four entries are non-negative.
Each row has total mass one because $\sum f=\sum g=t$. Direct substitution gives the signed residuals
\[
\begin{array}{c|cc}
 \vartheta & \text{fiber }1 & \text{fiber }2\\ \hline
 0 & -\frac7{30}f+\frac1{15}g & \frac13f-\frac16g\\[2pt]
 1 & \frac25f-\frac{19}{60}g & -\frac1{12}g\\[2pt]
 2 & -\frac{14}{75}f-\frac2{25}g & -\frac2{15}f+\frac25g
\end{array}
\quad\text{for }P_\vartheta-KQ_\vartheta.
\]
The supports of $f$ and $g$ are disjoint, and every row has total positive mass $\frac25t$; thus all three TV errors
are exactly $\frac25t$. This proves the first loss formula. Proposition~\ref{prop:fp} gives
$\Delta_N^{\mathrm{fp},(2)}=t\max_{u\in[0,1]}[\min\{u/2,3(1-u)/4\}+\min\{u/2,(1-u)/4\}]=\frac25t$.
Finally, if $b_k=\Pr\{S=k\}$, extended by zero outside $\{0,\ldots,N-1\}$, then
$D(k)=\frac3{10}(b_{k-1}-b_k)$. Unimodality of $(b_k)$ gives
$t=\frac3{10}\max_kb_k$, proving the last equality in \eqref{eq:finiteidentity}.
\end{proof}

\section{The critical rare-event channel along a window of compositions}\label{sec:rare}
\subsection{Model and exact structure}
Let $p=(p_1,\dots,p_d)$ be strictly positive and $q=(q_1,\dots,q_d)$ arbitrary probability vectors on the native symbols,
adjoin a symbol $c$, fix $\lambda>0$ and the two local randomisers
\[
 W_0=(p,0),\qquad W_{1,n}=\bigl(\lambda q/n,\ 1-\lambda/n\bigr)\qquad(n>\lambda),
\]
i.e.\ a type-$1$ user sends $c$ with probability $1-\lambda/n$ and otherwise a native symbol drawn from $q$. Let $k_n$ be
integers with $0\le k_n<n$ and $k_n/n\to\theta\in(0,1)$, $\alpha_n=\lambda/n$, $m_n=n-k_n$, $\mu=\theta\lambda$,
\[
 h_i=\frac{q_i}{p_i}-1,\qquad I=\chi^2(q\|p)=\sum_ip_ih_i^2,\qquad\kappa^2=\frac I{1-\theta}.
\]
Fix $m\ge1$. Throughout this section $n$ is large enough that $n>\lambda$ and $m_n>m$, which holds eventually because
$\theta<1$; ``uniformly in $n$'' refers to this range. Composition $j\in\{0,\dots,m\}$ has $k_n+j$ users of type $1$ and
$m_n-j\ge1$ of type $0$; let $P^{(j)}_n$ be the law of the histogram of the $n$ messages and $R_n=k_n-N(c)$ the retained
statistic, where $N(c)$ is the number of $c$-messages. Put
\[
 \Delta^{(m)}_n=\delta\bigl((P^{(j),R}_n)_{j\le m},(P^{(j)}_n)_{j\le m}\bigr),
\]
the optimal common reconstruction error for the window. Under composition $j$ the statistic takes the values $r\ge-j$ with probabilities
$q^{(n)}_{r,j}=\binom{k_n+j}{r+j}\alpha_n^{r+j}(1-\alpha_n)^{k_n-r}$, and given $R_n=r$ the native counts have the law
\begin{equation}\label{eq:condlaw}
 f^{(j)}_{n,r}=\Mult(m_n-j,p)*\Mult(r+j,q)
\end{equation}
(exactly $r+j$ type-$1$ users sent natives, drawn from $q$, and $m_n-j$ type-$0$ users sent natives from $p$). The
likelihood-ratio vector of fiber $r\ge0$ is exact and simple:
\begin{equation}\label{eq:lrvector}
 \frac{q^{(n)}_{r,j}}{q^{(n)}_{r,0}}=\prod_{i=1}^{j}\frac{\mu_{n,i}}{r+i},\qquad\mu_{n,i}=(k_n+i)\alpha_n\to\mu .
\end{equation}
In the limit the coarse vectors are Poisson, $q_r=(\pi_{r+j})_{j\le m}$ with $\pi_k=e^{-\mu}\mu^k/k!$ and $\pi_k=0$ for
$k<0$, indexed by $r\ge-m$; the fibers $-m\le r<0$ are charged only by the compositions $j\ge-r$.

\begin{lemma}[Scores of a $j$-fold replacement]\label{lem:jscore}
For $r\ge0$ let $L^{(j)}_{n,r}=df^{(j)}_{n,r}/df^{(0)}_{n,r}$. Then $f^{(j)}_{n,r}\ll f^{(0)}_{n,r}$,
$\E_{f^{(0)}_{n,r}}L^{(j)}_{n,r}=1$, and uniformly in $r\ge0$ and in $n$ with $m_n>m$,
\begin{equation}\label{eq:chi2uniform}
 \E_{f^{(0)}_{n,r}}\bigl(L^{(j)}_{n,r}-1\bigr)^2\le\frac{j^2I}{m_n-j}+\frac{C_j}{m_n^2}.
\end{equation}
If $q=p$ then $L^{(j)}_{n,r}\equiv1$ for all $j,r,n$. If $q\ne p$ (so $\kappa>0$), then for fixed $r$, with $Y$ the observed
native counts in the fiber and the observable statistic
\[
 Z_n(Y)=\frac{\sqrt n}{\kappa\,m_n}\sum_ih_i\,(Y_i-m_np_i),
\]
one has, jointly for $j=1,\dots,m$,
\begin{equation}\label{eq:scorecoupling}
 \E_{f^{(0)}_{n,r}}\bigl|\sqrt n\,(L^{(j)}_{n,r}-1)-j\kappa Z_n(Y)\bigr|\longrightarrow0,\qquad
 Z_n(Y)\Rightarrow N(0,1)\ \text{under }f^{(0)}_{n,r}.
\end{equation}
The same statements hold in a fiber $r=-i<0$ for the compositions $j\ge i$ relative to composition $i$, with $j-i$ in
place of $j$ and $m_n-i$ in place of $m_n$.
\end{lemma}
\begin{proof}
Realise $f^{(0)}_{n,0}$ by $m_n$ iid labels $Y_1,\dots,Y_{m_n}\sim p$ and $f^{(j)}_{n,0}$ by replacing $j$ fixed labels by iid
draws from $q$. Given the histogram, the labels are a uniform arrangement, so exactly as in Lemma~\ref{lem:ustat},
\[
 L^{(j)}_{n,0}=\binom{m_n}{j}^{-1}\sum_{|I|=j}\prod_{i\in I}\frac{q_{Y_i}}{p_{Y_i}}
 =\binom{m_n}{j}^{-1}\sum_{|I|=j}\prod_{i\in I}\bigl(1+h(Y_i)\bigr)
 =1+\sum_{k=1}^j\frac{(j)_k}{(m_n)_k}\sum_{|S|=k}\prod_{i\in S}h(Y_i),
\]
since a $k$-subset $S$ is contained in $\binom{m_n-k}{j-k}$ of the $j$-subsets. The $k=1$ term is
$\frac j{m_n}\sum_ih(Y_i)=\frac j{m_n}\sum_ih_i(X_i-m_np_i)$ (as $\sum_ip_ih_i=0$), of variance $j^2I/m_n$; the terms with
$k\ge2$ are degenerate $U$-statistics of distinct orders, mutually orthogonal and orthogonal to the first, with second
moments $\frac{(j)_k^2}{(m_n)_k^2}\binom{m_n}{k}I^k=O(m_n^{-k})$. Both laws in \eqref{eq:condlaw} are obtained from their
$r=0$ versions by adding the same independent $T_r\sim\Mult(r,q)$, so with $Y=X+T_r$,
\begin{equation}\label{eq:rao}
 L^{(j)}_{n,r}(Y)=\E\bigl[L^{(j)}_{n,0}(X)\,\big|\,Y\bigr],
\end{equation}
which gives absolute continuity, $\E L=1$, and by conditional Jensen the uniform bound \eqref{eq:chi2uniform}. If $q=p$
then $h\equiv0$ and $L^{(j)}_{n,0}\equiv1$, hence $L^{(j)}_{n,r}\equiv1$. Let $q\ne p$. For fixed
$r$, $|\E[\sum_ih_iX_i\mid Y]-\sum_ih_iY_i|\le|\E[\sum_ih_iX_i\mid Y]-\sum_ih_iX_i|+|\sum_ih_iT_{r,i}|\le4r\max_i|h_i|$, so
$\sqrt n(L^{(j)}_{n,r}-1)$ differs from $\frac{j\sqrt n}{m_n}\sum_ih_i(Y_i-m_np_i)=j\kappa Z_n(Y)$ by $O(\sqrt n/m_n)$ plus the
conditional expectation of the higher-order terms, whose $L^2$ norm is $O(\sqrt n/m_n)$. The variance of
$\frac{\sqrt n}{m_n}\sum_ih_i(X_i-m_np_i)$ is $nI/m_n\to\kappa^2$, and the central limit theorem for the bounded centred
summands together with $|Z_n(Y)-Z_n(X)|\le2r\max|h_i|\sqrt n/(\kappa m_n)\to0$ gives $Z_n(Y)\Rightarrow N(0,1)$. For
negative fibers replace composition $0$ by composition $i=-r$ and $m_n$ by $m_n-i$.
\end{proof}

Fibers $r\ge0$ thus have collinear scores $H_{j,r}=j\kappa Z$ relative to composition $0$, and a fiber $r=-i<0$ has scores
$(j-i)\kappa Z$ for $j\ge i$ relative to composition $i$. Let $\mathcal V_m(\mu)$ be the value of the collinear program
\eqref{eq:collinearLP} with fibers $r\ge-m$, coarse vectors $q_r=(\pi_{r+j})_j$, $G_r=Z$, and the common factors $\kappa$ and
$c_r=\varphi(0)$ divided out:
\begin{equation}\label{eq:Vm}
\begin{aligned}
 \mathcal V_m(\mu)&=\max\Bigl\{\sum_{r\ge-m}\bigl[\ell_r(\alpha^+_r)-\ell_r(\alpha^-_r)\bigr]:\ w\in\Delta,\ m\in\R^{\{r\ge-m\}},\
 \alpha^\pm_r\in\Pi_r(w,m)\Bigr\},\\
 \ell_r(\alpha)&=\sum_{j=0}^mj\,\pi_{r+j}\alpha_j .
\end{aligned}
\end{equation}
(The reference composition of a negative fiber is immaterial, as noted after \eqref{eq:collinearLP}.)  For a truncation
level $R\ge0$ let $\mathcal V^{(R)}_m(\mu)$ be the same program with the fibers $r>R$ deleted. Restricting $w$ to a subset
$J$ of compositions gives the program of the sub-family (the box forces $\alpha_j=0$ for $j\notin J$), and everything
below applies to sub-families verbatim.

The maximum in \eqref{eq:Vm} is attained. Indeed, the prior simplex and the coordinate boxes for the tests are compact,
and the face levels satisfy $0\le m_r\le\sum_j\pi_{r+j}$. All constraints are closed in the resulting countable product.
The objective series converges uniformly there, since the absolute gain of fiber $r$ is at most
$\sum_jj\pi_{r+j}$ and these bounds are summable. Compactness and continuity therefore give a maximiser.

\begin{theorem}[Window of compositions]\label{thm:rare}
Assume $q\ne p$. For every $m\ge1$,
\[
 \lim_{n\to\infty}\sqrt n\,\Delta^{(m)}_n=\kappa\,\varphi(0)\,\mathcal V_m(\mu),\qquad
 \mathcal V_m^{(R)}(\mu)\le\mathcal V_m(\mu)\le\mathcal V_m^{(R)}(\mu)+\sum_{r>R}\sum_{j=1}^mj\,\pi_{r+j}\quad(R\ge0).
\]
For $m=1$, $\mathcal V_1(\mu)=\max_{r\ge0}\pi_{r+1}/(1+\mu/(r+1))$. If $q=p$ then $\Delta^{(m)}_n=0$ for every $n$.
\end{theorem}
\begin{proof}
\emph{Truncation.}  A feasible point of the full program, restricted to the fibers $r\le R$, is feasible for the truncated
program and loses only the gains of the deleted fibers, each at most $\ext_{B_w}(\ell_r)\le\sum_jj\pi_{r+j}$. Conversely a
feasible point of the truncated program extends to the full one: put $m_{r'}=\max\langle\alpha,q_{r'}\rangle$ over the
finitely many tests in use and let $\alpha^\pm_{r'}$ be a maximiser, for every $r'>R$; the extension is feasible with the
same value. The same two arguments show that the program with fibers $r\ge-m$ has a finite value and that
$\mathcal V^{(R)}_m\uparrow\mathcal V_m$.

\emph{Lower bound.}  Fix $R$ and let $\mathcal V^{(R)}_{m,n}$ be the value of the truncated program with the exact finite-$n$
coarse vectors $(q^{(n)}_{r,j})_j$ in place of the Poisson ones, attained at $(w,m,\alpha^\pm)$. Define finite-$n$ tests on
fiber $r\le R$ by $\alpha(r,x)=\alpha^+_r$ if $Z_n(x)>0$ and $\alpha^-_r$ otherwise, and on fibers $r>R$ by the extension of the
truncation step. As in the lower bound of Theorem~\ref{thm:tangent} these tests have zero coarse defect in
\eqref{eq:identity}, and their value is
\[
 \sum_{r\le R}\sum_jq^{(n)}_{r,j}\Bigl[(\alpha^+_{r,j}-\alpha^-_{r,j})\,\E_{f^{(\mathrm{ref})}_{n,r}}\bigl[(L^{(j)}_{n,r}-1)\1\{Z_n>0\}\bigr]
 +\alpha^-_{r,j}\,\E_{f^{(\mathrm{ref})}_{n,r}}\bigl[L^{(j)}_{n,r}-1\bigr]\Bigr],
\]
where the last expectation vanishes, and by \eqref{eq:scorecoupling} the first equals
$n^{-1/2}\kappa\,\varphi(0)\,(j-j_{\mathrm{ref}})(1+o(1))$. Hence $\sqrt n\,\Delta^{(m)}_n\ge\kappa\varphi(0)\mathcal V^{(R)}_{m,n}-o(1)$.
Finally $\mathcal V^{(R)}_{m,n}\to\mathcal V^{(R)}_m(\mu)$. We give the finite-program continuity argument explicitly.
Write $S=\{-m,\ldots,R\}$ and $d=m+1$. The limiting vectors $q_r$ are nonzero and pairwise non-proportional:
negative fibers have distinct supports, and for $r\ge0$ one has $q_{r,1}/q_{r,0}=\mu/(r+1)$.
Put $w^0_j=1/d$, $\gamma=w^0/2$, choose $0<\rho<1/(2d)$, and set
$b_r=\gamma+\rho q_r/|q_r|$, $t^0_r=\langle b_r,q_r\rangle$. Taking both tests equal to $b_r$ gives strict box
inequalities and, for $s\ne r$, the positive cross slack
\[
 t^0_s-\langle b_r,q_s\rangle
 =\rho|q_s|\left(1-\frac{\langle q_r,q_s\rangle}{|q_r||q_s|}\right)>0.
\]
Mix an optimal point with this strict point using a fixed weight $\eta>0$, and denote its tests by
$\widetilde\alpha^\pm_r$. For perturbed vectors $q'_r$ retain the mixed prior and the tests
$\widetilde\alpha^+_r$, and put
\[
 t'_r=\langle\widetilde\alpha^+_r,q'_r\rangle,\qquad
 \alpha'^-_r=\widetilde\alpha^-_r+
 \frac{\langle\widetilde\alpha^+_r-\widetilde\alpha^-_r,q'_r\rangle}{|q'_r|^2}\,q'_r.
\]
The correction is $O(\max_r|q'_r-q_r|)$ because its numerator vanishes at $q'_r=q_r$. It restores both face equalities
exactly, and all box and cross inequalities survive sufficiently small perturbations by the strict margins. Letting the
perturbation tend to zero and then $\eta\downarrow0$ proves lower semicontinuity of the optimum. Upper semicontinuity
follows by a convergent subsequence of optimal points: priors and tests lie in a common compact set, levels are bounded,
and the constraints are closed. The objective coefficients also converge. Apply this argument to
$q'_r=(q^{(n)}_{r,j})_j$; no individual coordinate is required to be positive.
For a sub-family first delete zero and singleton-support rows, which have zero gain and extend back by the maximiser
construction of the truncation step. The remaining rays are distinct because, for any two active indices $j_0<j_1$,
the ratio $\pi_{r+j_1}/\pi_{r+j_0}$ is strictly decreasing in $r$. If no rows remain the value is zero.
Let $R\to\infty$.

\emph{Upper bound.}  Apply the upper-bound argument of Theorem~\ref{thm:tangent} to the fibers $r\le R$ only, with the
cross constraints coming from the fibers $r\le R$ only (deleting constraints enlarges the polytopes). The fibers $r>R$
contribute to the second sum of \eqref{eq:split} at most
$\sum_{r>R}\sum_jq^{(n)}_{r,j}\E|L^{(j)}_{n,r}-1|\le Cn^{-1/2}\sum_{r>R}\sum_jq^{(n)}_{r,j}$ by \eqref{eq:chi2uniform}, a tail that
is small uniformly in $n$ for large $R$ by Markov's inequality for $\Bin(k_n+j,\alpha_n)$. Lemma~\ref{lem:central} is
replaced by the trivial fact that in fiber $r\le R$ the reference density is $1$ (the reference law is the law itself), so
the defect argument reads $\E\varepsilon^N\le Cn^{-1/2}$ directly. The limit of the subsequence argument is bounded by
$\mathcal V^{(R)}_m(\mu)\le\mathcal V_m(\mu)$ after Lemma~\ref{lem:jscore} replaces Lemma~\ref{lem:coupling}, and the finite-$n$
coarse vectors converge to the Poisson ones. The case $m=1$ is Corollary~\ref{cor:binary} with $\tau_r=\pi_{r+1}/\pi_r=\mu/(r+1)$
and $G_{\tau_r}=\pi_{r+1}$ (in the normalised units). If $q=p$ then $h\equiv0$ and $L^{(j)}_{n,r}\equiv1$, so the common
conditional law given $R_n$ is an exact common kernel.
\end{proof}

\begin{remark}
The function $\mathcal V_m$ is continuous on $(0,\infty)$. Its finite truncations are continuous by the argument above,
and for $\mu\in[a,b]\subset(0,\infty)$ the tail in Theorem~\ref{thm:rare} is at most
$\frac{m(m+1)}2\,b/(R+1)$ by Markov's inequality for a Poisson variable of mean $\mu$; thus convergence of the
truncations is locally uniform. No rate is imposed on $k_n/n\to\theta$: the lower bound uses the exact finite-$n$
coarse vectors, so in particular composition drift on the $n^{-1/2}$ scale does not change the limit coefficient.
\end{remark}

\subsection{Three consecutive compositions}\label{sec:three}
Let $m=2$ and write
\[
 \sigma(\mu)=\frac{e^{-\mu}\mu(\mu+1)}{\mu^2+\mu+1},\qquad
 p_1(\mu)=\max_{r\ge0}\frac{\pi_{r+1}}{1+\frac{\mu}{r+1}},\qquad
 p_2(\mu)=\max_{r\ge0}\frac{2\pi_{r+2}}{1+\frac{\mu^2}{(r+1)(r+2)}} ,
\]
so that, by Theorem~\ref{thm:rare} applied to sub-families, $\kappa\varphi(0)p_1$ is the limit coefficient of each adjacent
pair $\{k_n,k_n+1\}$, $\{k_n+1,k_n+2\}$ and $\kappa\varphi(0)p_2$ that of the outer pair $\{k_n,k_n+2\}$ (scores $2\kappa Z$, ratios
$\pi_{r+2}/\pi_r$). Put $v(\mu)=\mu(\mu+1)/(\mu^2+\mu+1)$, so $\sigma=\pi_0v$, and
\[
 S(\mu)=\sum_{r\ge1}\frac{\mu^r}{r\cdot r!},\qquad
 F(\mu)=\frac{e^\mu-1-\mu}{\mu}+\frac{v(\mu)}2 .
\]
$F$ is strictly increasing on $(0,\infty)$ (both terms are: the first is $\sum_{k\ge2}\mu^{k-1}/k!$, and
$v'(\mu)=(1+2\mu)/(\mu^2+\mu+1)^2>0$), $F(0+)=0$ and $F(1)=e-2+\frac13>1$; let $\mu_1$ be the unique solution of
$F(\mu_1)=1$, $\mu_1=0.955345921673726826\ldots$. For $0<\mu\le\frac12$ one has the rational bound
\begin{equation}\label{eq:Fbound}
 F(\mu)\le\frac{\mu/2}{1-\mu/3}+\frac{v(1/2)}2\le\frac3{10}+\frac3{14}=\frac{18}{35}<1,
\end{equation}
because $k!\ge2\cdot3^{k-2}$ for $k\ge2$ gives $\sum_{k\ge2}\mu^{k-1}/k!\le\frac\mu2\sum_{i\ge0}(\mu/3)^i$, and $v$ is
increasing with $v(\frac12)=\frac37$.

\begin{theorem}[Three compositions]\label{thm:three}
\begin{enumerate}[label=\textup{(\alph*)}]
\item For every $\mu>0$, $\mathcal V_2(\mu)\ge\max\{\sigma(\mu),p_1(\mu),p_2(\mu)\}$.
\item $\sigma(\mu)>\max\{p_1(\mu),p_2(\mu)\}$ if and only if $\mu<1$; at $\mu=1$, $\sigma=p_2=\frac2{3e}$.
\item For every $0<\mu\le1$, $\mathcal V_2(\mu)=\sigma(\mu)$.
\item For every $\mu\ge1$, $\mathcal V_2(\mu)=p_2(\mu)$. More precisely, with $M(\mu)=\max_k\pi_k$, every feasible point of
 \eqref{eq:Vm} or of a truncation $\mathcal V^{(R)}_2$ with prior $w$ has value at most
 \begin{equation}\label{eq:priorbound}
  p_2(\mu)-\bigl(3p_2(\mu)-2M(\mu)\bigr)w_1,\qquad\text{where }3p_2(\mu)-2M(\mu)\ge0,\ \text{with equality only at }\mu=1.
 \end{equation}
\end{enumerate}
Consequently, if $q\ne p$, the limit coefficient is known for every $\mu>0$:
\begin{equation}\label{eq:phase}
 \lim_{n\to\infty}\sqrt n\,\Delta^{(2)}_n=\kappa\,\varphi(0)\,\mathcal V_2(\mu),\qquad
 \mathcal V_2(\mu)=\begin{cases}\dfrac{e^{-\mu}\mu(\mu+1)}{\mu^2+\mu+1},&0<\mu\le1,\\[3mm]
 \displaystyle\max_{r\ge0}\frac{2\pi_{r+2}}{1+\mu^2/((r+1)(r+2))},&\mu\ge1,\end{cases}
\end{equation}
the two branches agreeing at $\mu=1$ (value $\frac2{3e}$). The family coefficient is strictly larger than the largest of the
three pairwise coefficients $\kappa\varphi(0)\max\{p_1,p_2\}$ if and only if $0<\mu<1$; for $\mu\ge1$ it equals the
coefficient of the outer pair. The ratio of the family coefficient to the largest pairwise coefficient equals
$(1+\mu)^2/(1+\mu+\mu^2)$ for $0<\mu\le\sqrt3-1$, $(1+\mu)(2+\mu^2)/(2\mu(1+\mu+\mu^2))$ for $\sqrt3-1\le\mu\le1$, and $1$ for
$\mu\ge1$; its maximum over all $\mu>0$ is $3/(4-\sqrt3)=(12+3\sqrt3)/13=1.3228\ldots$, attained at $\mu=\sqrt3-1$.
\end{theorem}

\begin{proof}[Proof of \textup{(a)} and \textup{(b)}]
(a) The pair values are the values of the program restricted to $w$ supported on a pair. For $\sigma$ exhibit a feasible
point: $w\propto(1,\mu^{-1},\mu^{-2})$, and the vertex tests
\[
 \alpha^+_0=w_1e_1,\ \alpha^-_0=w_0e_0;\qquad\alpha^+_{-1}=w_2e_2,\ \alpha^-_{-1}=w_1e_1;\qquad
 \alpha^\pm_{-2}=w_2e_2;\qquad\alpha^\pm_r=w_0e_0\ (r\ge1).
\]
Feasibility: in fiber $0$ both tests give $\langle\cdot,q_0\rangle=w_0\pi_0=w_1\pi_1=:m_0$ and the remaining vertex gives
$w_2\pi_2=m_0/2$; in fiber $-1$ ($q_{-1}=(0,\pi_0,\pi_1)$) both tests give $w_1\pi_0=w_2\pi_1=:m_{-1}$ and $w_0e_0$ gives $0$; in
fiber $-2$ ($q_{-2}=(0,0,\pi_0)$) the maximum $w_2\pi_0$ is attained only at $w_2e_2$; in fiber $r\ge1$ the three vertices give
$w_0\pi_r$, $w_1\pi_{r+1}=w_0\pi_r/(r+1)$ and $w_2\pi_{r+2}=w_0\pi_r/((r+1)(r+2))$, so $m_r=w_0\pi_r$ is attained at $w_0e_0$ and
every cross constraint holds. Value: $\ell_0(w_1e_1)-\ell_0(w_0e_0)=\pi_1w_1$, $\ell_{-1}(w_2e_2)-\ell_{-1}(w_1e_1)=2\pi_1w_2-\pi_0w_1
=\pi_1w_2$ (using $\pi_0w_1=\pi_1w_2$), all other fibers $0$; and
$\pi_1(w_1+w_2)=e^{-\mu}\mu\,\frac{\mu^{-1}+\mu^{-2}}{1+\mu^{-1}+\mu^{-2}}=\sigma(\mu)$.

(b) The ratio of consecutive terms in $p_1$ is $\mu(r+1+\mu)/((r+1)(r+2+\mu))<1$ iff $\mu^2<(r+1)(r+2)$, so for $\mu<\sqrt2$ the
maximum is at $r=0$: $p_1=e^{-\mu}\mu/(1+\mu)$. In $p_2$ the ratio of consecutive terms is
$\mu[(r+1)(r+2)+\mu^2]/((r+1)[(r+2)(r+3)+\mu^2])$, which is $<1$ for $\mu\le1$ and all $r\ge0$; so for $\mu\le1$,
$p_2=2e^{-\mu}\mu^2/(2+\mu^2)$. Hence, for $0<\mu\le1$, $\sigma/p_1=(1+\mu)^2/(1+\mu+\mu^2)>1$ and
$\sigma/p_2=(1+\mu)(2+\mu^2)/(2\mu(1+\mu+\mu^2))$, whose numerator minus denominator is $-(\mu-1)(\mu^2+2\mu+2)$; so on
$(0,1]$, $\sigma>\max\{p_1,p_2\}$ iff $\mu<1$, with $\sigma=p_2$ at $\mu=1$. For $\mu>1$,
$p_2\ge2\pi_2/(1+\mu^2/2)=2e^{-\mu}\mu^2/(2+\mu^2)>\sigma$ by the same computation (for $\mu>\sqrt2$ the maximising atom
of $p_1$ moves and the displayed ratio $\sigma/p_1$ no longer applies, but it is not needed; at $\mu=\sqrt2$ the atoms
$r=0$ and $r=1$ tie, so the $r=0$ formula still applies there).
The two pair values coincide when $\mu/(1+\mu)=2\mu^2/(2+\mu^2)$, i.e.\ $\mu^2+2\mu-2=0$, $\mu=\sqrt3-1$; the first ratio
increases on $(0,1)$ (its derivative has the sign of $(1+\mu)(1-\mu)$), and the second decreases because
\[
 \frac{d}{d\mu}\frac{(1+\mu)(2+\mu^2)}{2\mu(1+\mu+\mu^2)}
 =-\frac{2\mu^3+7\mu^2+4\mu+2}{2\mu^2(1+\mu+\mu^2)^2}<0\qquad(\mu>0).
\]
Hence the maximum is at $\mu=\sqrt3-1$, where the ratio is $3/(4-\sqrt3)$. For $\mu\ge1$ the family
value $p_2$ (part (d)) is at least every pairwise value and equals the outer-pair value, so the ratio is $1$. The
consequences follow from (c), (d) and Theorem~\ref{thm:rare}.
\end{proof}

The upper bounds in (c) and (d) use weak duality for the linear program \eqref{eq:Vm}.

\begin{lemma}[Dual certificates for arbitrary finite windows]\label{lem:dual}
Let $m\ge1$, fix $R\ge0$, and put $S_R=\{-m,\ldots,R\}$. For $r\in S_R$ let
$q_r=(\pi_{r+j})_{j=0}^m$, with $\pi_k=0$ for $k<0$. Suppose real numbers
$\nu^\pm_r$ ($r\in S_R$), nonnegative numbers $\lambda^\pm_{rr'}$
($r,r'\in S_R$, $r\ne r'$), and nonnegative numbers $\beta^\pm_{r,j}$
($r\in S_R$, $0\le j\le m$) satisfy
\begin{align*}
 &\nu^+_{r'}+\nu^-_{r'}+\sum_{\substack{r\in S_R\\r\ne r'}}\bigl(\lambda^+_{rr'}+\lambda^-_{rr'}\bigr)=0
 &&\text{for every }r'\in S_R,\qquad\text{(i)}\\
 &\pm j\,\pi_{r+j}\le\nu^\pm_r\pi_{r+j}+\sum_{\substack{r'\in S_R\\r'\ne r}}\lambda^\pm_{rr'}\pi_{r'+j}+\beta^\pm_{r,j}
 &&\text{for every }r\in S_R,\ 0\le j\le m.\qquad\text{(ii$^\pm$)}
\end{align*}
Put $b_j=\sum_{r\in S_R}(\beta^+_{r,j}+\beta^-_{r,j})$, the \emph{budget} of coordinate $j$.
Every feasible point of $\mathcal V_m^{(R)}(\mu)$ with prior $w$ has objective at most $\sum_{j=0}^m b_jw_j$.
In particular, if
\[
 b_j\le V\qquad(0\le j\le m),\qquad\text{(iii)}
\]
then $\mathcal V_m^{(R)}(\mu)\le V$. If certificates with a common bound $V$ exist for arbitrarily large $R$, then
$\mathcal V_m(\mu)\le V$.
\end{lemma}
\begin{proof}
Let $(w,(m_r),(\alpha^\pm_r))$ be feasible for the finite program. Thus
$\langle\alpha^\pm_r,q_r\rangle=m_r$,
$\langle\alpha^\pm_r,q_{r'}\rangle\le m_{r'}$ for $r'\ne r$,
and $0\le\alpha^\pm_{r,j}\le w_j$. Add to its objective the finite sum
\[
 \begin{aligned}
 &-\sum_{r\in S_R,\,\pm}\nu^\pm_r
       (\langle\alpha^\pm_r,q_r\rangle-m_r)\\
 &-\sum_{\substack{r,r'\in S_R\\r\ne r'},\,\pm}\lambda^\pm_{rr'}
       (\langle\alpha^\pm_r,q_{r'}\rangle-m_{r'})
 -\sum_{\substack{r\in S_R,\ 0\le j\le m},\,\pm}\beta^\pm_{r,j}
       (\alpha^\pm_{r,j}-w_j).
 \end{aligned}
\]
The first line is zero, and each term in the second line is nonnegative. The resulting expression therefore bounds
$\sum_{r\in S_R}[\ell_r(\alpha^+_r)-\ell_r(\alpha^-_r)]$ from above. On collecting coefficients, that of each level $m_{r'}$ is the left side of (i), hence zero. That of each nonnegative test coordinate $\alpha^\pm_{r,j}$ is
\[
 \pm j\pi_{r+j}-\nu^\pm_r\pi_{r+j}
 -\sum_{\substack{r'\in S_R\\r'\ne r}}\lambda^\pm_{rr'}\pi_{r'+j}-\beta^\pm_{r,j}\le0
\]
by (ii$^\pm$). The remaining expression is at most $\sum_{j=0}^m b_jw_j$, which is at most $V$ under (iii) because
$w_j\ge0$ and $\sum_{j=0}^m w_j=1$. This proves the finite assertion without any assumption on which Poisson coordinates vanish. For an unbounded sequence of such truncations, Theorem~\ref{thm:rare} gives
$\mathcal V_m\le V+\sum_{r>R}\sum_{j=1}^m j\pi_{r+j}$, whose tail tends to zero. This proves the last assertion.
\end{proof}

\begin{proof}[Proof of \textup{(c)}]
By (a) it suffices to show $\mathcal V_2(\mu)\le\sigma(\mu)$, and by the truncation inequality of Theorem~\ref{thm:rare} it
suffices to certify $\mathcal V^{(R)}_2(\mu)\le\sigma(\mu)$ for every $R\ge2$. We give two certificates; the first is
valid exactly for $0<\mu\le\mu_1$ and is used on $(0,\frac12]$, where \eqref{eq:Fbound} verifies it without solving a
transcendental equation; the second covers $\frac12\le\mu\le1$. Throughout, $V=\sigma=\pi_0v$,
$a=\frac{2\mu^2+\mu+1}{\mu^2+\mu+1}$, $b=a+\mu$, and the common part of both certificates is
\begin{equation}\label{eq:common}
\begin{gathered}
 \nu^\pm_{-2}=\pm2;\qquad \nu^+_{-1}=a,\ \ \nu^-_{-1}=-b,\ \ \lambda^-_{-1,0}=v,\ \ \beta^+_{-1,2}=\beta^-_{-1,1}=V;\\
 \nu^+_0=0,\ \ \lambda^+_{0,-1}=\mu;\qquad
 \nu^+_r=-\tfrac1r,\ \ \lambda^+_{r,0}=\frac{\pi_r}{r\pi_0}=\frac{\mu^r}{r\cdot r!}\quad(1\le r\le R),
\end{gathered}
\end{equation}
all other $\lambda^+$ and $\beta^+$ being zero. We first verify the conditions that involve only \eqref{eq:common}.
Condition (i) at $r'=-2$: $2-2=0$; at $r'=-1$: $a-b+\mu=0$. (ii$^+$) at $r\ge1$: $j=0$: $0\le-\pi_r/r+\pi_r/r$;
$j=1$: $\pi_{r+1}\le-\pi_{r+1}/r+\pi_{r+1}(r+1)/r$ (equality, using $\pi_r\pi_1/\pi_0=(r+1)\pi_{r+1}$);
$j=2$: $2\pi_{r+2}\le\pi_{r+2}\bigl(-\frac1r+\frac{(r+1)(r+2)}{2r}\bigr)=\pi_{r+2}\frac{r+3}2$, true for $r\ge1$.
(ii$^+$) at $r=0$: $0\le0$; $\pi_1\le\mu\pi_0$; $2\pi_2\le\mu\pi_1$ (both equalities).
(ii$^+$) at $r=-1$ ($q_{-1}=(0,\pi_0,\pi_1)$): $j=1$: $\pi_0\le a\pi_0$ as $a\ge1$; $j=2$: $2\pi_1\le a\pi_1+V$, an equality
because $V/\pi_1=(\mu+1)/(\mu^2+\mu+1)=2-a$. (ii$^-$) at $r=-1$: $j=0$: $0\le v\pi_0$; $j=1$: $-\pi_0\le-b\pi_0+v\pi_1+V$,
i.e.\ $b\le1+v\mu+v$, an equality since $a+\mu-1-v=\mu^2(\mu+1)/(\mu^2+\mu+1)=v\mu$; $j=2$: $-2\pi_1\le-b\pi_1+v\pi_2$, i.e.\
$a+\mu\le2+v\mu/2$, which after clearing denominators reads $\mu^3+\mu^2\le2$, true for $\mu\le1$.
(ii$^\pm$) at $r=-2$ ($q_{-2}=(0,0,\pi_0)$): $2\pi_0\le2\pi_0$ and $-2\pi_0\le-2\pi_0$; nothing for $j<2$.
(iii) for $j=1,2$: the single boxes $\beta^-_{-1,1}=\beta^+_{-1,2}=V$.

\emph{First certificate ($0<\mu\le\mu_1$).}  Complete \eqref{eq:common} by
\begin{gather*}
 \nu^-_0=-(S^{(R)}+v),\qquad \lambda^-_{0,1}=1+\epsilon,\qquad \lambda^-_{0,r}=\tfrac1r\quad(2\le r\le R),\\
 \nu^-_1=-\epsilon,\qquad \beta^-_{1,0}=\epsilon\pi_1,\qquad \beta^-_{0,0}=V-\epsilon\pi_1,\qquad \nu^-_r=0\quad(r\ge2),
\end{gather*}
with $S^{(R)}=\sum_{r=1}^R\mu^r/(r\cdot r!)$ and $\epsilon=v/\mu=(1+\mu)/(1+\mu+\mu^2)\in(0,1)$, so that $\epsilon\pi_1=V$ and
$\beta^-_{0,0}=0$. Condition (i) at $r'=0$: $-(S^{(R)}+v)+\sum_{r=1}^R\mu^r/(r\,r!)+v=0$; at $r'=1$: $-1-\epsilon+(1+\epsilon)=0$;
at $r'\ge2$: $-\frac1r+\frac1r=0$. (iii) for $j=0$: $\beta^-_{1,0}+\beta^-_{0,0}=V$.
(ii$^-$) at $r\ge2$: $0\le0$, $-j\pi_{r+j}\le0$. (ii$^-$) at $r=1$: $j=0$: $0\le-\epsilon\pi_1+\epsilon\pi_1$;
$j=1,2$: $-j\pi_{1+j}\le-\epsilon\pi_{1+j}$ as $\epsilon\le1$. (ii$^-$) at $r=0$: $j=0$:
$0\le-(S^{(R)}+v)\pi_0+(1+\epsilon)\pi_1+\sum_{r=2}^R\pi_r/r=-(S^{(R)}+v)\pi_0+S^{(R)}\pi_0+V=0$;
$j=1$: $-\pi_1\le-(S^{(R)}+v)\pi_1+(1+\epsilon)\pi_2+\sum_{r=2}^R\pi_{r+1}/r$, i.e.\
\[
 S^{(R)}+v\le1+\frac{(1+\epsilon)\mu}2+\sum_{r=2}^R\frac{\mu^r}{r\,(r+1)!}
 \iff\sum_{r=1}^R\frac{\mu^r}{(r+1)!}+\frac v2\le1,
\]
using $\frac{\mu^r}{r\,r!}-\frac{\mu^r}{r\,(r+1)!}=\frac{\mu^r}{(r+1)!}$; the left side increases with $R$ to $F(\mu)$,
so this holds for all $R$ exactly when $F(\mu)\le1$, i.e.\ $\mu\le\mu_1$, and in particular for $\mu\le\frac12$ by
\eqref{eq:Fbound};
$j=2$: $-2\pi_2\le-(S^{(R)}+v)\pi_2+(1+\epsilon)\pi_3+\sum_{r=2}^R\pi_{r+2}/r$, implied by $S^{(R)}+v\le2$, which holds for
$\mu\le1$ because $S(1)=\sum_{r\ge1}\frac1{r\,r!}\le1+\frac14+\frac1{18}+\frac14\sum_{r\ge4}\frac1{r!}<\frac43$ and
$v\le\frac23$. All multipliers $\lambda,\beta$ are non-negative. Lemma~\ref{lem:dual} gives $\mathcal V^{(R)}_2\le\sigma$.

\emph{Second certificate ($\frac12\le\mu\le1$).}  Complete \eqref{eq:common} by
\begin{gather*}
 \nu^-_0=-(S^{(R)}+v),\qquad\lambda^-_{0,1}=1+\epsilon,\qquad\nu^-_1=-\epsilon,\qquad\beta^-_{1,0}=V,\\
 \lambda^-_{1,r}=\tfrac1r\quad(2\le r\le R),\qquad\nu^-_r=0\quad(r\ge2),
\end{gather*}
with $\epsilon=\epsilon^{(R)}=\frac1{\pi_1}\Bigl(V+\sum_{r=2}^R\frac{\pi_r}r\Bigr)=\frac v\mu+\sum_{r=2}^R\frac{\mu^{r-1}}{r\,r!}$;
now fiber $1$ feeds the fibers $r\ge2$ and fiber $0$ feeds only fiber $1$. Condition (i) at $r'=0$: as before; at $r'=1$:
$-1-\epsilon+1+\epsilon=0$; at $r'\ge2$: $-\frac1r+\frac1r=0$. (iii) for $j=0$: $\beta^-_{1,0}=V$.
(ii$^-$) at $r\ge2$: trivial. (ii$^-$) at $r=1$: $j=0$: $0\le-\epsilon\pi_1+\sum_{r=2}^R\pi_r/r+V=0$ by the choice of $\epsilon$;
$j=1$: $-\pi_2\le-\epsilon\pi_2+\sum_{r=2}^R\pi_{r+1}/r$, i.e.\
\begin{equation}\label{eq:C2}
 \epsilon\le1+\sum_{r=2}^R\frac{2\mu^{r-1}}{r\,(r+1)!};
\end{equation}
$j=2$: $-2\pi_3\le-\epsilon\pi_3+\sum_{r=2}^R\pi_{r+2}/r$, implied by \eqref{eq:C2} since $\epsilon\le1+\sum_{r\ge2}2/(r(r+1)!)<2$.
(ii$^-$) at $r=0$: $j=0$: $0\le-(S^{(R)}+v)\pi_0+(1+\epsilon)\pi_1=-(S^{(R)}+v)\pi_0+\pi_1+V+\sum_{r=2}^R\pi_r/r=0$;
$j=1$: $-\pi_1\le-(S^{(R)}+v)\pi_1+(1+\epsilon)\pi_2$, i.e.\ $S^{(R)}+v\le1+\frac{(1+\epsilon)\mu}2$; since
$\frac{(1+\epsilon)\mu}2=\frac\mu2+\frac v2+\frac12\sum_{r=2}^R\frac{\mu^r}{r\,r!}=\frac{S^{(R)}+v}2$, this is
\begin{equation}\label{eq:C1}
 S^{(R)}+v\le2 ,
\end{equation}
true for $\mu\le1$ as shown above; $j=2$: $-2\pi_2\le-(S^{(R)}+v)\pi_2+(1+\epsilon)\pi_3$, implied by \eqref{eq:C1}.
It remains to verify \eqref{eq:C2} on $[\frac12,1]$. Substituting $\epsilon$ and dividing by $\mu$, \eqref{eq:C2} is equivalent to
\[
 \sum_{r=2}^R\mu^{r-2}\,\frac{r-1}{r\,(r+1)!}\ \le\ \frac{1}{\mu}\Bigl(1-\frac v\mu\Bigr)=\frac{\mu}{1+\mu+\mu^2}.
\]
The left side is at most its value at $\mu=1$, $R=\infty$, which is
$\frac1{12}+\frac1{36}+\frac1{160}+\frac1{900}+\sum_{r\ge6}\frac{r-1}{r\,(r+1)!}$; the tail is at most
$\sum_{k\ge7}\frac1{k!}\le\frac1{5040}\cdot\frac1{1-1/8}=\frac8{7\cdot5040}$, so the whole left side is at most
$\frac{4653}{39200}<\frac3{25}$; the right side is increasing on $(0,1)$ (derivative proportional to $1-\mu^2$) and equals
$\frac27>\frac3{25}$ at $\mu=\frac12$.
Lemma~\ref{lem:dual} gives $\mathcal V^{(R)}_2\le\sigma$ for $\frac12\le\mu\le1$. The two certificates cover $(0,1]$ since
$\mu_1>\frac12$, and the proof of (c) is complete.
\end{proof}

\begin{remark}[Interpretation of the multipliers]
The multipliers are dual variables of the linear program \eqref{eq:Vm}, and the cross multipliers can be read as transport
rates of a limiting reconstruction kernel: $\lambda^+_{r,0}=\mu^r/(r\,r!)$ as the rate at which mass of order $n^{-1/2}$
with density proportional to $Z_+$ would be moved from fiber $0$ into fiber $r$, $\lambda^-_{0,r}$ and $\lambda^-_{1,r}$
as return flows of $Z_-$-shaped mass, $\lambda^+_{0,-1}=\mu$ and $\lambda^-_{-1,0}=v$ as the exchange between the two lowest
fibers. This is an interpretation of the dual variables only: no finite-$n$ kernel with these rates is constructed here,
and Theorem~\ref{thm:three} does not need one (the value is pinned by the tests of part (a) from below and by the
certificates from above). The simplex of tests of part (a) is explicit: fiber $R_n=0$ decides between $k_n$ and $k_n+1$ by
the sign of the native score, fiber $R_n=-1$ between $k_n+1$ and $k_n+2$, and the prior $w\propto(1,\mu^{-1},\mu^{-2})$ makes
the coarse posterior tie in both fibers.
\end{remark}

\subsection{The outer pair for \texorpdfstring{$\mu\ge1$}{mu >= 1}}\label{sec:outer}
Throughout this subsection $\mu\ge1$ and
\[
 V=p_2(\mu)=\max_{k\ge0}\frac{2\pi_k\pi_{k+2}}{\pi_k+\pi_{k+2}},\qquad M=\max_{k\ge0}\pi_k,\qquad H=\{k\ge0:\pi_k>V\};
\]
the harmonic-mean form of $p_2$ agrees with the definition in Section~\ref{sec:three} because
$\pi_{k+2}/\pi_k=\mu^2/((k+1)(k+2))$. The maxima exist because the terms tend to zero; $M=\pi_{\lfloor\mu\rfloor}$. The upper bound in Theorem~\ref{thm:three}(d)
is a single explicit certificate for Lemma~\ref{lem:dual}, valid for every truncation level $R$ and every $\mu\ge1$: on the
plus side it follows the Poisson recurrence exactly, on the minus side it moves mass along the chain $r\to r+1$ at rates
built from the clipped quantities $d_k=(1-V/\pi_k)_+$, and the three budgets of Lemma~\ref{lem:dual} become the downward
variations of three short modifications of the Poisson mass function. No infinite dual optimum is used: the bound is
proved for each finite truncation and passed to the full program by the truncation inequality of Theorem~\ref{thm:rare}.

\begin{lemma}[High atoms]\label{lem:high}
Let $\mu\ge1$. Then $0<V<M$; the set $H$ is $\{a\}$ or $\{a,a+1\}$ for some $a\ge0$, and it contains $\lfloor\mu\rfloor$;
and $V\ge\frac23M$, with strict inequality for $\mu>1$.
\end{lemma}
\begin{proof}
A harmonic mean of two positive numbers lies between them and equals the larger one only if the two are equal, so
$V=M$ would require $\pi_k=\pi_{k+2}=M$ for some $k$; but $\pi_{k+1}/\pi_k=\mu/(k+1)$ shows that the Poisson law has at
most two maximal atoms and that they are adjacent. Hence $0<V<M$. The ratios $\pi_{k+1}/\pi_k$ decrease, so $(\pi_k)$ is
nondecreasing up to $\lfloor\mu\rfloor$ and nonincreasing afterwards; therefore $H$ is an interval, and it contains
$\lfloor\mu\rfloor$ because $\pi_{\lfloor\mu\rfloor}=M>V$. If $k$ and $k+2$ both belonged to $H$, the harmonic mean of
$\pi_k$ and $\pi_{k+2}$ would exceed $V$, against the definition of $V$; so $H$ has at most two elements. For the bound let
$m=\lfloor\mu\rfloor\ge1$, so $\pi_m=M$, $\pi_{m-1}=\frac m\mu M$ and $\pi_{m+1}=\frac\mu{m+1}M$; the pair $(m-1,m+1)$ gives
\[
 \frac VM\ \ge\ \frac{2\pi_{m-1}\pi_{m+1}}{M(\pi_{m-1}+\pi_{m+1})}=\frac{2m\mu}{m(m+1)+\mu^2}\ \ge\ \frac{2m}{2m+1}\ \ge\ \frac23,
\]
the middle inequality being equivalent to $(\mu-m)(m+1-\mu)\ge0$. For $\mu>1$: if $m\ge2$ the right side is at least
$\frac45$; if $1<\mu<2$ then $m=1$ and $\frac{2\mu}{2+\mu^2}>\frac23$ because $(\mu-1)(2-\mu)>0$.
\end{proof}

\paragraph{The chain certificate.}
Put $d_k=(1-V/\pi_k)_+$ for $k\ge0$, $d_k=0$ for $k<0$, and
\begin{equation}\label{eq:rates}
 a_r=r+2+d_{r+3}-d_{r+1}\quad(r\ge-2),\qquad a_{-3}=0 .
\end{equation}
Then $a_{-2}=d_1\ge0$, and $a_r\ge r+2-d_{r+1}>0$ for $r\ge-1$ because $0\le d_k<1$. Fix $R\ge0$ and define, for the fibers
$-2\le r\le R$,
\begin{equation}\label{eq:chain}
\begin{gathered}
 \nu^+_r=-r,\qquad \lambda^+_{r,r-1}=\mu\ \ (-1\le r\le R),\qquad \beta^+_{r,j}=0;\\
 \lambda^-_{r,r+1}=a_r\ \ (r<R),\qquad \nu^-_r=r-\mu-a_{r-1}\ \ (r<R),\qquad \nu^-_R=R-a_{R-1};\\
 v_{r,j}=a_r\pi_{r+j+1}-\mu\pi_{r+j}\ \ (-2\le r<R),\qquad v_{-3,j}=v_{R,j}=0,\qquad
 \beta^-_{r,j}=\bigl(v_{r-1,j}-v_{r,j}\bigr)_+ ,
\end{gathered}
\end{equation}
all other cross multipliers being zero. All $\lambda$ and $\beta$ are non-negative.

\begin{lemma}[Rows, balances and budgets]\label{lem:chainrows}
The multipliers \eqref{eq:chain} satisfy (i) and (ii$^\pm$) of Lemma~\ref{lem:dual} for the program with fibers $-2\le r\le R$,
the plus rows with equality. The budgets are
\begin{equation}\label{eq:budgets}
 b_j=\sum_{r=-2}^{R}\bigl(v_{r-1,j}-v_{r,j}\bigr)_+ ,
\end{equation}
the total downward variation of the sequence $(v_{-3,j},v_{-2,j},\dots,v_{R,j})$, which starts and ends at $0$. For
$-2\le r<R$,
\begin{equation}\label{eq:vclosed}
 v_{r,0}=\pi_{r+1}\bigl(1+d_{r+3}-d_{r+1}\bigr),\qquad
 v_{r,1}=\pi_{r+2}\bigl(d_{r+3}-d_{r+1}\bigr),\qquad
 v_{r,2}=-\pi_{r+3}\bigl(1-d_{r+3}+d_{r+1}\bigr),
\end{equation}
so that $v_{-2,0}=0$.
\end{lemma}
\begin{proof}
Column $r'$ receives $\lambda^+_{r'+1,r'}=\mu$ when $r'+1\le R$ and $\lambda^-_{r'-1,r'}=a_{r'-1}$ when $r'\ge-1$. Condition
(i) therefore reads $2+(-2-\mu-a_{-3})+\mu=0$ at $r'=-2$, $-r'+(r'-\mu-a_{r'-1})+\mu+a_{r'-1}=0$ for $-1\le r'<R$, and
$-R+(R-a_{R-1})+a_{R-1}=0$ at $r'=R$. The Poisson recurrence $k\pi_k=\mu\pi_{k-1}$ holds for every integer $k$ (both sides
vanish for $k\le0$); with $k=r+j$ it gives $j\pi_{r+j}=\mu\pi_{r+j-1}-r\pi_{r+j}$, which is the plus row of fiber $r\ge-1$
with equality. At $r=-2$ the plus row is $j\pi_{j-2}\le2\pi_{j-2}$: equality for $j=2$ and $0\le0$ otherwise. For a minus
row with $r<R$, the right side minus the left side before the box term is
\[
 (r+j-\mu-a_{r-1})\pi_{r+j}+a_r\pi_{r+j+1}
 =\mu\pi_{r+j-1}-\mu\pi_{r+j}-a_{r-1}\pi_{r+j}+a_r\pi_{r+j+1}=v_{r,j}-v_{r-1,j},
\]
where for $r=-2$ one uses $a_{-3}=0$ and $\mu\pi_{j-3}=0$ ($j\le2$), in agreement with $v_{-3,j}=0$. At $r=R$ it is
$(R+j-a_{R-1})\pi_{R+j}=\mu\pi_{R+j-1}-a_{R-1}\pi_{R+j}=-v_{R-1,j}=v_{R,j}-v_{R-1,j}$. Adding $\beta^-_{r,j}=(v_{r-1,j}-v_{r,j})_+$
makes every minus row hold. Finally $\mu\pi_{r+j}=(r+j+1)\pi_{r+j+1}$ for every integer $r+j$, so
$v_{r,j}=\bigl(a_r-(r+j+1)\bigr)\pi_{r+j+1}=(1-j+d_{r+3}-d_{r+1})\pi_{r+j+1}$, which is \eqref{eq:vclosed}; $v_{-2,0}=0$ because
$\pi_{-1}=0$.
\end{proof}

\begin{lemma}[Endpoint budgets]\label{lem:endpoints}
$b_0\le V$ and $b_2\le V$.
\end{lemma}
\begin{proof}
Put $F_k=\pi_k(1+d_{k+2}-d_k)$ for $k\ge0$ and $G_k=\pi_k(1-d_k+d_{k-2})$ for $k\ge1$. By \eqref{eq:vclosed} the two
sequences in \eqref{eq:budgets} are $(0,0,F_0,F_1,\dots,F_R,0)$ for $j=0$ and $(0,-G_1,-G_2,\dots,-G_{R+2},0)$ for $j=2$. A
sequence that starts and ends at $0$ has downward variation equal to its maximum if it is non-negative and unimodal
(nondecreasing, then nonincreasing), and equal to minus its minimum if it is non-positive with a single valley. It therefore
suffices to show that $0\le F_k\le V$, $0\le G_k\le V$, and that $(F_k)_{k\ge0}$ and $(G_k)_{k\ge1}$ are unimodal.

\emph{Values.}  By Lemma~\ref{lem:high}, $d_k$ and $d_{k+2}$ are never both positive. If $k\in H$ then $d_{k+2}=0$ and
$F_k=\pi_k\cdot V/\pi_k=V$. If $k\notin H$ and $k+2\notin H$ then $F_k=\pi_k\le V$. If $k\notin H$ and $k+2\in H$ then
$F_k=\pi_k(2-V/\pi_{k+2})$, which lies in $(\pi_k,V]$: it exceeds $\pi_k$ because $\pi_{k+2}>V$, and
$\pi_k(2-V/\pi_{k+2})\le V$ is equivalent to $V\ge2\pi_k\pi_{k+2}/(\pi_k+\pi_{k+2})$, true by the definition of $V$. In the
same way $G_k=V$ for $k\in H$, $G_k=\pi_k\le V$ if $k,k-2\notin H$, and $G_k=\pi_k(2-V/\pi_{k-2})\in(\pi_k,V]$ if $k\notin H$,
$k-2\in H$.

\emph{Shape.}  Let $a=\min H$. The sequence $(\pi_k)$ is nondecreasing up to $\lfloor\mu\rfloor\in H$ and nonincreasing
afterwards. $F$ differs from $\pi$ only on $H$, where it equals $V$, and at the ``boosted'' indices $k\notin H$ with $k+2\in H$,
namely $k=a-2$ if $H=\{a\}$ and $k\in\{a-2,a-1\}$ if $H=\{a,a+1\}$ (indices $<0$ are absent); $G$ differs from $\pi$ on $H$ and
at $k\notin H$ with $k-2\in H$, namely $k=a+2$, respectively $k\in\{a+2,a+3\}$. Since boosts are upward and all values are
capped by $V$, every comparison between two unmodified entries, or between an entry of $H$ and a neighbour, is inherited from
the Poisson shape ($\pi_{a-1}\le V$ and $\pi_{a+1}\le V$, or $\pi_{a+2}\le V$, because those indices are outside $H$). The
remaining comparisons are the following.

If $H=\{a\}$, then $a=\lfloor\mu\rfloor\ge1$ and $a\le\mu<a+1$. For $a\ge2$ we need $F_{a-2}\le\pi_{a-1}$: as $V\ge\pi_{a-1}$,
$F_{a-2}\le\pi_{a-2}(2-\pi_{a-1}/\pi_a)$, and $\pi_{a-2}(2-\pi_{a-1}/\pi_a)\le\pi_{a-1}$ is equivalent to
$2-\frac a\mu\le\frac\mu{a-1}$, i.e.\ to $\mu^2-2(a-1)\mu+a(a-1)=(\mu-a+1)^2+a-1\ge0$. On the right we need
$G_{a+2}\le\pi_{a+1}$: as $V\ge\pi_{a+1}$, $G_{a+2}\le\pi_{a+2}(2-\pi_{a+1}/\pi_a)$, and $\pi_{a+2}(2-\pi_{a+1}/\pi_a)\le\pi_{a+1}$ is
equivalent to $2-\frac\mu{a+1}\le\frac{a+2}\mu$, i.e.\ to $\mu^2-2(a+1)\mu+(a+1)(a+2)=(\mu-a-1)^2+a+1\ge0$; and
$G_{a+2}\ge\pi_{a+2}\ge\pi_{a+3}$.

If $H=\{a,a+1\}$, then $\lfloor\mu\rfloor\in\{a,a+1\}$, so $a\le\mu<a+2$. With $\pi_{a-1}=\frac a\mu\pi_a$,
$\pi_{a-2}=\frac{a(a-1)}{\mu^2}\pi_a$, $\pi_{a+1}=\frac\mu{a+1}\pi_a$, $\pi_{a+2}=\frac{\mu^2}{(a+1)(a+2)}\pi_a$ and
$\pi_{a+3}=\frac{\mu^3}{(a+1)(a+2)(a+3)}\pi_a$, a direct computation gives, for $a\ge2$,
\[
 F_{a-1}-F_{a-2}=\frac{2a}{\mu^2}\bigl[\pi_a(\mu-a+1)-V\bigr]>0,
\]
because $\mu-a+1\ge1$ and $\pi_a>V$ (for $a\le1$ the index $a-2$ is absent), and
\[
 G_{a+2}-G_{a+3}=\frac{2\mu^2\bigl[\pi_a(a+3-\mu)-V\bigr]}{(a+1)(a+2)(a+3)}>0,
\]
because $a+3-\mu>1$; moreover $G_{a+3}\ge\pi_{a+3}\ge\pi_{a+4}$. Thus in both cases $F$ rises to the plateau $V$ on $H$ and
falls afterwards, and so does $G$ (for $a=0$ the sequence $G$ begins with $G_1=V$).
\end{proof}

\begin{lemma}[Middle budget]\label{lem:middle}
$b_1\le2(M-V)\le V$.
\end{lemma}
\begin{proof}
By \eqref{eq:vclosed}, $v_{r,1}=\pi_{r+2}(d_{r+3}-d_{r+1})$, and at most one of $d_{r+3}$, $d_{r+1}$ is positive. Positive
entries occur only for $r+3\in H$, negative ones only for $r+1\in H$; with $H=\{a\}$ or $\{a,a+1\}$ the positive entries sit
at $r\in\{a-3\}$, respectively $\{a-3,a-2\}$, and the negative ones at $r\in\{a-1\}$, respectively $\{a-1,a\}$, every other entry
being $0$ (indices below $-2$ or above $R-1$ are not part of the sequence, which only shortens a block). So the sequence
$(0,v_{-2,1},\dots,v_{R-1,1},0)$ consists of zeros, a block of at most two non-negative entries, a block of at most two
non-positive entries, and zeros; its downward variation is at most the largest positive entry plus the largest absolute
value of a negative entry. Since $0\le d_k\le1-V/M$, every $|v_{r,1}|\le\pi_{r+2}(1-V/M)\le M-V$. Hence $b_1\le2(M-V)$, and
$2(M-V)\le V$ is $V\ge\frac23M$, Lemma~\ref{lem:high}.
\end{proof}

\begin{proof}[Proof of Theorem~\textup{\ref{thm:three}(d)}]
The lower bound $\mathcal V_2\ge p_2$ is part (a). For the upper bound fix $R\ge0$. By Lemmas~\ref{lem:chainrows},
\ref{lem:endpoints} and \ref{lem:middle}, the multipliers \eqref{eq:chain} satisfy (i) and (ii$^\pm$) with budgets
$b_0,b_2\le V$ and $b_1\le2(M-V)$, so Lemma~\ref{lem:dual} bounds the value of every feasible point of $\mathcal V^{(R)}_2$
with prior $w$ by
\[
 b_0w_0+b_1w_1+b_2w_2\le V(w_0+w_2)+2(M-V)w_1=V-(3V-2M)w_1 .
\]
A feasible point of the full program \eqref{eq:Vm}, restricted to the fibers $r\le R$, is feasible for the truncated program
with the same $w$, and its value is the limit of the truncated values (the series converges absolutely, Theorem~\ref{thm:rare});
so \eqref{eq:priorbound} holds for every feasible point of \eqref{eq:Vm}, and $\mathcal V_2(\mu)\le V=p_2(\mu)$. The sign
statement in \eqref{eq:priorbound} is Lemma~\ref{lem:high}.
\end{proof}

\begin{corollary}[Optimal priors]\label{cor:priors}
\begin{enumerate}[label=\textup{(\alph*)}]
\item For $\mu>1$, every feasible point of \eqref{eq:Vm} with value $\mathcal V_2(\mu)$ has $w_1=0$, and a feasible point with
 value at least $\mathcal V_2(\mu)-\epsilon$ has $w_1\le\epsilon/(3p_2(\mu)-2M(\mu))$.
\item For $0<\mu<1$, every feasible point with value $\mathcal V_2(\mu)$ has $w_0,w_1,w_2>0$.
\item At $\mu=1$ both kinds occur: $w=(\frac13,0,\frac23)$ (the outer-pair point of part (a) of Theorem~\ref{thm:three})
 and $w=(\frac13,\frac13,\frac13)$ (the simplex point of the same part) are both optimal.
\end{enumerate}
The statement concerns the priors of optimal tests; it does not say that an optimal reconstruction kernel for $\mu>1$
rebuilds the middle composition exactly.
\end{corollary}
\begin{proof}
(a) is \eqref{eq:priorbound} with $3p_2-2M>0$. (b) If $w_j=0$ the box $B_w$ forces $\alpha_j=0$ in every test, so the point is
feasible for the program of the remaining compositions, whose value is at most $\max\{p_1,p_2\}<\sigma=\mathcal V_2$ by
parts (b) and (c). (c) At $\mu=1$, $\pi_0=\pi_1=e^{-1}$, $\pi_2=\frac1{2e}$, $p_2=\sigma=\frac2{3e}$; the outer-pair point has
$w=(\pi_2,0,\pi_0)/(\pi_0+\pi_2)=(\frac13,0,\frac23)$ and the simplex point has $w\propto(1,\mu^{-1},\mu^{-2})=(1,1,1)$.
\end{proof}

\begin{remark}[Two different dual solutions]
In the three-composition problem, the multipliers observed in the LP solver output for $\mu>1$ (Section~\ref{sec:num}) have a different, knot-dependent structure:
with $r^*$ the maximiser in $p_2$, the fibers $r>r^*+1$ carry $\nu^+_r=-x_r$ and $\lambda^+_{r,r^*}=x_r\pi_r/\pi_{r^*}$ with
$x_r=2\bigl/\bigl(\frac{(r+2)!\,r^*!}{(r^*+2)!\,r!}-1\bigr)$, and the lower fibers a cascade of $\lambda^-$ multipliers. The chain
certificate \eqref{eq:chain} is a different dual-feasible point, uniform in $\mu\ge1$ and in $R$. For every finite $R$ its
budgets satisfy $b_0,b_2\le V$ and $b_1<V$ for $\mu>1$. Once the endpoint prefixes reach the plateau on $H$, $b_0=b_2=V$;
this holds in particular for all sufficiently large $R$, but not for every $R$ (at $\mu=2$, $R=0$ one has
$V=\frac85e^{-2}$ and $b_0=F_0=\frac65e^{-2}<V$). The strict middle budget is what excludes the middle prior. The conditions (i), (ii$^\pm$) and the three budget bounds were also evaluated using exact rational arithmetic at
rational $\mu\ge1$; this is a check of the algebra, not part of the proof.
\end{remark}

\begin{remark}[Four compositions]\label{rem:four}
For $m=3$, $\mu=2$ and $\kappa^2=\frac12$ (the example $p=(\frac12,\frac12)$, $q=(\frac34,\frac14)$, $\theta=\frac12$,
$\lambda=4$), the truncated program gives $\kappa\varphi(0)\mathcal V_3(2)=0.06681$ against the pair values $0.03818$
($k,k{+}1$), $0.06108$ ($k,k{+}2$) and $0.06545$ ($k,k{+}3$): a penalty reappears for four compositions at a value of $\mu$
where three compositions show none. These values were first observed numerically; the exact coefficient and strict pairwise gap now follow from Theorem~\ref{thm:four-compositions}.
\end{remark}

\begin{example}[Statements that are false]\label{ex:counter}
\textup{(1)} ``The coefficient of a finite family equals the largest pairwise coefficient'': false by
Proposition~\ref{prop:fourthirds} (ratio $4/3$) and by Theorem~\ref{thm:three} for $\mu<1$.
\textup{(2)} ``Fiber-preserving kernels are asymptotically optimal'': false already for two hypotheses,
Proposition~\ref{prop:fp} and Example~\ref{ex:fpgap}.
\textup{(3)} ``Collinear scores exclude a penalty'': false with coarse information (the rare channel has collinear scores
$j\kappa Z$), true without it (Corollary~\ref{cor:single}).
\textup{(4)} ``The penalty of Theorem~\ref{thm:three} is a differential-privacy statement'': it is not. It concerns one
reconstruction kernel for a family of compositions; it does not bound the privacy curve of the mechanism or by itself
imply a utility improvement under a global privacy constraint.
\end{example}

\section{General composition windows}\label{sec:windows-general}
The Poisson masses $\pi_k$ are as in Section~\ref{sec:rare}. When removing their common factor $e^{-\mu}$ in proofs, we use
the rescaled scalar masses
\[
 q_k=e^\mu\pi_k=\frac{\mu^k}{k!}\quad(k\ge0),\qquad q_k=0\quad(k<0).
\]

\subsection{Pairwise coefficients and admissible paths}\label{sec:window-paths}
For $1\le d\le m$ put
\begin{equation}\label{eq:general-pairs}
 p_d(\mu)=\max_{k\ge0}\frac{d\pi_k\pi_{k+d}}{\pi_k+\pi_{k+d}},\qquad
 P_m(\mu)=\max_{1\le d\le m}p_d(\mu).
\end{equation}
A pair of compositions $a<a+d$ has coefficient $p_d$: translate the fiber index to $k=r+a$ and subtract the multiple $a\sum_j\alpha_j\pi_{r+j}$ from the score, which does not change its extent. Apply Corollary~\ref{cor:binary} to finite truncations, with score difference $d$, and pass to the limit using Theorem~\ref{thm:rare}. The positive ratios $\pi_{k+d}/\pi_k$ are all distinct because their consecutive ratio is $(k+1)/(k+d+1)<1$. The displayed maximum is attained, since its terms tend to zero and its $k=0$ term is positive.

\begin{definition}[Admissible vertex path]\label{def:path}
Choose integers $0\le j_0<\cdots<j_s\le m$, $s\ge1$, distinct fibers $r_0,\ldots,r_{s-1}$, and positive weights $u_0,\ldots,u_s$ such that, for every $0\le i<s$,
\[
 u_i\pi_{r_i+j_i}=u_{i+1}\pi_{r_i+j_{i+1}}>0,
 \qquad u_h\pi_{r_i+j_h}\le u_i\pi_{r_i+j_i}\quad(0\le h\le s).
\]
The path value is
\[
 H(J,r;\mu)=\frac{\sum_{i=0}^{s-1}(j_{i+1}-j_i)u_i\pi_{r_i+j_i}}
 {\sum_{h=0}^s u_h}.
\]
Write $W_m(\mu)$ for the supremum of these values. Neither endpoint of the window is required to be active.
\end{definition}
Every path is feasible for \eqref{eq:Vm}: normalize its weights to give the prior, use the two tied coordinate vertices at each tie fiber, and use one maximizing coordinate vertex for both tests at all remaining fibers. Every such vertex belongs to the box, and its value in every fiber is bounded by that fiber's maximizing vertex value, so all cross constraints hold. Thus $W_m\le\mathcal V_m$.

\begin{conjecture}[Vertex paths, Hypothesis P]\label{conj:paths}
For every integer $m\ge1$ and every real $\mu>0$,
\[
 \mathcal V_m(\mu)=W_m(\mu).
\]
\end{conjecture}
The supremum defining $W_m$ is a maximum, as proved in Proposition~\ref{prop:path-envelope}. The results below prove the conjecture for $m\le3$ and in the displayed full-chain and endpoint-pair regimes for every $m$. The remaining intermediate regimes for $m\ge4$ are open. An exact certificate at a rational parameter point proves the conjecture at that point, not on a surrounding interval.

\subsection{The exact full-chain phase}\label{sec:small-general}
Put $S_m(\mu)=\sum_{k=0}^m\mu^k$.
\begin{theorem}\label{thm:small-general}
Let $m\ge2$, and let $d_m$ be the unique positive root of
\[
 \sum_{k=2}^{m+1}\mu^k=2.
\]
Then the full-chain value
\[
 \sigma_m(\mu)=e^{-\mu}\left(1-\frac1{S_m(\mu)}\right)
\]
is equal to $\mathcal V_m(\mu)$ if and only if $0<\mu\le d_m$. The thresholds decrease from $d_2=1$ to $\sqrt3-1$. For $\mu>d_m$, the admissible path $(0,2,3,\ldots,m)$ in fibers $(0,-2,-3,\ldots,-(m-1))$ has value strictly larger than $\sigma_m(\mu)$.
\end{theorem}
\begin{proof}
Remove the common factor $e^{-\mu}$ and use the rescaled masses $q_k$. Put
\[
 S=S_m(\mu),\quad V=1-S^{-1},\quad C=\mu^{m+1}/S,\qquad A_k=S^{-1}\sum_{h=k}^m\mu^h\quad(1\le k\le m),\quad A_{m+1}=0.
\]
Then $A_1=V=\mu-C$ and $A_{k+1}=\mu A_k-C$. Suppose first that $\mu(S-1)\le2$. This is equivalent to $(\mu+2)V\le2$, implies $\mu\le1$, and gives $\mu S\le\mu+2\le3$.

Apply Lemma~\ref{lem:dual} with coordinate set $\{0,\ldots,m\}$ and fibers $\{-m,\ldots,R\}$. For every $R\ge0$, set $a_{-m-1}=0$, $a_{-k}=A_k$ for $1\le k\le m$, $a_0=1+V$, and $a_r=r+1$ for $r\ge1$. Take
\[
\begin{gathered}
 \nu^+_r=-r,\qquad \lambda^+_{r,r-1}=\mu\quad(-m<r\le R),\\
 \lambda^-_{r,r+1}=a_r\quad(r<R),\qquad
 \nu^-_r=r-\mu-a_{r-1}\quad(r<R),\qquad \nu^-_R=R-a_{R-1},
\end{gathered}
\]
with all unlisted cross multipliers zero and $\beta^+=0$. For $j\ge1$ put $\beta^-_{-j,j}=V$. For coordinate $j=0$, when $R=0$ put $\beta^-_{0,0}=V$; when $R\ge1$ put $\beta^-_{0,0}=V(1-\mu)$ and $\beta^-_{1,0}=V\mu$. Every other box multiplier is zero. All required multipliers are nonnegative and every coordinate budget is $V$.

Balances (i) follow by cancellation, and the plus rows are equalities by $(r+j)q_{r+j}=\mu q_{r+j-1}$, including the lowest fiber where only $j=m$ has nonzero mass. For $-m\le r<R$ put $v_{r,j}=(a_r-r-j-1)q_{r+j+1}$ and append $v_{-m-1,j}=v_{R,j}=0$. The minus-row slack before adding the box term is $v_{r,j}-v_{r-1,j}$.

For $j=0$ the only possibly nonzero entries are $v_{-1,0}=V$ and, when retained ($R\ge1$), $v_{0,0}=V\mu$; their decreases are exactly the stated box terms. For $1\le j\le m$ the sequence has initial zeros, the single nonnegative entry $A_{j+1}$ at $r=-j-1$ (absent for $j=m$), then negative entries $-h_n$, $n=r+j+1\ge1$, with
\[
 h_n=\begin{cases}(n-A_{j+1-n})q_n,&1\le n\le j,\\(j-V)q_{j+1},&n=j+1,\\j q_n,&n\ge j+2.\end{cases}
\]
The first downward jump is $A_{j+1}+\mu(1-A_j)=V$. For $1\le n<j$, with $k=j+1-n$, the difference $h_n-h_{n+1}$ has the sign of
\[
 D=(n+1)(n-\mu)-nA_k+C.
\]
At $n=1$ one has $k\ge2$, so $A_k\le A_2=\mu V-C$ gives $D\ge2-(\mu+2)V\ge0$. For $n\ge2$, $A_k<1$ and $\mu\le1$ give $D\ge n^2-n-1>0$. From $h_j$ to $h_{j+1}$ the coefficient $j-V$ is unchanged and $q_{j+1}/q_j\le1$. The next difference has the sign of $(j+2)(j-V)-j\mu$: for $j=1$ it is nonnegative by $\mu S\le3$, while for $j\ge2$ it is at least $j^2-2>0$. Thereafter $q_{n+1}/q_n=\mu/(n+1)\le1$. Thus no further downward jump occurs. Truncating a nonpositive branch and appending zero adds no downward jump. All minus rows follow. Weak duality and the truncation inequality yield $\mathcal V_m(\mu)\le e^{-\mu}V$.

For the reverse inequality take unnormalized prior weights $w_j=\mu^{-j}$. In fiber $r=-i$, $0\le i<m$, vertices $i$ and $i+1$ tie at the maximum $\mu^{-i}$, since a coordinate $j\ge i$ has weighted mass $\mu^{-i}/(j-i)!$ and earlier coordinates have mass zero. This full chain is admissible for every $\mu>0$, with value
\[
 e^{-\mu}\frac{\sum_{i=0}^{m-1}\mu^{-i}}{\sum_{j=0}^m\mu^{-j}}=e^{-\mu}V.
\]

For the exact endpoint consider the path $(0,2,3,\ldots,m)$ with weights $1,2\mu^{-2},2\mu^{-3},\ldots,2\mu^{-m}$. In fiber $0$ the first pair ties at $1$, and every later coordinate has weighted mass $2/j!\le1$. All remaining ties in fiber $-j$ are maximal by the same factorial comparison. Hence this path is admissible for every $\mu>0$. With $T=\sum_{k=0}^{m-2}\mu^k$, its value divided by $e^{-\mu}$ is $P=2(\mu^m+T-1)/(\mu^m+2T)$, and direct multiplication gives
\[
 P-V=\frac{\mu^{m-1}[\mu(S-1)-2]}{S(\mu^m+2T)}.
\]
It therefore beats the full chain exactly when $\mu>d_m$. Finally the defining polynomial is strictly increasing on $(0,\infty)$, and increasing $m$ strictly decreases its root. The infinite sum equals $\mu^2/(1-\mu)$, whose value is $2$ at $\mu=\sqrt3-1$, proving the stated limit.
\end{proof}

\begin{corollary}[Exact family-versus-pair law in the full-chain phase]
\label{cor:generalgapphase}
Let $m\ge2$, let $d_m$ be the positive root of
$d_m(S_m(d_m)-1)=2$, and put $\alpha=\sqrt3-1$.
The largest pairwise coefficient $P_m(\mu)$ is independent of $m$ on
$0<\mu\le1$, and for $0<\mu\le d_m$,
\[
 \frac{\mathcal V_m(\mu)}{P_m(\mu)}
 =\begin{cases}
 (1+\mu)\dfrac{1-\mu^m}{1-\mu^{m+1}},&0<\mu\le\alpha,\\[2mm]
 \dfrac{2+\mu^2}{2\mu}\dfrac{1-\mu^m}{1-\mu^{m+1}},&
       \alpha\le\mu\le d_m.
 \end{cases}
\]
The removable value at $\mu=1$ is understood when $m=2$.
This ratio increases on the first interval and decreases on the second.
Its maximum over the full-chain phase occurs at the same value
$\mu=\alpha$ for every $m$, and equals
\[
 \sqrt3\,\frac{1-\alpha^m}{1-\alpha^{m+1}}.
\]
These maxima increase strictly to $\sqrt3$ as $m\to\infty$.
\end{corollary}
\begin{proof}
For $\mu\le1$ the general pair formula gives
\[
 P_m(\mu)=e^{-\mu}\max\left\{\frac\mu{1+\mu},
                                      \frac{2\mu^2}{2+\mu^2}\right\},
\]
with crossover $\alpha$. The exact full-chain theorem gives the displayed
ratio. The first branch can be written
$1+\mu(1-\mu^{m-1})/(1-\mu^{m+1})$; its derivative has the sign of
\[
 1-m\mu^{m-1}+m\mu^{m+1}-\mu^{2m}>0\qquad(0<\mu<1),
\]
because
\[
 \frac{1-\mu^{2m}}{1-\mu^2}
 =\sum_{k=0}^{m-1}\mu^{2k}>m\mu^{m-1}.
\]
Pairing terms with indices $k$ and $m-1-k$ proves the last strict inequality
by $(\mu^k-\mu^{m-1-k})^2\ge0$.
The second branch is the product of the strictly decreasing positive
functions $(2+\mu^2)/(2\mu)$ and $S_{m-1}/S_m$ on $(0,1)$; the latter is
$1-1/(1+\mu^{-1}+\cdots+\mu^{-m})$.
Finally $d_m>\alpha$ since
$\alpha(S_m(\alpha)-1)<\alpha^2/(1-\alpha)=2$.
This proves the maximum statement. The maxima increase with $m$ because
$S_{m-1}/S_m=1-\alpha^m/S_m$ increases, and their limit follows from
$0<\alpha<1$.
\end{proof}

\subsection{An explicit outer-pair regime for every window}\label{sec:outer-general}
For $m\ge2$ put
\[
 p_m(\mu)=\max_{k\ge0}\frac{m\pi_k\pi_{k+m}}{\pi_k+\pi_{k+m}},\qquad
 L_m=\max\left\{2(m+1),\frac{3m}{4}b(b+1)\right\},\quad b=\lceil m/2\rceil.
\]
\begin{theorem}\label{thm:outer-general}
For every $m\ge2$ and $\mu\ge L_m$, one has $\mathcal V_m(\mu)=p_m(\mu)$. In particular, $\mu\ge m^3$ suffices. In the latter range, with $M=\max_k\pi_k$, every feasible point with prior $w$ has objective at most
\[
 p_m(\mu)-\frac M4\sum_{j=1}^{m-1}w_j.
\]
Consequently every optimal prior is supported on the two endpoint compositions.
\end{theorem}
\begin{proof}
Write $V=p_m(\mu)$, $T=2V/m$, $d_k=(1-T/\pi_k)_+$ for $k\ge0$, and $d_k=0$ for $k<0$. The following nonnegative rates will be used at every finite truncation:
\[
 a_r=\max\left\{0,r+1+\frac m2(1+d_{r+m+1}-d_{r+1})\right\}\quad(r\ge-m).
\]
Use Lemma~\ref{lem:dual} with $\nu^+_r=-r$, $\lambda^+_{r,r-1}=\mu$ for $r>-m$, $\lambda^-_{r,r+1}=a_r$ for $r<R$, and
\[
 \nu^-_r=r-\mu-a_{r-1}\ (r<R),\qquad \nu^-_R=R-a_{R-1},\qquad a_{-m-1}=0.
\]
All other cross multipliers are zero. Put $v_{r,j}=(a_r-r-j-1)\pi_{r+j+1}$ for $-m\le r<R$, append zeros at $r=-m-1,R$, and set
\[
 \beta^+_{r,j}=0,\qquad\beta^-_{r,j}=(v_{r-1,j}-v_{r,j})_+.
\]
Column balances cancel and the plus rows follow from $(r+j)\pi_{r+j}=\mu\pi_{r+j-1}$. The minus-row slack before the box term is $v_{r,j}-v_{r-1,j}$, including both boundary rows. Thus only the budgets require estimates.

We first prove the endpoint budgets for every $\mu>0$. The high-atom set $H=\{k:\pi_k>T\}$ is a nonempty interval $[a,b]$ containing $\lfloor\mu\rfloor$, and has length at most $m$. Indeed, $T<M$: equality would require maximal Poisson atoms separated by $m\ge2$, whereas at most two adjacent atoms can be maximal. If two indices separated by $m$ belonged to $H$, their harmonic mean would exceed $T$, contrary to its definition. Define
\[
 F_k=\pi_k(1+d_{k+m}-d_k)\quad(k\ge0),\qquad
 G_k=\pi_k(1-d_k+d_{k-m})\quad(k\ge1).
\]
Both are nonnegative and at most $T$. They equal $T$ on $H$ and equal $\pi_k$ except at the boosted blocks $H-m$ and $H+m$, respectively. For a boosted $F$ entry, its upper bound is equivalent to $T\ge2\pi_k\pi_{k+m}/(\pi_k+\pi_{k+m})$; the same argument applies to $G$.

Here are all shape comparisons, to verify unimodality. At two consecutive boosted $F$ entries $k,k+1$, put $l=k+1\ge1$. Since $\pi_l\le T<\pi_{l+m}$,
\[
 \mu>\left(\prod_{i=1}^m(l+i)\right)^{1/m}
 \ge\sqrt{l(l+m)}\ge\frac{2l(l+m)}{2l+m}.
\]
The product inequality follows by pairing its factors. Since $T<\pi_{k+m+1}$, direct subtraction gives
\[
 F_{k+1}-F_k=\frac{\pi_k}{l}\left[2(\mu-l)-\frac{mT}{\pi_{k+m}}\right]>0.
\]
At the last boosted entry $k=b-m\ge0$, let $l=k+1$. If $l\in H$, the next entry is $T\ge F_k$. Otherwise $\pi_l\le T<\pi_{l+m-1}$, so the same paired-factor argument, now for $m-1$ factors, gives $\mu>2l(l+m)/(2l+m)$. As $b+1\notin H$,
\[
 T/\pi_b\ge\mu/(l+m)\ge2-\mu/l,
\]
which is exactly $F_k\le\pi_l$. Entry into the boosted block preserves increase because boosting is upward. All other comparisons are inherited from the Poisson sequence and the plateau $T$ on $H$. Thus $F$ is unimodal.

At consecutive boosted $G$ entries $k,k+1$, put $l=k-m$. Since $\pi_l>T\ge\pi_k$,
\[
 \mu<\left(\prod_{i=1}^m(l+i)\right)^{1/m}\le l+(m+1)/2.
\]
Hence
\[
 G_k-G_{k+1}=\frac{\pi_k}{k+1}\left[2(k+1-\mu)-\frac{mT}{\pi_l}\right]>0.
\]
At the first boosted entry $k=a+m$, if $k-1\in H$ the previous value is $T\ge G_k$. Otherwise $\pi_a>T\ge\pi_{k-1}$ gives $\mu<a+m/2$ by the arithmetic-mean bound on the $m-1$ intervening factors. Since $a-1\notin H$, $T/\pi_a\ge a/\mu\ge2-k/\mu$, which is exactly $G_k\le\pi_{k-1}$. The exit from the boosted block preserves decrease because $G_k\ge\pi_k$. Thus $G$ is unimodal.

Clipping affects no nonzero entry of coordinate $0$, so $v_{r,0}=(m/2)F_{r+1}$. For coordinate $m$, putting $k=r+m+1\ge1$ gives
\[
 v_{r,m}=-\min\{(m/2)G_k,k\pi_k\}.
\]
The minimum of two nonnegative unimodal sequences is unimodal because its upper-level sets are intersections of intervals; $k\pi_k=\mu\pi_{k-1}$ is unimodal. Restricting either sequence to a finite interval and appending zeros preserves unimodality. Their downward variations are therefore at most $(m/2)T=V$, proving $b_0,b_m\le V$.

For the interior budgets put $\epsilon=1-T/M$. Suppose first that $\mu\ge2(m+1)$ and $\epsilon\le2/(3m)$. Then $T\ge2M/3$, while $\pi_k\le\pi_m\le M/2$ for $0\le k\le m$, so $d_k=0$ for every $k\le m$. Thus clipping separates exactly from the perturbation:
\[
 a_r=a_r^0+\frac m2(d_{r+m+1}-d_{r+1}),\qquad a_r^0=\max\{0,r+1+m/2\}.
\]
For $k=r+j+1\ge0$ and $c=m/2-j$ (with zero extension to $k<0$), the baseline is $v^0_{r,j}=\max\{c,-k\}\pi_k$. If $c\ge0$, it is $c\pi_k$; if $c<0$, it is $-\min\{-c,k\}\pi_k$. The latter magnitude is log-concave, since both $\min\{-c,k\}$ and $\pi_k$ are log-concave on their interval supports. Hence the baseline has one sign and unimodal magnitude, of height at most $|c|M$.

The zero-extended sequence $d$ is log-concave. On its positive support, $\log d=f(\log\pi)$ with $f(x)=\log(1-Te^{-x})$ increasing and concave; applying these two properties to the midpoint inequality for $\log\pi$ proves the claim. Thus each of the two products
\[
 \pi_{r+j+1}d_{r+m+1},\qquad \pi_{r+j+1}d_{r+1}
\]
is nonnegative log-concave, hence unimodal, of height at most $\epsilon M$. Their finite restrictions followed by zero have upward and downward variation at most this height. Subadditivity of downward variation therefore yields, for $1\le j\le m-1$,
\[
 b_j\le |m/2-j|M+m\epsilon M
 \le(m/2-1)M+m\epsilon M\le V,
\]
since the gap to $V=(m/2)(1-\epsilon)M$ is $[1-(3m/2)\epsilon]M$.

Finally let $h=\lfloor\mu\rfloor$, $a=\lfloor m/2\rfloor$, $b=\lceil m/2\rceil$, and $\delta=\mu-h$. The pair $(h-a,h+b)$ and $\prod(1-x_i)\ge1-\sum x_i$ give
\[
 \frac{\pi_{h-a}}M\ge1-\frac{a(a+1)}{2\mu},\qquad
 \frac{\pi_{h+b}}M\ge1-\frac{b(b+1)}{2\mu}.
\]
For the first product use factors $1-(\delta+u)/\mu$, $0\le u<a$; for the second use $1-(u-\delta)/(h+u)$, $1\le u\le b$, with $(u-\delta)/(h+u)\le u/\mu$. A harmonic mean is at least the smaller endpoint, so $\epsilon\le b(b+1)/(2\mu)$. The hypothesis $\mu\ge L_m$ supplies all the preceding conditions. Weak duality for every $R$ and the truncation theorem yield $\mathcal V_m\le V$, while the outer-pair path gives the reverse inequality.

If $\mu\ge m^3$, then $b(b+1)\le m^2$ implies $\epsilon\le1/(2m)$. The displayed interior-budget gap is at least $M/4$, so weak duality gives the stated prior penalty, also after passage to the full absolutely convergent objective series.
\end{proof}

\subsection{Unconditional comparison bounds}\label{sec:window-comparisons}
\begin{proposition}[A family--pair gap approaching two]\label{prop:gap-two-lower}
Let $P_m(\mu)$ denote the largest coefficient obtained by restricting
$\mathcal V_m(\mu)$ to any pair of compositions. For every $m\ge2$,
\[
 P_m(\sqrt2)=e^{-\sqrt2},\qquad
 \frac{\mathcal V_m(\sqrt2)}{P_m(\sqrt2)}
 \ge \frac{2\lfloor m/2\rfloor}{\lfloor m/2\rfloor+1}.
\]
In particular, no constant strictly below two uniformly bounds the
family coefficient by the largest pairwise coefficient over all windows.
\end{proposition}
\begin{proof}
Using the rescaled masses $q_k$, the two-hypothesis formula gives a pair of gap $d$ the coefficient
\[
 e^{-\mu}\max_{r\ge0}\frac{d q_r q_{r+d}}{q_r+q_{r+d}}.
\]
Set $\mu=\sqrt2$. The masses $q_k$ decrease for $k\ge1$, and
$\max_k q_k=\sqrt2$. For $d=1$, the displayed fraction is at most
$\sqrt2/2<1$. For $d=2$, its value at $r=0$ is one; for $r\ge1$ it
is at most $2q_3=2\sqrt2/3<1$. For $d=3$, its value at $r=0$ is
$3\sqrt2/(3+\sqrt2)<1$, while for $r\ge1$ it is at most
$3q_4=1/2$. For $d\ge4$, all $r\ge0$ give the bound
\[
 \frac{d q_r q_{r+d}}{q_r+q_{r+d}}
 \le d q_{r+d}\le d q_d
 =\frac{2^{d/2}}{(d-1)!}\le\frac23<1.
\]
The final sequence decreases because its consecutive ratio is
$\sqrt2/d<1$. Thus the largest pair coefficient is exactly $e^{-\sqrt2}$.

Put $s=\lfloor m/2\rfloor$ and use the active compositions
$0,2,\ldots,2s$, with equal prior weights $1/(s+1)$. For each
$i=0,\ldots,s-1$, tie compositions $2i$ and $2i+2$ in fiber $-2i$.
Here $q_0=q_2=1$ and $q_1=\sqrt2$ is irrelevant because odd
compositions have prior weight zero. At that fiber the earlier active
compositions have mass zero, the tied compositions have mass one,
and all later active compositions have masses at most $q_4=1/6$.
The tie is therefore maximal. Use its two vertices as the tests and use
any maximizing vertex for both tests in all other fibers. This gives a
feasible point of the full program. Each of the $s$ tie fibers contributes
$2e^{-\sqrt2}/(s+1)$, proving the lower bound.
\end{proof}

For $\mu>0$ write $M=\max_{k\ge0}\pi_k$ and define
\[
 P_\infty(\mu)=\max_{0\le k<l}
 \frac{(l-k)\pi_k\pi_l}{\pi_k+\pi_l}.
\]
This maximum exists: its terms with second index $l$ are at most
$l\pi_l$, which tends to zero, whereas the term $(k,l)=(0,1)$
is positive. In particular only finitely many second indices can
contribute to the maximum. The notation does not introduce an infinite
composition experiment; it is an explicit functional of the Poisson atoms.

\begin{proposition}\label{prop:unconditional-comparison-bounds}
For every $m\ge1$ and $\mu>0$,
\[
 \mathcal V_m(\mu)\le
 \min\left\{\mu M,\frac{mM}{2},2P_\infty(\mu)\right\}.
\]
Consequently $\mathcal V_m(\mu)\le2P_m(\mu)$ holds whenever
$0<\mu\le1$, whenever $P_m(\mu)=P_\infty(\mu)$, or whenever
$P_m(\mu)\ge mM/4$. In particular it holds at $\mu=\sqrt2$
for all $m\ge2$, and it also holds throughout
$\mu\ge\lceil m/2\rceil(\lceil m/2\rceil+1)$.
At $\mu=\sqrt2$ the inequality is strict for every finite $m\ge2$.
Together with the even-step lower bound, this proves
\[
 \lim_{m\to\infty}\frac{\mathcal V_m(\sqrt2)}{P_m(\sqrt2)}=2.
\]
\end{proposition}

\begin{proof}
We apply Lemma~\ref{lem:dual} to finite dual certificates, so no approximation of a multiplier
by floating point values is involved. Fix $R\ge0$, set
$a_{-m-1}=0$, and choose any nonnegative rates $a_r$ for
$-m\le r<R$. Use
\[
 \nu^+_r=-r,\quad
 \lambda^+_{r,r-1}=\mu\quad(r>-m),\qquad
 \lambda^-_{r,r+1}=a_r\quad(r<R),
\]
\[
 \nu^-_r=r-\mu-a_{r-1}\quad(r<R),\qquad
 \nu^-_R=R-a_{R-1}.
\]
All other cross multipliers vanish. Put
\[
 v_{r,j}=(a_r-r-j-1)\pi_{r+j+1}\quad(-m\le r<R),
 \qquad v_{-m-1,j}=v_{R,j}=0,
\]
and let
\[
 \beta^+_{r,j}=0,\qquad
 \beta^-_{r,j}=(v_{r-1,j}-v_{r,j})_+.
\]
The column balances cancel exactly. The plus rows are equalities by
$(r+j)\pi_{r+j}=\mu\pi_{r+j-1}$; at $r=-m$, only coordinate
$j=m$ has nonzero mass and the same equality holds. The minus-row
slack before the box term equals $v_{r,j}-v_{r-1,j}$, including
the two boundary rows. Thus the budget in coordinate $j$ is precisely
the downward variation of this finite sequence.

First take $a_r=(r+1+c)_+$ for $0\le c\le m$ and write
$k=r+j+1$. Its nonzero entries satisfy
\[
 v_{r,j}=\max\{c-j,-k\}\pi_k
 =\begin{cases}(c-j)\pi_k,&j\le c,\\
 -\min\{j-c,k\}\pi_k,&j>c.
 \end{cases}
\]
The positive sequence in the first case is unimodal. In the second
case, $k\mapsto\min\{j-c,k\}$ is nonnegative and log-concave on
its interval support: its consecutive ratios decrease until the
constant part is reached, and the ratio at the join lies between its
two neighbors. Multiplication by the log-concave Poisson sequence
preserves log-concavity, hence unimodality. Restricting a one-sign
unimodal sequence to a finite interval and appending zeros has
downward variation equal to its largest absolute entry. The initial
missing entry at $k=0$ for $j=m$ is zero in the negative case and is
nonnegative in the positive case, so it cannot increase this bound.
Consequently the budgets are at most $cM$ for $j\le c$ and
$(j-c)M$ for $j>c$. Choosing $c=m/2$ gives $mM/2$. Choosing
$c=0$ instead gives the sharper budget bound
\[
 \max_{k\ge0}\min\{j,k\}\pi_k
 \le\max_{k\ge0}k\pi_k=\mu M.
\]

For the third bound take a constant rate $a_r=a\ge0$. The sequence
becomes $v_{r,j}=(a-k)\pi_k$. Its positive part and the magnitude
of its negative part have disjoint consecutive supports, in that
order, and each is log-concave: on each support it is a positive
linear sequence multiplied by the Poisson sequence. Each therefore
has a single hump. Its downward variation, with the finite boundary
zeros, is at most
\[
 \max_k(a-k)_+\pi_k+\max_k(k-a)_+\pi_k.
\]
To justify the lowest boundary explicitly, for $j<m$ the sequence
contains its entire support from $k=0$ onwards; for $j=m$ it starts
at $k=1$, omitting the nonnegative entry $a\pi_0$. Omitting such
an initial entry cannot increase downward variation. Truncating the
right tail and appending zero cannot increase the stated sum of the
positive and negative hump heights.

Put $P=P_\infty(\mu)$. The intervals
\[
 I_k=[k-P/\pi_k,\ k+P/\pi_k]\qquad(k\ge0)
\]
intersect pairwise. Indeed, for $k<l$ their intersection criterion is
exactly $P\ge(l-k)\pi_k\pi_l/(\pi_k+\pi_l)$. Their common
intersection is nonempty: any finite collection of pairwise intersecting
closed intervals has nonempty intersection by taking the largest lower
endpoint and smallest upper endpoint; the intersections with the fixed
compact interval $I_0$ then have the finite intersection property.
Every upper endpoint is positive, so the intersection contains an
$a\ge0$. Equivalently one can take the explicit value
\[
 a=\max\{0,\sup_{k\ge0}(k-P/\pi_k)\},
\]
whose defining supremum is finite, because the terms tend to
$-\infty$. For this $a$, $|(a-k)\pi_k|\le P$ for all $k$.
Each hump height is therefore at most $P$, proving the third bound.
Weak duality at every $R$ and monotone truncation prove all three
inequalities for the infinite program.

For $0<\mu\le1$ one has $M=\pi_0$ and
$P_m\ge\pi_0\mu/(1+\mu)$, from the adjacent pair at fiber zero.
Thus $\mathcal V_m/P_m\le1+\mu\le2$. The other two stated
conditions follow immediately from the bounds. At $\mu=\sqrt2$, the rescaled masses satisfy $|(1-k)q_k|\le1$ for every $k$:
the values at $k=0,1,2$ are $1,0,1$, the value at $k=3$ is
$2\sqrt2/3<1$, and the ratio of consecutive positive magnitudes
for $k\ge3$ is
$k\sqrt2/((k-1)(k+1))\le3\sqrt2/8<1$.
Thus $a=1$ lies in every interval $I_k$ with $P=\pi_0$, proving
$P_\infty\le\pi_0$ by the pairwise intersection criterion; the
pair $(0,2)$ gives equality. Hence
$P_m=P_\infty=e^{-\sqrt2}$ for $m\ge2$. The even-step paths
of $s+1$ vertices have ratio $2s/(s+1)$ in windows $m\ge2s$.
Taking $s=\lfloor m/2\rfloor$, the ratio is squeezed between
$2\lfloor m/2\rfloor/(\lfloor m/2\rfloor+1)$ and $2$, proving
the displayed limit.

For strictness take an attained optimum and let $b$ be its largest
active coordinate. If $b=0$ its gain is zero. Otherwise discard the
coordinates above $b$ and the zero-mass fibers below $-b$, obtaining the
same feasible point in the window $\{0,\ldots,b\}$ with $w_b>0$.
Use the constant rate $a=1$. The positive hump of
$(1-k)\pi_k$ has height $\pi_0$ at $k=0$, and its negative hump
has height $\pi_0$ at $k=2$. Coordinate $b$ starts at $k=1$
and omits the positive hump entirely. Its budget is therefore at
most $\pi_0$, while every other budget is at most $2\pi_0$.
It follows that the value is at most
$2\pi_0-\pi_0w_b<2\pi_0=2P_m$.

Finally put $c=\lceil m/2\rceil$, $b=m-c$ and $\ell=\lfloor\mu\rfloor$,
and suppose $\mu\ge c(c+1)$. Then $\ell\ge b$ and
$M=\pi_\ell$. The elementary product inequality
$\prod_i(1-x_i)\ge1-\sum_i x_i$ for $0\le x_i\le1$ follows
inductively from $(1-x)(1-y)=1-x-y+xy\ge1-x-y$.
For the atom to the left of the mode it gives
\[
 \frac{\pi_{\ell-b}}M
 =\prod_{h=0}^{b-1}\frac{\ell-h}{\mu}
 \ge1-\frac{b(b+1)}{2\mu}\ge\frac12,
\]
where $\mu-\ell<1$ bounds each deficit by $(h+1)/\mu$.
For the right atom,
\[
 \frac{\pi_{\ell+c}}M
 =\prod_{t=1}^{c}\frac{\mu}{\ell+t}
 \ge1-\frac{c(c+1)}{2\mu}\ge\frac12,
\]
since $(\ell+t-\mu)/(\ell+t)\le t/\mu$.
These atoms have gap $m$, so their pair coefficient is at least
$mM/4$. The already proved centered-ramp bound completes this last
unconditional factor-two comparison regime; path optimality in that regime remains a separate question.
\end{proof}

\begin{proposition}[A sharp comparison for the path envelope]\label{prop:path-envelope}
Let $W_m(\mu)$ be the maximum value of admissible paths, and let
$P_m(\mu)$ be the largest pair coefficient. Then, for every $m\ge1$
and $\mu>0$,
\[
 P_m(\mu)\le W_m(\mu)<2P_m(\mu).
\]
The constant two cannot be decreased uniformly in $m$ and $\mu$.
Consequently the path conjecture would imply the sharp universal bound
$\mathcal V_m(\mu)<2P_m(\mu)$.
\end{proposition}
\begin{proof}
Every pair is a path. Normalize the weights of any path to sum to one,
and write $c_i=w_{j_i}\pi_{r_i+j_i}
=w_{j_{i+1}}\pi_{r_i+j_{i+1}}$. Its edge contribution is
\[
 (j_{i+1}-j_i)c_i
 =\frac{(j_{i+1}-j_i)\pi_{r_i+j_i}\pi_{r_i+j_{i+1}}}
 {\pi_{r_i+j_i}+\pi_{r_i+j_{i+1}}}
 (w_{j_i}+w_{j_{i+1}})
 \le P_m(\mu)(w_{j_i}+w_{j_{i+1}}).
\]
The sum of the incident weights is
$2-w_{j_0}-w_{j_s}<2$. For attainment, take the adjacent pair value in fiber zero
$L=e^{-\mu}\mu/(1+\mu)>0$ and an integer $B\ge\lceil\mu\rceil$
with $m\mu^B/B!<\mu/(1+\mu)$. If the first translated fiber is at
least $B$, its numerator contribution per unit first weight is below $L$;
deleting that vertex strictly improves a path of value at least $L$.
Restriction preserves the remaining ties. Repeated deletion and translation
leave a path whose first fiber is below $B$. All later fibers are smaller:
the Poisson likelihood ratio between increasing compositions decreases
strictly with the fiber. Thus finitely many paths attain the global maximum,
and strictness persists at the maximum. The even-step paths at $\mu=\sqrt2$ from Proposition~\ref{prop:gap-two-lower} have ratio $2s/(s+1)\to2$.
\end{proof}

For $m\le3$, the solved window formulas are attained by paths. Combining them with Proposition~\ref{prop:path-envelope} gives $\mathcal V_m(\mu)<2P_m(\mu)$ for every $\mu>0$ in those windows. Proposition~\ref{prop:unconditional-comparison-bounds} supplies the additional unconditional comparison regimes for arbitrary $m$.

\begin{proposition}[Strict family gain throughout the unit interval]\label{prop:gap-unit}
For every $m\ge3$ and $0<\mu\le1$,
\[
 \mathcal V_m(\mu)>P_m(\mu),\qquad
 P_m(\mu)=e^{-\mu}\max\left\{
 \frac{\mu}{1+\mu},\frac{2\mu^2}{2+\mu^2}\right\}.
\]
\end{proposition}
\begin{proof}
For $0<\mu\le1$, the masses $q_k$ decrease for $k\ge0$. Hence the
pair coefficient of every fixed gap is maximized at $r=0$. Gaps
$d\ge3$ give at most $d q_d\le\mu^3/2\le\mu^2/2$, whereas gap two
gives $2\mu^2/(2+\mu^2)\ge2\mu^2/3$. This proves the pair formula.
The full chain is feasible for every $\mu>0$: choose weights
proportional to $\mu^{-j}$ and tie $j,j+1$ in fiber $-j$; later
vertices have relative values $1/(k-j)!$ and earlier ones have value
zero. Its coefficient, writing $S_m=1+\mu+\cdots+\mu^m$, is
$e^{-\mu}(1-1/S_m)$. This increases strictly in $m$.
For $0<\mu<1$, the $m=2$ chain exceeds each pair term, by direct
subtraction (or the proved three-composition theorem). At $\mu=1$,
the chain is $e^{-1}m/(m+1)>2e^{-1}/3$ for $m\ge3$.
\end{proof}

\subsection{Four compositions: the complete phase diagram}\label{sec:m3-phases}
Let $\alpha$ be the positive root of $x^4+x^3+x^2-2$,
let $\beta\in(1.7097797257,1.7097797258)$ be a root of
$x^5-3x^4-2x^3+6x^2-12x+24$, and let
$\gamma\in(1.9423645569,1.9423645570)$ be a root of
$2x^5-3x^4+6x^3-6x^2-36x+36$. Put $\delta=\sqrt[3]{12}$ and
\[
\begin{aligned}
 C(x)&=\frac{x+x^2+x^3}{1+x+x^2+x^3},&
 A(x)&=\frac{2x(x^2+1)}{x^3+2x+2},\\
 B(x)&=\frac{x^2(x^2+4)}{x^3+2x^2+4},&
 D(x)&=\frac{2x^2(x^2+3)}{x^3+6x+12},\\
 E_r(x)&=\frac{3x^{r+3}}{r!\,[x^3+(r+1)(r+2)(r+3)]}\quad(r\ge0).
\end{aligned}
\]
The isolating intervals above uniquely select $\beta$ and $\gamma$;
each of their defining polynomials has another positive root.

\begin{theorem}[Four compositions; computer-assisted phase theorem]\label{thm:four-compositions}
For every real $\mu>0$,
\[
 e^\mu\mathcal V_3(\mu)=
 \begin{cases}
 C(\mu),&0<\mu\le\alpha,\\
 A(\mu),&\alpha\le\mu\le\sqrt2,\\
 B(\mu),&\sqrt2\le\mu\le\beta,\\
 E_0(\mu),&\beta\le\mu\le\gamma,\\
 D(\mu),&\gamma\le\mu\le\delta,\\
 \displaystyle\max_{r\ge0}E_r(\mu),&\mu\ge\delta.
 \end{cases}
\]
Adjacent formulas agree at their common endpoints. Each value is attained
by a path of vertex tests. If $p_d$ denotes the coefficient for two
compositions separated by $d$, then
\[
 \mathcal V_3(\mu)>\max\{p_1(\mu),p_2(\mu),p_3(\mu)\}
 \quad\Longleftrightarrow\quad
 \mu\in(0,\beta)\cup(\gamma,\delta).
\]
\end{theorem}
\begin{proof}
The matching lower bounds use the following feasible paths:
\[
\begin{array}{c|c|c}
\text{value}&\text{active compositions}&\text{tie fibers}\\\hline
C&(0,1,2,3)&(0,-1,-2)\\
A&(0,2,3)&(0,-2)\\
B&(0,1,3)&(1,-1)\\
E_0&(0,3)&(0)\\
D&(0,2,3)&(1,-1)\\
E_r&(0,3)&(r).
\end{array}
\]
For $a<b$, the positive likelihood ratio
$\pi_{r+b}/\pi_{r+a}$ strictly decreases with $r$. Descending tie
fibers therefore make the tied vertices maximal over the active
compositions. The weights and displayed rational values follow directly
from the tie equations.

We describe the status and extent of the upper-bound verification
explicitly. For the full-chain phase $0<\mu\le\alpha$, the general
certificate applies because $\mu^2+\mu^3+\mu^4\le2$. For the remaining
phases up to $\mu=5$, the upper bounds are proved by exact finite
polynomial verification of explicit rational dual certificates.
For every $\mu\ge5$, the analytic clipped-certificate argument below
supplies the unbounded final branch. Neither step assumes the path
conjecture.

The ancillary directories \texttt{anc/m3\_low/} and \texttt{anc/m3\_high/} are part of the source submission. All proof-sign decisions use only Python integer arithmetic and \texttt{fractions.Fraction}; no numerical tolerance or floating-point LP output enters the proof. Run the following commands from the submission directory:
\begin{quote}\ttfamily
python anc/m3\_low/run\_all\_m3.py\\
python anc/m3\_high/verify\_clipped.py
\end{quote}

For the finite verification remove the common factor $e^{-\mu}$ and use the rescaled masses $q_k$. Given nonnegative rates
$a_r$, put
\[
 v_{r,j}=(a_r-r-j-1)q_{r+j+1},\qquad v_{-4,j}=0.
\]
The chain balances and plus rows hold identically. It remains to verify
nonnegative $\beta_{r,j}$ satisfying
\[
 \beta_{r,j}\ge v_{r-1,j}-v_{r,j},\qquad
 \sum_r\beta_{r,j}\le V\quad(j=0,1,2,3).
\]
The five exact verifiers in \texttt{anc/m3\_low/} establish these
inequalities on $[\alpha,12/5]$ using $528$ nonidentity rational
inequalities. Every numerator and denominator sign is checked by exact
integer polynomial arithmetic; denominators are also certified to have
no zero in the relevant interval. Algebraic endpoints are isolated by
Sturm root counts. After division by their exact factors, the
substitution $x=(a+bt)/(1+t)$ produces polynomials with nonnegative
integer coefficients. The rate tails there are $a_r=r+1$, so the
verified prefixes extend to every sufficiently large truncation without
additional budget.

On $[12/5,5]$ use $V=\max_r E_r(\mu)$, $T=2V/3$,
$d_k=(1-T/q_k)_+$ for $k\ge0$ and $d_k=0$ for $k<0$, and the clipped rates
\[
 a_r=\max\left\{0,\ r+\frac52+
           \frac32(d_{r+4}-d_{r+1})\right\}.
\]
The independent verifier
\texttt{anc/m3\_high/verify\_clipped.py} partitions this interval at
exactly isolated algebraic branch points and checks $756$ nonidentity
rational inequalities and $217$ identities on $20$ cells. Its $1512$
integer polynomial sign checks pass without subdivision. The positive
tail variations are included by telescoping; no tail mass is discarded.
The terminal-row balances and budgets hold at every sufficiently large
truncation. Weak duality and the truncation theorem therefore yield
all the upper bounds through $\mu=5$.

For completeness, the analytic sufficient conditions for the same
clipped certificate at $\mu\ge5$ can be verified without computation.
Its endpoint budgets are at most $V$ for every positive mean. Put
$M=\max q_k$ and $\epsilon=1-T/M$. The interior variation estimate gives
\[
 b_1,b_2\le\frac M2+3\epsilon M<V
 \quad\text{if}\quad d_0=d_1=d_2=d_3=0,
                    \quad\epsilon<\frac29.
\]
Here $q_6/q_3=\mu^3/120>1$ proves $T>q_3$, hence the required vanishing.
Let $n=\lfloor\mu\rfloor\ge5$. The outer pairs
$(n-1,n+2)$ and $(n-2,n+1)$ imply
\[
 \frac TM\ge\max\left\{
 \frac{2n\mu^2}{n(n+1)(n+2)+\mu^3},\
 \frac{2n(n-1)\mu}{n(n-1)(n+1)+\mu^3}\right\}.
\]
The first expression increases on $[n,n+1]$, and the second decreases.
Split at $n+1/3$. At that point both exceed $7/9$: after clearing
positive denominators their excess numerators, with $z=n-5\ge0$, are
\[
 4z^3+44z^2+\frac{377}{3}z+\frac{758}{27},\qquad
 4z^3+41z^2+\frac{326}{3}z+\frac{488}{27}.
\]
Thus $\epsilon<2/9$ throughout $[5,\infty)$, completing the upper bound.

Finally, on each genuine-family phase the exact polynomial comparisons
with every two-vertex path are strict. A nonzero polynomial with
nonnegative coefficients is positive at every positive argument;
the removed endpoint factors have no interior zero. At $\alpha$ and
$\sqrt2$, exact endpoint substitution still gives strict inequalities
against all pairs. At $\beta,\gamma,\delta$ and throughout the
outer-pair phases a pair attains the family value. This proves the
stated criterion for the gap.
\end{proof}

The final branch has an explicit infinite sequence of fiber transitions.
For $r\ge0$, let $\theta_r$ be the unique positive root of
\[
 x^4-(r+1)x^3+(r+1)(r+2)(r+3)x
 -(r+1)(r+2)(r+3)(r+4)=0.
\]
Then $\theta_r$ strictly increases, $r+1<\theta_r<r+4$, and
$\theta_0<\delta$. Thus $E_1$ is optimal on
$[\delta,\theta_1]$, while $E_r$ is optimal on
$[\theta_{r-1},\theta_r]$ for $r\ge2$.
Indeed, the denominator sequence
$r!x^3+(r+3)!$ is strictly log-convex, by factorial log-convexity and
Cauchy--Schwarz, so $E_r$ is strictly log-concave in $r$. The ratio
$E_{r+1}/E_r$ equals one at the displayed root. Its polynomial has
strictly positive derivative because, writing $a=r+1$,
\[
 4x^3-3ax^2+(r+1)(r+2)(r+3)
 \ge (r+1)(r+2)(r+3)-\frac{a^3}{4}>0.
\]
The first five transitions have the following rational isolating intervals:
\[
\begin{array}{c|c}
r&\text{interval containing }\theta_r\\\hline
1&(3.2406678406,3.2406678407)\\
2&(4.2935634536,4.2935634537)\\
3&(5.3274307991,5.3274307992)\\
4&(6.3513269327,6.3513269328)\\
5&(7.3692158614,7.3692158615).
\end{array}
\]

\subsection{Toward Hypothesis P: path knots and interpolation}\label{sec:path-knots}
For a path whose first active composition is $j_0>0$, the shifts
$j'_i=j_i-j_0$ and $r'_i=r_i+j_0$ preserve the masses, ties, score differences and path value.
Thus shifting the active compositions to zero reduces the general case to $j_0=0$ below.

\begin{proposition}[Path cuts and the knots of a chain certificate]
\label{prop:path-dual-knots}
Let $J=(j_0,\ldots,j_s)$, $0=j_0<\cdots<j_s\le m$, be an
admissible path, with tie fibers $r_0>\cdots>r_{s-1}$, positive
unnormalized weights $w_{j_i}$, and value $V$. Put
\[
 T_i=w_{j_i}\pi_{r_i+j_i}
     =w_{j_{i+1}}\pi_{r_i+j_{i+1}},\quad
 W_i=\sum_{h=0}^{i}w_{j_h},\quad
 N_i=\sum_{h=0}^{i-1}(j_{h+1}-j_h)T_h.
\]
Thus $V=N_s/W_s$. Define
\begin{equation}\label{eq:path-knot-explicit}
 z_i=j_i+\frac{VW_i-N_i}{T_i}\qquad(0\le i<s).
\end{equation}
The following statements hold.
\begin{enumerate}[label=\textup{(\alph*)}]
\item If the path maximizes the value among all admissible paths in the
window, then
\[
 j_i\le z_i\le j_{i+1}\qquad(0\le i<s).
\]
In particular $r_i+z_i\ge0$. This assertion assumes optimality only
among paths, and does not assume the path conjecture.
\item Fix a truncation $R\ge\max\{0,r_0\}$, and consider the dual
ansatz of Lemma~\ref{lem:dual} with $\beta^+=0$,
$\nu^+_r=-r$, $\lambda^+_{r,r-1}=\mu$, and
$\lambda^-_{r,r+1}=a_r\ge0$, with the endpoint conventions and
$\nu^-_r$ determined by balance. Set
\[
 v_{r,j}=(a_r-r-j-1)\pi_{r+j+1},\qquad
 v_{-m-1,j}=v_{R,j}=0,\qquad
 \beta^-_{r,j}=(v_{r-1,j}-v_{r,j})_+.
\]
If this certificate has every coordinate budget at most $V$ and
attains the value of the given path, then its rates at all tie knots
are forced:
\begin{equation}\label{eq:path-forced-rates}
 a_{r_i-1}=r_i+z_i.
\end{equation}
Equivalently,
\[
 z_0=\frac{V}{\pi_{r_0}},\qquad
 z_i=j_i+
 \frac{V+(z_{i-1}-j_i)\pi_{r_{i-1}+j_i}}{\pi_{r_i+j_i}}
 \quad(1\le i<s).
\]
These identities do not determine the rates away from the tie knots.
\end{enumerate}
\end{proposition}
\begin{proof}
For (a), retain only the prefix ending at $j_i$. A one-vertex prefix or suffix has value zero; a restriction with at least two vertices is admissible by
restriction, with value $N_i/W_i$, so $N_i\le VW_i$. Retain instead
the suffix starting at $j_{i+1}$. It too is admissible, and translating
its first composition to zero preserves all masses, ties and score
differences. Hence
\[
 N_s-N_i-(j_{i+1}-j_i)T_i\le V(W_s-W_i).
\]
Using $N_s=VW_s$ gives
$VW_i-N_i\le(j_{i+1}-j_i)T_i$. Both claims about $z_i$ follow;
$r_i+j_i\ge0$ because the tie has positive mass.

For (b), in each fiber choose the smallest maximizing active
composition for the minus test and the largest for the plus test;
if the fiber maximum is zero, choose $j_s$ for both. Strict decrease
of each positive likelihood ratio
\[
 \frac{\pi_{r+k}}{\pi_{r+j}}
 =\frac{\mu^{k-j}}{\prod_{h=j+1}^{k}(r+h)}\quad(j<k)
\]
shows that the minus test uses $j_i$ precisely for
$r_i\le r<r_{i-1}$, with the two endpoint intervals extended to the
ends of the truncation. It also shows that the chosen plus test in
fiber $r$ maximizes in fiber $r-1$, and the chosen minus test
maximizes in fiber $r+1$. Thus the nearest-neighbor cross slacks
vanish whenever their multipliers are nonzero.

In the weak-duality identity, equality at the path makes every
active-coordinate budget equal to $V$, every selected minus-row
slack zero, and every box multiplier at an unselected active
coordinate zero. With the displayed minimal choice of $\beta^-$,
this says
\[
 v_{r-1,j_i}-v_{r,j_i}\ge0\quad\hbox{on the interval where $j_i$
 is selected},\qquad
 v_{r-1,j_i}-v_{r,j_i}\le0\quad\hbox{elsewhere}.
\]
The active budget therefore telescopes on that one interval. Writing
$z_i'=a_{r_i-1}-r_i$ gives
\[
 V=z_0'\pi_{r_0},\qquad
 V=(z_i'-j_i)\pi_{r_i+j_i}
   -(z_{i-1}'-j_i)\pi_{r_{i-1}+j_i}\quad(1\le i<s).
\]
Multiply the latter identity by $w_{j_i}$ and use the tie identity
at $r_{i-1}$. Induction gives
$T_i(z_i'-j_i)=VW_i-N_i$, proving
\eqref{eq:path-knot-explicit} and \eqref{eq:path-forced-rates}.
\end{proof}

\begin{remark}[The unresolved interpolation step]
\label{rem:path-dual-interpolation}
The preceding proposition supplies the forced knot values for a
maximizing path. A sufficient condition for the path conjecture at
$(m,\mu)$ is the existence, for arbitrarily large $R\ge\max\{0,r_0\}$, of nonnegative rates interpolating these
knots such that every active sequence $v_{\cdot,j_i}$ decreases
only on the interval where $j_i$ is the minus-test composition and
increases elsewhere, and every inactive coordinate satisfies
\[
 \sum_{r=-m}^{R}(v_{r-1,j}-v_{r,j})_+\le V.
\]
Indeed, the active budgets then telescope to $V$, the inactive
budgets obey the displayed bound, and Lemma~\ref{lem:dual} applies. For the last active coordinate, the terminal telescoping identity follows from $VW_s=N_s$. For a certificate uniform in $R$, one may use an eventual
tail $a_r=r+1+c$, with $c\ge0$, starting at $K\ge\lceil\mu\rceil$:
$v_{r,j}=(c-j)\pi_{r+j+1}$ has constant sign and decreasing absolute
value, so its variation together with its final jump to zero is
independent of the truncation endpoint. Existence of the required
interpolation for an arbitrary maximizing path has not been proved.
\end{remark}

\section{Numerical checks}\label{sec:num}
All numbers below were computed with Python 3.13, NumPy 2.5.3 and SciPy 1.18.1, with HiGHS called through
\texttt{scipy.optimize.linprog} at primal and dual feasibility tolerances $10^{-10}$ for every program quoted in this
section. ``Finite-experiment LP'' means the exact finite minimax
formulation \eqref{eq:def} over all common kernels, solved in floating point; it is not exact arithmetic. For every
such solve we recomputed the largest total-variation error of the returned kernel from the kernel itself (``actual'') and
its row-sum error; across all recorded runs the row-sum error is at most $2.8\cdot10^{-13}$ and the smallest kernel entry
is at least $-9.8\cdot10^{-11}$. Exact-arithmetic checks (Python fractions, no solver) were performed for the rational identity of
Proposition~\ref{prop:fourthirds} and for the conditions (i), (ii$^\pm$) and the three budget bounds of the chain certificate
\eqref{eq:chain} at $17$ rational
values $\mu\in\{1,\frac{101}{100},\frac{11}{10},\frac54,\frac32,\frac74,2,\frac52,3,\frac72,4,5,\frac{13}2,8,10,\frac{25}2,20\}$, each at
all seven truncation levels $R\in\{0,1,2,3,5,8,60\}$ ($119$ combinations; the common factor $e^{-\mu}$ cancels in every
condition, and the maxima defining $V$ and $M$ are searched independently of $R$), together with a negative control at
$\mu=\frac{99}{100}$, which fails as it must. These are checks of algebra that is proved in the text, not substitutes for the proofs.
The dual certificates of Theorem~\ref{thm:three}(c) are proved analytically; their inequalities were also evaluated at
sample points in floating point.

\paragraph{Two hypotheses, two fibers (Example~\ref{ex:fpgap}).}
Tangent value $\frac9{50}\varphi(0)=0.0718096$; fiber-preserving limit $\frac6{25}\varphi(0)=0.0957461$. Finite-experiment LP,
$\sqrt N\Delta_N$ (objective; the actual error agrees to $8\cdot10^{-9}$ or better): $0.070918$ ($N{=}20$), $0.071362$
($40$), $0.071586$ ($80$), $0.071697$ ($160$), $0.071754$ ($320$); fiber-preserving LP: $0.094557$, $0.095150$, $0.095447$,
$0.095597$, $0.095671$.

\paragraph{Three hypotheses, two fibers (Proposition~\ref{prop:fourthirds}).}
Collinear LP value $0.24000000000000007\,\varphi(0)$ in floating point, $\frac6{25}\varphi(0)=0.0957461$ exactly; finite-experiment
LP $\sqrt N\Delta_N$: $0.095150$ ($N{=}40$), $0.095447$ ($80$), $0.095597$ ($160$), actual errors within $4\cdot10^{-9}$.

\paragraph{Non-collinear scores ($d=3$).}
$p=(0.4,0.35,0.25)$, coarse rows $(0.5,0.5)$, $(0.7,0.3)$, $(0.35,0.65)$, changed-message laws $(0.6,0.25,0.15)$,
$(0.3,0.5,0.2)$ in fiber~1 and $(0.5,0.2,0.3)$, $(0.2,0.4,0.4)$ in fiber~2 (hypothesis $0$ is $p$ in both). The scores
$H_{\vartheta,y}=\langle v_{\vartheta,y},Z\rangle$, $Z\sim N(0,I_2)$, are not collinear; the program \eqref{eq:tangent} was
discretised over $K$ equally spaced directions of $Z/|Z|$. To specify the finite sums, the basis of $L^2_0(p)$,
listed in native-symbol order, was
\[
 \phi_1=(\sqrt{7/6},-\sqrt{32/21},0),\qquad
 \phi_2=(1/\sqrt3,1/\sqrt3,-\sqrt3).
\]
The directions were $u_k=(\cos(2\pi k/K),\sin(2\pi k/K))$, $k=0,\ldots,K-1$, with weight
$\sqrt{\pi/2}/K$ from the radial expectation $\E|Z|=\sqrt{\pi/2}$. The resulting Riemann sums, rather than continuum values, are:
$0.0718599$ ($K{=}64$), $0.0718709$ ($256$), $0.0718697$ ($1024$). No rigorous bound on the angular error is claimed.
Finite-experiment LP $\sqrt N\Delta_N$: $0.0720643$ ($N{=}12$), $0.0720363$ ($24$), $0.0719684$ ($48$); actual errors
$0.0720643$, $0.0720363$, $0.0719684$ (this is the run with the largest recorded row-sum error, $2.75\cdot10^{-13}$, and
the smallest kernel entry, $-9.78\cdot10^{-11}$). Pairwise values (64 directions): $0.04775$, $0.05004$, $0.05985$.

\paragraph{Rare channel, three compositions.}
$p=(\frac12,\frac12)$, $q=(\frac34,\frac14)$, $\theta=\frac12$, $\kappa^2=\frac12$, $k_n=\lfloor n/2\rfloor$,
$\lambda\in\{1,4\}$ ($\mu\in\{\frac12,2\}$); fibers $r\le18$ kept in the finite-$n$ programs (omitted coarse mass
$\le5\cdot10^{-13}$, so the programs optimise a truncated experiment whose difference from the full one is below the
displayed precision).
\begin{table}[htbp]\centering\small
\begin{tabular}{lrrrrrrr}
\toprule
$\mu$ & prediction & $n{=}60$ & $120$ & $240$ & $480$ & $960$ & $1920$\\
\midrule
$\tfrac12$ & $\kappa\varphi(0)\sigma=0.073328$ & $0.077159$ & $0.075388$ & $0.074460$ & $0.073964$ & $0.073695$ & $0.073547$\\
$\tfrac12$ & simplex-test bound & $0.076416$ & $0.074889$ & $0.074113$ & $0.073722$ & $0.073525$ & $0.073427$\\
$2$ & $\kappa\varphi(0)p_2=0.061084$ & $0.063819$ & $0.062494$ & $0.061799$ & $0.061444$ & $0.061265$ & ---\\
\bottomrule
\end{tabular}
\caption{$\sqrt n\,\Delta^{(2)}_n$ from the finite-experiment LP (objective; actual errors agree to $9\cdot10^{-9}$). The
``simplex-test bound'' row is the value of the tests of Theorem~\ref{thm:three}(a) evaluated with the exact finite-$n$
weights and conditional total variations, a rigorous lower bound at every $n$. The $n=1920$ computation used $20349$ states and took $889$ s. At $\mu=\frac12$ the three pairwise limits are
$0.057033$, $0.057033$, $0.038022$; at $\mu=2$ the prediction is Theorem~\ref{thm:three}(d).}\label{tab:rare}
\end{table}

\paragraph{The program across $\mu$.}
The truncated program \eqref{eq:Vm} with $R=40$ agrees with $\max\{\sigma(\mu),p_2(\mu)\}$ to within $3.4\cdot10^{-11}$ at
$\mu\in\{0.2,0.3,\dots,1.0,1.05,1.1,1.2,1.3,1.5,2,3,5\}$, and with $R=60$ to within $6\cdot10^{-11}$ at $\mu=1.02$ and at the
first six transitions of the maximising atom of $p_2$ (roots of $\mu^3-(r+1)\mu^2+(r+1)(r+2)\mu-(r+1)(r+2)(r+3)=0$,
$r=0,\dots,5$, i.e.\ $\mu=1.776,2.850,3.886,4.907,5.922,6.933$) with offsets $\pm0.02$; for $\mu\le1$ the predicted value is
Theorem~\ref{thm:three}(c), for $\mu\ge1$ it is Theorem~\ref{thm:three}(d). These solver values were obtained before the
proof of part (d) and are reported as a consistency check of the solver against the theorem, not as evidence for it. The
optimal priors returned by the solver ($R=40$) have $w_1=0$ at
$\mu\in\{1.05,1.1,1.2,1.3,1.5,2,3,5\}$, all three weights positive at $\mu\in\{0.5,0.8,0.95\}$ (at $\mu=0.5$,
$w=(\frac17,\frac27,\frac47)\propto(1,\mu^{-1},\mu^{-2})$), and $w=(\frac13,\frac13,\frac13)$ at $\mu=1$, in agreement with
Corollary~\ref{cor:priors} (a consistency check, not a verification of the corollary). Additional computations gave the four-composition value of
Remark~\ref{rem:four} and the non-collinear pairwise values ($0.04775$, $0.05004$, $0.05985$ at 64 directions). The analytic certificates were re-evaluated numerically: the first is valid at
$\mu\in\{0.1,\dots,0.955\}$ and fails at $0.96$ in the single inequality (ii$^-$), $r=0$, $j=1$, as it must; the second is
valid at every tested $\mu\in[0.1,1]$ and fails below $\mu\approx0.08$ in (ii$^-$), $r=1$, $j=1$, consistent with the
analytic range $[\frac12,1]$ used in the proof.

\paragraph{Evidence for Hypothesis P.}
For $m=2,\ldots,8$ and $\mu=i/20$, $i=1,\ldots,160$, all 1120 parameter points have an exact rational dual certificate
matching the exact maximum over admissible vertex paths. Numerical LPs at truncations $R=40,50,60$ agree with those values.
These are certificates at the stated parameter points, not a proof on a continuum of means or for all window sizes.
The certificates and two exact checkers are included in \texttt{anc/grid/}; from the submission root run
\begin{quote}\small\ttfamily
python anc/grid/verify\_saved\_grid.py\\
python anc/grid/root\_audit\_certificates.py
\end{quote}
The second checker verifies two finite truncations for each saved certificate; the common eventual-tail formula accounts
for the remaining fibers. Both commands check all 1120 certificates using exact rational arithmetic. The
computer-assisted continuum proof for four compositions uses the separate verifiers in \texttt{anc/m3\_low/} and
\texttt{anc/m3\_high/}; \texttt{anc/README.md} lists all verification commands.

\section{Prior work and what is new}\label{sec:lit}
\begin{itemize}
\item \textbf{Finite parameter sets.}  Torgersen \cite[Thm.\ 1, p.\ 222; Thm.\ 2, p.\ 223]{Torgersen1970} characterises
$\varepsilon$-deficiency on a finite parameter set by Bayes risks of bounded-loss decision problems and by sublinear
(support) functionals; the randomisation criterion is Le Cam's \cite{LeCam1964}. Lemma~\ref{lem:identity} is this
characterisation in LP form. The distinction between pairwise and family comparison is also classical there.
\item \textbf{Exposed faces.}  The directional derivative of a support function is the support function of the exposed
face (Danskin \cite{Danskin1967}; Bertsekas \cite[\S5.4.4]{Bertsekas}). The program \eqref{eq:tangent} is built from
such faces; the identity itself is not claimed as new, and Proposition~\ref{prop:abstract} shows that the passage from
weakly convergent scores to the optimised face program needs nothing beyond the support-functional criterion, this
derivative and the compactness of convex bodies. What is specific to this paper is the identification of the scores and of
the rate in the growing shuffle experiments (Section~\ref{sec:bundle}, Lemma~\ref{lem:jscore}), the treatment of countably
many fibers in Theorem~\ref{thm:rare}, and the explicit evaluations.
\item \textbf{Two hypotheses.}  The testing characterisation of binary deficiency is in Torgersen's paper and, in the form
$\|\rho_1-t\rho_2\|_1\ge\|\sigma_1-t\sigma_2\|_1-2(1+t)\varepsilon$ for all $t\ge0$, in Jen\v cov\'a \cite[Thm.\ 5, \S4.1]{Jencova2012}.
Corollary~\ref{cor:binary} recovers it from the program.
\item \textbf{Deficiency of a statistic in smooth models.}  Skovgaard \cite[\S4, (4.5)--(4.6)]{Skovgaard1981} defines exactly
the loss $\inf_\Pi\sup_\beta\TV(P_\beta,\Pi Q_\beta)$ over common kernels for the reduction of an iid sample to the MLE, bounds it
by the square root of the relative Fisher-information loss, and obtains $O(n^{-1/2})$, and $O(n^{-1})$ after adjoining the
observed information; Michel \cite{Michel1978} proves the latter more generally (not inspected here beyond Skovgaard's
account). Pollard \cite[Thm.\ 7 and \S4]{Pollard2012} explains under differentiability in quadratic mean how an
insufficient statistic can preserve Fisher information, and establishes asymptotic equivalence of the full and reduced
product experiments in the Kagan--Shepp example. These are smooth-model results; the present paper gives exact
coefficients in finite-alphabet shuffle models with a finite family, where the loss comes from a non-smooth coarse
component (the tag) interacting with a local Gaussian-polynomial component.
\item \textbf{Local comparison of experiments.}  Exact optimised first-order deficiency coefficients are classical.
Torgersen's institutional reports \cite{Torgersen1972a,Torgersen1972b,Torgersen1984} (Statistical Research Reports of the
University of Oslo, public archive copies, read here in their OCR text and page images) treat two experiments
$E=(P_\theta)$, $F=(Q_\theta)$ differentiable in variation norm at an interior parameter point $\theta^0$, restricted to
the $\varepsilon$-neighbourhood $E_\varepsilon$, $F_\varepsilon$ of $\theta^0$. For unrestricted decision problems put
\[
 D=\min\Bigl\{\max_i\|\dot P_{\theta^0,i}M-\dot Q_{\theta^0,i}\|_{\mathrm{var}}:\ P_{\theta^0}M=Q_{\theta^0}\Bigr\},
\]
where $M$ ranges over transitions. With the full-variation deficiency $\delta_{\mathrm{var}}$, the 1972 reports write
$\delta_{\mathrm{var}}(E_\varepsilon,F_\varepsilon)/(2\varepsilon)\to\dot\delta_{1972}=D/2$
(\cite[Thm.\ 6.1]{Torgersen1972a} for a one-dimensional parameter; \cite[Thm.\ 6.6]{Torgersen1972b} for a finite-dimensional
one), whereas the 1984 report writes $\delta_{\mathrm{var}}(E_\varepsilon,F_\varepsilon)/\varepsilon\to\dot\delta_{1984}=D$
(\cite[Thm.\ 3.1, Remarks 1--2]{Torgersen1984}, Section~2 of which gives the support-functional and randomisation criteria
for families of finite measures). Thus in this paper's half-total-variation convention the coefficient is $D/2$ in either
notation. The finite-dimensional result uses an $\ell^1$ neighbourhood, permits its $2m$ vertices instead, and also covers
restricted $k$-action problems, with convergence uniform in $k$. The weak-score criterion also has a direct classical
antecedent: \cite[Thm.\ 5.1(vi), p.\ 5.1]{Torgersen1972b} combines weak convergence of the canonical score laws with
convergence of their first absolute moments, with uniform integrability used explicitly in the proof.
In the local-comparison theorems every local hypothesis merges at one
base law in both experiments; in Theorem~\ref{thm:tangent} the model allows distinct fixed coarse laws
$Q_\vartheta=(a_{\vartheta,y})_y$, while the conditional fiber laws merge.
Proposition~\ref{prop:abstract} shows that this informative-base case follows from the same classical ingredients,
so we do not present the face program as a new general local-comparison principle; the distinct coarse laws are a
difference in formulation, not by themselves a novelty. What we regard as the contribution is the identification of the
scores of the growing shuffle models (Sections~\ref{sec:bundle} and \ref{sec:rare}), the explicit coefficients, the
family-versus-pairwise comparison, and the certificates of Theorem~\ref{thm:three}. The published chapter
\cite{Torgersen1985}, of which \cite{Torgersen1984} is the exact-title predecessor, and \S9.4 of \cite{Torgersen1991} were
not read, and \cite{Torgersen1976} is known to us from indexed excerpts only; whether the informative-base formulation of
Proposition~\ref{prop:abstract} appears there is not determined, and no exhaustive priority search has been made for the
model-specific results. Matsuda \cite[Thm.\ 2]{Matsuda1984} is further classical background in the Skovgaard--Michel direction.
Under its regularity assumptions, a sequence of asymptotic maximum-likelihood estimators of order $O(n^{-r/2})$, together with
the first $k$ log-likelihood derivatives evaluated at it, is sufficient for an approximating family at total-variation
distance $O(n^{-\min\{k,r\}/2})$, uniformly on compacts. These
are approximate-sufficiency statements for smooth models; they do not contain an optimised coefficient for a non-smooth
coarse component.
\item \textbf{Expansions.}  Hoeffding \cite{Hoeffding1948}, Takeuchi--Takemura \cite{TT87}, Diaconis--Griffiths \cite{DG14},
and Roos \cite{Roos2000,Roos2002} supply the projection and Krawtchouk background of Section~\ref{sec:bundle}.
Gaussian-polynomial limits of symmetric statistics are classical; Rubin--Vitale \cite{RV80} describe limits through
products of Hermite polynomials. Lemma~\ref{lem:coupling} supplies the explicit common-Gaussian representation and
quantitative $L^2$ remainder used here. For independent categorical summands, Daskalakis--Kamath--Tzamos
\cite[App.\ C.2, Lemmas 22--23]{DKT15} use matching moment profiles to obtain a common truncated approximator and
total-variation control. These approximation results do not optimise a common reconstruction kernel for a retained
statistic. Sreekumar--Goldfeld--Kato \cite{SGK22} give limit distributions of $f$-divergence estimators; the
hinge-type expectations $\E(H)_+$ here are elementary.
\item \textbf{Features and additional observations.}  Classical studies of statistical information in additional
observations include Le Cam \cite{LeCam1974}, Helgeland \cite{Helgeland1982}, and Mammen \cite{Mammen1986}.
For Le Cam's 1974 paper and Mammen's paper only bibliographic records were checked here, so no theorem-level comparison
with them is asserted. van Rooyen--Williamson \cite[\S\S4--5]{VRW14} formulate
finite-experiment deficiency as linear optimisation over a common decoder and relate generic feature quality to instance
reconstruction under a fixed instance prior. Their generic-feature criterion quantifies over decision problems; here a
fixed hypothesis family is given and the leading law-reconstruction loss is computed.
The preservation restriction has a classical precedent in Helgeland's distinction between free deficiency and
insufficiency based on approximation preserving the observed subalgebra \cite[pp.\ 78, 80]{Helgeland1982}; their orders
can differ. Proposition~\ref{prop:fp} evaluates the specific tag-preserving restriction in the present model.
Sample amplification uses the same minimax total-variation criterion: Axelrod et al.\ \cite[\S3]{AxelrodEtAl2024}
compare $n$ iid observations with $n+m$ observations, obtain an exact Gaussian answer, and bound shuffled replacements
by $\chi^2$ (Theorem~6.2 and Lemma~5.8 of their arXiv version). Those experiments differ from our retained-statistic
reduction and its finite composition window.
\item \textbf{Permutation mixtures.}  Han--Niles-Weed \cite[Thm.\ 1.5]{HanNilesWeed2025} bound the distance between
neighbouring permutation mixtures by $O(n^{-1/2})$ in total variation under a finite $\chi^2$-diameter assumption, allowing
heterogeneous independent background components. They compare mixture laws directly rather than optimise reconstruction
from a specified retained statistic. The raw rare channel here has $W_0(c)=0<W_{1,n}(c)$, so that finite-diameter
assumption cannot be invoked directly for its two rows.
\item \textbf{Anonymous aggregation.}  Ghazi--Manurangsi--Pagh--Velingker \cite[\S1.4, \S4.2]{GMPV20} build a lower bound from
the first $t$ at which the $t$-message marginals of two encodings differ, with non-sharp constants and no deficiency or
kernels; the ``first differing moment'' idea is theirs in spirit.
\item \textbf{The author's earlier papers.}  The full-support Gaussian study \cite{Shvets2602} identifies the
fixed-composition covariance and Fisher constant, controls the likelihood-ratio residual uniformly over interior
compositions in $L^2$ and $L^4$, and derives divergence and two-sided local privacy expansions. Its canonical
fixed-message unbundled specialisation already gives the product-kernel $U$-statistic likelihood ratio and its Hoeffding
decomposition. We re-derive the needed replacement algebra here and treat arbitrary correlated bundle laws, higher first
differing factorial-moment degrees, and a finite family reconstructed by one common kernel.
The critical-regime study \cite{Shvets2603count} addresses a different operational problem: minimising worst-case squared
count error over common single-message channels subject to a global shuffled privacy constraint. It retains the observed
rare marks, gives the limiting count-estimation risk of each fixed finite marked configuration, and characterises the
all-channel optimal value by an infimum over such configurations. Here the channel and retained statistic are fixed, and
we evaluate the optimal total-variation reconstruction loss. The rare model considered here need not satisfy global
$(\varepsilon,\delta)$ privacy at any fixed $\delta<1$: the event $N(c)>0$ has probabilities $0$ and $1-\lambda/n$ under the
all-zero dataset and its one-one neighbour. Thus our reconstruction coefficients do not provide a global privacy guarantee
or a count-MSE improvement over \cite{Shvets2603count}. All model-specific estimates used here are proved within this paper.
\end{itemize}

\begin{center}\small
\begin{tabular}{p{0.26\textwidth}p{0.68\textwidth}}
\toprule
Status & Content\\
\midrule
Classical & finite-family identity and support-functional criterion (Torgersen); first-order local deficiency
coefficients as deficiencies of derivative families at one base law (Torgersen 1972, 1984); exposed-face derivative
(Danskin); Hoeffding/Krawtchouk expansions; minimax theorem; binary knot formula.\\
Proved here & Theorem~\ref{thm:tangent} (tangent program, with the $L^2$ coupling Lemma~\ref{lem:coupling}) and its abstract
form Proposition~\ref{prop:abstract} (informative coarse base, fixed finite tag set); Proposition~\ref{prop:fp}
(exact fiber-preserving value, strict gap criterion); collinear LP form; Corollaries~\ref{cor:binary}--\ref{cor:single};
Proposition~\ref{prop:fourthirds} (exact $4/3$, rational certificate) and Corollary~\ref{cor:finiteidentity}
(finite-sample identity and explicit optimal kernel); Theorem~\ref{thm:rare} (window, drift-insensitive);
Theorem~\ref{thm:three} (simplex lower bound, $\sigma>\max\{p_1,p_2\}\iff\mu<1$, $\mathcal V_2=\sigma$ on $(0,1]$ by two
explicit certificates, $\mathcal V_2=p_2$ on $[1,\infty)$ by the chain certificate with the prior bound
\eqref{eq:priorbound}, the full phase formula \eqref{eq:phase}, ratio formula and its global maximum $1.3228$);
Corollary~\ref{cor:priors} (optimal priors exclude the middle composition above $\mu=1$, use all three below);
Theorems~\ref{thm:small-general} and \ref{thm:outer-general} (exact full-chain endpoint and explicit endpoint-pair regime for every window); Corollary~\ref{cor:generalgapphase} and Propositions~\ref{prop:gap-two-lower}, \ref{prop:gap-unit} and \ref{prop:path-envelope} (family--pair comparisons); Proposition~\ref{prop:path-dual-knots} (maximizing-path cut bounds and necessary knot identities); Proposition~\ref{prop:unconditional-comparison-bounds} (unconditional bounds, a factor-two comparison on proved parameter regimes, and a limiting ratio of two at $\mu=\sqrt2$).\\
Exact computation & Theorem~\ref{thm:four-compositions}: all real $\mu>0$, with exact polynomial certificate checks on finite parameter cells in \texttt{anc/m3\_low/} and \texttt{anc/m3\_high/}, plus analytic tails and the analytic large-mean proof. The 1120 rational parameter-point certificates in \texttt{anc/grid/} establish Hypothesis~P on the stated grid only.\\
Open & Conjecture~\ref{conj:paths} in intermediate regimes for $m\ge4$; the upper bound $\mathcal V_m\le2P_m$ outside the proved comparison regimes; the priority of the informative-base formulation against \cite{Torgersen1985,Torgersen1991}.\\
Numerical only & finite-$N$ convergence tables; non-collinear example; observed off-knot choices of the solver's dual rates for $\mu>1$.\\
\bottomrule
\end{tabular}
\end{center}

\section*{Use of AI tools}
The author posed the problem, directed the research, and decided which results to pursue and include. ChatGPT (OpenAI) assisted in developing proof strategies and dual certificates, writing and running numerical and exact-arithmetic verification scripts, reviewing primary sources, and revising the manuscript. In the general-window extension, ChatGPT assisted in deriving the full-chain and endpoint-pair certificates and unconditional comparison bounds, analysing path-knot constraints, developing the four-composition exact verifiers, and drafting the statements and proofs. Claude (Anthropic) assisted in developing and drafting proofs, writing and running verification scripts, comparing the results with prior work, and preparing the manuscript. Separate sessions of both tools were used for independent audits of proofs, certificates, computations, and references. The author checked all results and takes full responsibility for the content of this paper.

\end{document}